\documentclass[onecolumn]{article}
\usepackage{hyperref, graphicx}
\usepackage{amsmath}
\usepackage{amssymb}
\usepackage{slashed} 
\usepackage{yfonts}
\usepackage[T1]{fontenc}
\usepackage[utf8]{inputenc}
\usepackage[english]{babel}
\usepackage{amsthm}
\usepackage{tikz-cd}
\usepackage{mathabx}
\usepackage{stmaryrd}
\usepackage{braket}
\usepackage{physics}
\usepackage{geometry}
\usepackage[strings]{underscore}
\usepackage[
backend=biber,
style=numeric-comp,
sorting=none,
]{biblatex}
\theoremstyle{theorem}
\newtheorem{theorem}{Theorem}[section]
\newtheorem{lemma}{Lemma}[section]
\newtheorem{corollary}{Corollary}[section]
\newtheorem{proposition}{Proposition}[section]

\usepackage{authblk}
    \allowdisplaybreaks

\hypersetup{
pdfstartview = {FitH},
}
\hypersetup{
	colorlinks=true,         
	linkcolor=purple,          
	citecolor=red,        
	urlcolor=blue            
}

\DeclareMathAlphabet{\mathpzc}{OT1}{pzc}{m}{it}

\theoremstyle{remark}
\newtheorem{remark}{Remark}[section]

\theoremstyle{example}

\theoremstyle{definition}
\newtheorem{definition}{Definition}[section]

\usepackage{mathtools}

\usepackage{mathrsfs}
\numberwithin{equation}{section}

\newcommand{\beq}{\begin{eqnarray}}
\newcommand{\eeq}{\end{eqnarray}}

\def\ker{\operatorname{ker}} 

\def\s{\mathfrak s}

\def\g{\mathfrak g}

\def\h{\mathfrak h}

\def\bbZ{\mathbb{Z}}
\def\R{\mathbb{R}}

\def\demi{{\frac{1}{2}}}

\newcommand{\cF}{{\cal F}}

\newcommand{\cG}{{\cal G}}

\newcommand{\cH}{{\cal H}}

\begin{document}

\title{Gauging the Gauge and Anomaly Resolution}
\author[1,2]{{ \sf Hank Chen}\thanks{hank.chen@uwaterloo.ca}}
% \affiliation{Department of Applied Mathematics, University of Waterloo, Waterloo, Ontario, Canada}
% \author[1]{\sffamily Florian Girelli\thanks{florian.girelli@uwaterloo.ca}}
% \author[1]{\sffamily Panagiotis Tsimiklis\thanks{ptsimiklis@uwaterloo.ca}}
\author[1]{{\sf Florian Girelli}\thanks{fgirelli@uwaterloo.ca}}
\affil[1]{\small Department of Applied Mathematics, University of Waterloo, 200 University Avenue West, Waterloo, Ontario, Canada, N2L 3G1}
\affil[2]{\small Beijing Institute of Mathematical Sciences and Applications, Beijing 101408, China}

\maketitle

\begin{abstract}
    In this paper, we explore the algebraic and geometric structures that arise from a procedure we dub "gauging the gauge", which involves the promotion of a certain global, coordinate independent symmetry to a local one. By gauging the global 1-form shift symmetry in a gauge theory, we demonstrate that the structure of a {Lie algebra crossed-module} and its associated 2-gauge theory arises. Moreover, performing this procedure once again on a 2-gauge theory generates a 3-gauge theory, based on {Lie algebra 2-crossed-modules}. As such, we show that the physical procedure of "gauging the gauge" can be understood mathematically as a \textit{categorification}. Applications of such higher-gauge structures are considered, including general relativity, high-energy physics and condensed matter theory. Of particular interest is the mechanism of \textit{anomaly resolution}, in which one introduces a higher-gauge structure to absorb curvature defects. This mechanism has been shown to allow one to consistently gauge an anomalous background symmetry in QFT 
    %\cite{Cordova:2018cvg}
    . 
\end{abstract}
% Keywords
% \keyword{gauge theories; quantum field theories; higher symmetries; anomalies} 

\tableofcontents
\setcounter{section}{-1}

\section{Introduction}
{Understanding}  the symmetries of a physical system is tantamount to constructing a mathematical model that describes its properties. From condensed matter to high-energy physics, models such as Ginzburg--Landau theory of phase transitions and the Standard Model, are based upon this fact. In fact, sometimes a handle on symmetry is all we have in more esoteric areas of theoretical physics, such as string theory and quantum gravity, as no sufficient experimental data and/or techniques are yet available.

In general, there are two notions of symmetry in a physical system: { {global}} and { {local}}. Global symmetries act uniformly upon the entire system, while the action of the local symmetry depends on the particular physical configuration in $X$. Of course, the latter is the more general notion, and a global symmetry can be made local by the procedure of { {gauging}}: we promote the constant group elements to depend on points in $X$. Because of the locality requirement for the symmetry, the notion of derivative is typically not transforming covariantly unless one adds some compensating term, a connection. In a sense, the notion of connection encodes the break down of the covariance of the transformation of the derivative.

The structure of physical systems with  local symmetries is described geometrically by a { {principal bundle}} $P\rightarrow X$ \cite{book-baez}; solutions to the equations of motion determine configurations on $P$ that represent a class of { {physical}} configurations on $X$ under the action of transformations by $G$. Many local geometric quantities on $P$, such as connection and curvature, are important for understanding the physical system.

Moreover, recent developments in condensed matter theory \cite{book-wen,Halperin:1978,Bilal:2008} had also shown that, aside from symmetry, topology plays a central role in the characterization of physical systems as well. In particular, the topology of the principal bundle $P\rightarrow X$ determines the presence of { {defects}} in the physical system, which can alter its physical properties in a drastic, non-perturbative manner \cite{Mermin:1979,Freed:2014}. Also from a quantum gravity point of view, understanding topology can be relevant. Indeed, one might need to sum over all the possible topological configurations to build the transition amplitudes. In fact, for example, 3D gravity is essentially about the topology of spacetime, since it does not have local degrees of freedom.

It is important to note that there are different notions of ``topological features''. We can characterize the topology of the manifold of interest (such as spacetime), the topology of the principal bundle in the case of a gauge theory, which is related to the cohomological features of the gauge symmetry. The notion of magnetic monopole is related to the topology and geometry of the principal bundle. Such topological features in the gauge symmetry structure appear for example in quantum field theory, where they manifest as obstructions to defining the path integral \cite{Bilal:2008,adler2004anomalies}. 
Typically, all such non-trivial topological features can be cast as {{anomalies}}, as due to their presence, some quantities (such as $d_AF$, or the 1- or 2-curvature) are non-trivial. 

%\medskip 

Given a physical system, one typically tries to identify the symmetries in order to characterize the physical conserved quantities or the quantum theory (built as representations of this symmetry). Conversely, one can   construct new symmetry structures and then try to identify some physical systems associated to them. One way to proceed is to {{categorify}} the notion of symmetries. It consists in using category theory tools to build new types of symmetries \cite{Baez:2003fs}. For example the notion of 2-group comes up naturally when considering higher-dimensional homotopy types \cite{Brown,Ang2018}. This approach relies on beautiful elaborate mathematics which might leave the interested reader wondering what are the physical motivations behind it. One goal of this paper is to argue that this categorification process can be seen as a gauging principle so that dealing with categorified local symmetries here is nothing else than gauging (a gauge). Hence a $(n+1)$-gauge theory can be interpreted as a gauged $n$-gauge theory. This idea is  not new to the expert in categorical symmetries, but we try to make this point in a pedagogical way, using mostly basic tools of gauge theory and exploring how and when there could be some generalizations.    

%\medskip 

Once we have in hand some new symmetries, one needs to identify systems where such symmetries are realized. It has been known since some time that categorified symmetries are the natural structure to probe the notion of topology, thanks to the categorical ladder \cite{Crane:1994fw,Mackaay:ek}. For example the moduli space of flat connection can be used to characterize topological features of a 3D manifold \cite{Fock:1998nu, Fock:1998nu}. One can expect the moduli space of flat 2-connections (i.e., with a categorified gauge symmetry) to probe some topological features of a 4D manifold \cite{Alvarez:2021, Cui_2017, Douglas:2018}. These categorified symmetries have been explored from the physics perspective, mainly in the context of condensed matter \cite{Kapustin:2013uxa,Zhu:2019,Kong:2020} or string theory~\cite{Carey:2005,Sati:2009ic}. 

%\medskip

Now why do these categorified symmetries take the form they do, and why are they so useful?  In this paper, we seek to answer this question through a procedure of ``gauging-the-gauge'', which highlights  the fact that these higher symmetries reproduce some of the feature of the topological anomalies we discussed above. For example, it is well-known that a magnetic monopole corresponds to a singular structure in the principal bundle, namely the Bianchi identity is violated, due to the fact that $d^2\neq 0$. Recent work has pointed out that such magnetic monopole physics can be recovered using a 2-gauge theory, without a violation of the Bianchi identity \cite{Cordova:2018cvg,Benini_2019,Dubinkin:2020kxo}. Furthermore, we show that from the purely field theoretic requirement of covariance, we can completely recover the structure of higher-gauge theory. This procedure in fact stack: we will demonstrate that, by absorbing a violation of the 2-Bianchi identity, we recover a non-singular 3-gauge theory.

\subsection{Organization of the Paper}

The article is organized as follows. 

In Section \ref{sec:gauge0}, we recall quickly  what is the usual notion of gauging. Then in Section \ref{sec:gau1gau}, we proceed to gauge the gauge by making local, in some sense, a  hidden shift symmetry for the 1-curvature. The lack of covariance is encoded in a 2-connection. To have a general structure, the notion of a Lie group crossed-module, equivalent to a strict Lie 2-group, is naturally introduced. We then discuss the different possible generalizations of the gauge symmetry structure. In particular, we describe how the structures of a weak Lie 2-algebra \cite{Chen:2012gz,Baez:2003fs} manifests when the (1-)Bianchi identity is relaxed. We also point out how  having a specific non-zero 2-curvature can be related to the topological properties of the corresponding Lie algebra crossed-module, encoded by the so-called differential {{Postnikov class}} \cite{Baez:2005sn,Kapustin:2013uxa}. We discuss then how the gauge symmetry structure needs to be adapted to account for such case.

%\medskip 

In Section \ref{sec:2app}, we discuss some applications of  2-gauge theory. The main example is a topological theory, 2-BF theory \cite{Girelli:2007tt, Martins:2010ry}, which can be seen in some sense as a sort of BF theory. According to the space-time dimension, there are different types of applications. For example, in 3D, BF theory (i.e., in particular 3D gravity) can be seen as a specific (2-)gauge fixed 2-BF theory. This had not been pinpointed before, to the best of our knowledge, and it could provide some interesting  new directions to explore 3D gravity. In 4D, 2-BF theory is somehow the canonical topological theory. Its symmetry structure is associated to a 2-Drinfeld double and one would expect that topological  excitations (string like or point like) should be encoded in terms of representations of 2-Drinfeld double, a direct generalization of the Drinfeld double in the 3D case. In 5D, it was known that 2-gauge symmetries could be related to the $w_2$ and $w_3$ Stiefel--Whitney classes \cite{JuvenWang,Thorngren2015}. We highlight how this can be done through  a 2-BF theory, which was not emphasized previously to the best of our knowledge. 

We also discuss how the magnetic monopole can be recovered from considering a 2-gauge theory, with some interesting mobility conditions on the currents \cite{Dubinkin:2020kxo}. This is illustrating the notion of {{anomaly resolution}}, which exchanges a system with non-trivial topological feature (in this case the non-trivial topology for the principal bundle) for a system with an extra gauge symmetry structure which allows to reproduce the physics of the non-trivial topological feature \cite{Cordova:2018cvg,Benini_2019}. A closely related notion of { {anomaly inflow}} \cite{Witten:2019}, and its relation to higher symmetries \cite{Thorngren:2015gtw}, has also been studied recently in the high-energy physics literature.

In Section \ref{sec:gaug2gau}, we gauge the 2-gauge, to obtain a 3-gauge theory. The gauging follows the same step as in the 2-gauge case. We introduce some more general  shift transformation, which do not leave the  derivative (of the 2-connection) covariant. The lack of covariance is encoded in a 3-connection.   We highlight where some constraints, such as the 1-Bianchi identity and the Peiffer condition, can be weakened. We construct the 3-gauge transformations given in terms of a ``Lie algebra 2-crossed-module''. As a direct generalization of the Lie algebra crossed-module case, we discuss how the presence of a non-trivial (first) differential Postnikov class for the 2-crossed-module is associated with a non-zero 3-curvature and non-trivial 1-gauge transformation. 

In Section \ref{sec:3app}, we discuss some applications when dealing with a 3-gauge theory. We recall some of the main features of a 3-BF theory, following \cite{Radenkovic:2019qme}. Extending this, we point out that a 5D spacetime hosting a 3-BF theory (resp. more generally a ``3-Chern--Simons theory'') would have a 3-Drinfeld double (resp. ``Hopf 3-algebra'') symmetry at play. We then  study the 3-Yang--Mills theory $S_\text{3YM}$ based on Lie 3-group gauge principle. The associated 3-conservation laws yield new and interesting higher-mobility constraints that have not appeared previously, as far as we know.

%\medskip

In Appendix \ref{algxmod}, we reproduce in full detail the classification of strict Lie 2-algebras/Lie algebra crossed-modules \cite{Wag:2006}. In Appendix \ref{weakpost}, we explore the relationship between weak 2-gauge theories based on weak Lie 2-algebras \cite{Kim:2019owc,Baez:2005sn} and the Postnikov anomaly. In Appendix \ref{framedquasi}, we examine the topology of framed bordisms/manifolds and their relevance to topological phases (in particular the 5D $w_2w_3$ gravitational anomaly).

\subsection{Original Contributions}
Aside from the demonstration that the \textit{{gauging-the-gauge}} procedure provides an {ab ovo} %MDPI: We removed the bold. Please confirm this revision. %Hank: please italicize this phrase; it is worth emphasizing
 field theoretic explanation of higher-gauge principles and anomaly cancellation, other original contents of this paper are described in the following.

\begin{itemize}
    \item In Section \ref{strpt} we show that the anomalous 4D boundary phase of the 5D topological order $w_2w_3$ \textit{{cannot}} be described as a discrete 2-BF theory. We propose a resolution of this problem by twisting the \textit{{dual}} gauge sector, effectively implementing a \textit{{non-double}}/non-BF field theoretic description of the 4D phase.
    \item In Section \ref{loopgauge}, we provide a higher-gauge principle suitable for the \textit{{weak}} string 2-algebra, which manifests the appropriate Dixmier-Douady class for the string 2-group as a 2-curvature anomaly. We then applied the gauging-the-gauge procedure to resolve the string structure obstruction through a 3-gauge theory. 
\end{itemize}

{The} ``loop model'' of the \textit{{string 2-algebra}} is a crossed-module equivalent to the weak string 2-algebra defined in \cite{Baez:2005sn}. This resolution of the string structure obstruction is an example of a generalization of the symmetry-enrichment procedure, and generalizes the 2-group Green--Schwarz mechanism in \cite{Benini_2019,Cordova:2018cvg}.

\subsubsection*{Acknowledgements}
Part of this work was completed when HC was affiliated with the Beijing Institute of Mathematical Sciences and Applications. HC is supported by the National Foundation of Science of China (grant number: W2533012).

\section{Gauging the 0-Gauge}\label{sec:gauge0}
In the following, we first review in a pedestrian way the notion of gauging a global symmetry. %, and see how the structure of a (flat) principal $G$-bundle arises. 
This is standard material, for which one can find many introductions (e.g., \cite{book-baez}). 

Let $X$ denote a $d$-dimensional smooth manifold admitting an action by a Lie group $G$. Consider a (smooth) function $\phi\in C^{\infty}(X)$ transforming under a representation $\pi:G\rightarrow \operatorname{GL}(V)$ of the group $G$ for some vector space $V$, that is  $\phi\in C^{\infty}(X)\otimes V$, namely $\phi$ lies in the algebra of $V$-valued smooth functions on $X$. 

Note that $\pi$ is an homomorphism, and the field $\phi$ transforms as %$V\subset C^\infty(X,\mathbb{R})$ of smooth functions $\phi:X\rightarrow\mathbb{R}$ transforming under $G$ as
\begin{equation}
    \phi(x)\rightarrow \pi (g) \phi(x),\qquad g\in G.\nonumber
\end{equation}
 {If} %MDPI: Please check through the paper if indentation should be added to the first line after equations. The following highlights are the same. %Hank: please never indent after equations, unless a new line was already opened in the previous draft.
 $g\in G$ is not a $G$-valued function of $X$, then the derivative $d\phi$ transforms covariantly,
 \begin{equation}
     d\phi \rightarrow d(\pi(g)\phi)=\pi(g) d\phi,\nonumber
 \end{equation}
  %as an element of the cotangent bundle $T^*X$, 
and  $G$ encodes a  (global) \textit{{0-gauge symmetry}}. %, or equivalently a "$(-1)$-form symmetry".

We can promote $g$ to be a $G$-valued function of $X$ itself, such that we still have the transformation law
\begin{equation}
    \phi(x)\rightarrow \pi({g(x)})\phi (x) \equiv g(x)\cdot \phi(x) =\phi'.\nonumber
\end{equation}
{In} this case we are dealing with a principal bundle with fiber $G$ and base $X$. {The representation $\pi$ of $G$ induces a representation of its Lie algebra $\operatorname{Lie}G=\mathfrak{g}$. For notational simplicity, we will no longer indicate $\pi$.}%MDPI: The footnote is not allowed. We moved it in the maintext, please confirm. The following highlights are the same. %Hank: then please move it to the end of the previous pararaph

The Leibniz rule for the exterior derivative $d$ dictates that 
\begin{equation}\label{d-cov}
    d\phi \rightarrow g(d+g^{-1}dg)\cdot \phi.\nonumber 
\end{equation}
{As} such it is not $d\phi$ that transforms covariantly, but the covariant derivative  $\nabla\phi\equiv (d+ g^{-1} dg)\phi$. Indeed, we can introduce the connection $A=g^{-1}dg\in \Omega^1(X)\otimes\mathfrak{g}$, to compensate for the lack of {covariance,} %MDPI: We revised the equations numbers, please confirm %Hank: confirm
  \begin{equation}
    gA \phi = d \phi'- gd\phi  \rightarrow A = g^{-1} d g. \label{lack}
\end{equation}
{Notice} that this connection has a natural invariance symmetry under the left translation for all $h\in G$ constant (i.e., $dh=0$).  
\begin{equation}
(hg)^{-1} d(hg ) = g ^{-1}dg\label{global0form}
\end{equation}
{This} is the well-known fact that  this is a left-invariant form.

%The term $A=gdg^{-1}\in \Omega^1(X)\otimes\mathfrak{g}$ is nothing else than a pure-gauge connection specified by $g$. 

Given the covariant derivative $\nabla = d+g^{-1}dg$, its associated curvature
\begin{equation}
    \operatorname{cur}\nabla  = [\nabla,\nabla] = d(g^{-1}dg) + (g^{-1}dg)\wedge (g^{-1}dg) = 0\nonumber% (dg^{-1})\wedge dg - dg^{-1}\wedge dg=
\end{equation}
vanishes, where we have used the identity $d(1) = d(g^{-1}g) = (dg^{-1})g+g^{-1}dg = 0$. This means that the connection $A=g^{-1}dg$ is \textit{{flat}}.

\subsection{The 0-Form Symmetry and 1-Gauge Transformations} The connection 1-form in an arbitrary gauge, $A\in\Omega^1(X)\otimes\mathfrak{g}$ and the associated curvature 2-form $\operatorname{cur}A=F=d_AA=dA+\frac{1}{2}[A\wedge A]$ transform as
\begin{equation}\label{1-gauge}
    A\rightarrow A^g=g^{-1}Ag + g^{-1}dg,\qquad F\rightarrow F^g=g^{-1}Fg.
\end{equation}
% Functions $f(\operatorname{tr}F^k)$ of the trace $\operatorname{tr}F^k$ for any $k\geq 0$ are gauge invariant. In other words, considered as a function $f=f(A)$ of the connection 1-form $A$ itself, we have the gauge redundancy
% \begin{equation}
%     f(A) = f(A^g),\qquad \forall~g\in\operatorname{Map}(X,G).\nonumber
% \end{equation}
% so $f$ is a function on the moduli space $\operatorname{Conn}^GX/\operatorname{Map}(X,G)$.
{Expressing} $g = \exp \lambda \approx 1+\lambda$ in terms of the infinitesimal gauge parameter $\lambda\in\Omega^0(X)\otimes\mathfrak{g}$, we achieve the (infinitesimal) {(1-)gauge transformation laws}
\begin{eqnarray}
A&\rightarrow& A^\lambda = A + [A,\lambda] + d\lambda \equiv A + d_A\lambda, \nonumber\\
F&\rightarrow& F^\lambda = F + [F,\lambda].\nonumber
\end{eqnarray} 
{They} endow the bundle $P\rightarrow X$ with a \textit{{0-form gauge symmetry}} parameterized by $\lambda$. 

%\medskip 

The {\it {Bianchi identity}} reads $d_AF = dF + [A\wedge F]=0$, which holds in general for any principal $G$-bundle with connection $A$. Since $F$ transforms covariantly, $d_AF$ also transforms  covariantly
\begin{equation}
   d_AF\rightarrow d_{A^\lambda}F^\lambda = d_AF + [d_AF,\lambda].\nonumber
\end{equation} 
{It} is possible (and consistent) to achieve a {\it {1-curvature anomaly}} $F=\sigma\neq0$, as long as $\sigma\in \Omega^2(X)\otimes\mathfrak{g}$ satisfies $d_A\sigma=0$, and transforms covariantly $\sigma\rightarrow g^{-1}\sigma g$. 

%\medskip

\subsection{Global 1-Form Symmetry} 
What we have recalled here is that, by gauging the global symmetry understood as a ``0-gauge'' symmetry, we obtain an ordinary 1-gauge bundle $P\rightarrow X$ that is flat. However, one may notice that the curvature 2-form $F=d_AA$ has a hidden symmetry in the presence of a non-trivial centre $Z(\mathfrak{g})$. This symmetry is  given by
\begin{equation}
    A\rightarrow A+\alpha,\label{global1form}
\end{equation}
where $\alpha$ is a closed 1-form valued in the centre  $Z(\mathfrak{g})$ 
%= \mathfrak{g}/[\mathfrak{g},\mathfrak{g}]$ 
of the Lie algebra $\g$, that is $\alpha\in \Omega^1_0(X)\otimes Z(\mathfrak{g})$. As such the above gauge structure in fact manifests a ``1-form  symmetry'' parameterized by $\alpha$, on top of the pre-existing 1-gauge 0-form symmetry parameterized by $\lambda$. This  1-form symmetry is affecting the connection $A$ but not its curvature.

\section{Gauging the 1-Gauge}\label{sec:gau1gau}
In the 1-gauge case, we have highlighted two different types of invariance, one specified by a left multiplication, in \eqref{global0form}, the other one by a 1-form shift in \eqref{global1form}. It is natural to ask what happens when we gauge each symmetry, i.e., we make them non-constant. For the former, making $h$ non-constant amounts to just another gauge transformation, so there is nothing new to be gained. The latter is more interesting, as it leads to some new structures.   

Relaxing the condition that $\alpha$ in  \eqref{global1form} is constant and valued in the centre $Z(\mathfrak{g})$ will be called ``gauging the 1-form gauge''. So we allow $\alpha\rightarrow a$ to become a generic 1-form $a\in\Omega^1(X)\otimes\mathfrak{g}$ that has non-trivial coordinate dependence on $X$, similar to the gauging procedure for the global/0-gauge symmetry.

\subsection{Shifting the Connection}
Typically, one may \textit{{a priori}} take a gauge bundle $P\rightarrow X$ with the non-trivial curvature $F=\sigma\neq 0$, then study the associated gauge theory. Alternatively, we may perform a {\it {particular}} 1-form shift such that $F \rightarrow F'$ is transformed to a non-trivial value. 

% So far, we started from the flat connection and the 1-gauge transformation left the connection flat.  

Indeed, under a generic 1-form shift.   
\begin{equation}
A\rightarrow A'=A+a,    \nonumber
\end{equation}
we see that the curvature transforms accordingly as
\begin{equation}\label{arbshift}
    F \rightarrow F'=d_{A'}A'=F +d_Aa+\frac{1}{2}[a\wedge a] =F+ d_Aa+\frac{1}{2}[a\wedge a].
\end{equation}
{In} the gauge where $A=0$, we just have 
\begin{equation}
    F'=da+\frac{1}{2}[a\wedge a],\nonumber
\end{equation}
which is the curvature of $a$ considered as a $G$-connection. As such we may shift the curvature to any value from zero, which serves as the central key fact for anomaly resolution discussed later. Usually, the ``gauging'' story ends here, and we deal with an arbitrary curvature associated to the connection in a particular 1-form gauge $A=a$. 

However, the above also shows that, by considering the 1-form shift as a higher-form gauge symmetry, the (1-)curvature quantity $F$ is a {\it {gauge datum}}, the notion of curvature is gauge dependent. We have then a pair of gauge structures, one encoded in $g$ which in a sense encodes the arbitrariness of the frame we deal with, and one encoded in $a$, which encodes the arbitrariness of the curvature. 

%In other words, the 1-form shift symmetry generates a source for the 1-curvature $F$, and this latter curvature source and the 1-form connection are treated on a different footing.

%\medskip 

One can realize that the transformation  \eqref{arbshift} can be seen as lack of covariance of the curvature 2-form under the arbitrary shift, analogous to the one of the derivative of the field $\phi$ under $\pi(g)$. To amend for the lack of covariance, we introduced a non-zero connection $A=gdg^{-1}$ in \eqref{lack}. 

Hence in a similar manner, to amend for the lack of covariance of the curvature under the arbitrary shift, we  introduce  a {\it {2-form}} gauge connection $\Sigma\in\Omega^2(X)\otimes\mathfrak{g}$ such that, in the gauge where $A=0$
\begin{equation}
    \Sigma\equiv (F'-F)= F'=da+\frac{1}{2}[a\wedge a] .\label{eq:2gau}
\end{equation}
{If} we define the curvature of $\Sigma$, as the 2-curvature, \begin{equation}
K=d_A\Sigma,     \nonumber
\end{equation}
then we see that by the Bianchi identity 
\begin{equation}
   d_{A}\Sigma= d_{A}F=0,  \nonumber
\end{equation}
so that this 2-connection is flat. Indeed as we shall see later, this 2-connection $\Sigma = da+\frac{1}{2}[a\wedge a]$ is a ``pure 2-gauge'', analogous to the flat pure 1-gauge $A=gdg^{-1}$ obtained from gauging the 0-gauge. 

% \smallskip

% A similar construction happens for $\ast F$. 
% \begin{equation}
%     \ast \Sigma\equiv (\ast F'-\ast F)= \ast F'=\ast(da+\frac{1}{2}[a\wedge a]) .\label{eq:ast-2gau}
% \end{equation}
% Because of the Hodge operator, the analogue of the Bianchi identity does not have to hold.  

The construction so far is restrictive, in a sense since we focus on a 2-connection with value in the same Lie algebra  $\mathfrak{g}$. It seems natural to make it valued   in some other Lie algebra $\mathfrak{h}$, together with a map $t:\mathfrak{h}\rightarrow\mathfrak{g}$ (a homomorphism of Lie algebras), which plays in a sense the same role as the representation $\pi$ when we dealt with a regular 1-gauge. The most natural notion to use is that of a Lie 2-algebra \cite{Baez:2003fs}. There are different notions of  it. The first we are interested in is the notion of {\it {strict}} Lie 2-algebra, which can be equivalently viewed as a Lie algebra crossed-module \cite{Wag:2006}.  The crossed-module formulation is most convenient to discuss the notion of 2-gauge theory. We shall also see how the notion of a weak Lie 2-algebra can be relevant in this setting.

\subsection{Lie Algebra Crossed-Modules} 
We first define the notion of Lie algebra crossed-modules, and introduce the fields relevant to building a 2-gauge theory. We will then seek to develop all the structures of a principal 2-bundle (see e.g., \cite{Wockel2008Principal2A}) from field-theoretic considerations.

\begin{definition}
Consider a pair of Lie algebras $\mathfrak{g}$ and $\mathfrak{h}$,  such that there is an action of $\mathfrak{g}$ on $\mathfrak{h}$ noted $\rhd$. We also introduce $t$  a Lie algebra homomorphism  such that we have  2-term complex 
\begin{equation}
    \mathfrak{G}: \mathfrak{h}\xrightarrow{t}\mathfrak{g},\label{eq:2comp}
\end{equation}
{The} 2-term algebra complex  \eqref{eq:2comp} is a \textit{ Lie algebra crossed-module} if the action and the $t$-map satisfy 
\begin{enumerate}
    \item  the Peiffer conditions, respectively the \textit{$\g$-equivariance of $t$} and the \textit{Peiffer identity},
\begin{equation}
    t(x\rhd y) = [x,t(y)],\qquad t(y)\rhd y' = [y,y'],\nonumber
\end{equation}
and 
\item the \textit{2-Jacobi identities}
\begin{eqnarray}
    [[x,x'],x''] + [[x',x''],x] + [[x'',x],x'] =0,\nonumber \\
    x\rhd (x'\rhd y) - x'\rhd (x\rhd y) - [x,x']\rhd y = 0,\nonumber
\end{eqnarray}
\end{enumerate}
for each $x,x',x''\in\mathfrak{g}$ and $y,y'\in\mathfrak{h}$.
\end{definition}

 {An important} %MDPI: We removed the \noindent format, please confirm %Hank: confirm
consequence of the Peiffer identity is that $\ker t\subset Z(\mathfrak{h})$ is contained in the centre of $\mathfrak{h}$.

The $\g$-equivariance property of $t$ can be summarized by the following diagram
\begin{equation}
\begin{tikzcd}
\mathfrak{h} \arrow[r, "t"] \arrow[d]                                    & \mathfrak{g} \arrow[d] \arrow[ld, "\rhd"']\\
\operatorname{Inn}\mathfrak{h} \arrow[r, "t"] & \operatorname{Inn}\mathfrak{g} 
\end{tikzcd},\label{equivdiag}
\end{equation} 
where $\operatorname{Inn}\g$ denotes the space of inner derivations on the Lie algebra $\g$.

%\medskip 

Let us consider now the relevant connections: the 1-form connection $A$ is valued in $\g$, while the 2-form connection $\Sigma$ is valued in $\h$. As we will see in Section \ref{sec:fake-flat}, $t$ is a Lie algebra homomorphism that allows us to connect fields valued in $\h$ to ones valued in $\g$.  This action $\rhd$ can be viewed in a sense as the gauge transformations induced by $\mathfrak{g}$ on the fields/2-gauge parameters with value in $\mathfrak{h}$. This will be discussed in Section \ref{1-2gauge}. 

The covariant derivative we will use is still $d_A$, i.e., it is defined in terms of the 1-connection $A$. We will therefore use the action to define the covariant derivative of a form with value in $\h$.  Taking an arbitrary $\mathfrak{h}$-valued n-form  $S\in\Omega^n(X)\otimes\mathfrak{h}$, we introduce the wedge product $\wedge^\rhd$ between a 1-form and and n-form,  $$\wedge^\rhd: (\Omega^1(X)\otimes \mathfrak{g})\otimes(\Omega^n(X)\otimes\mathfrak{h}) \xrightarrow{\wedge} \Omega^{n+1}(X)\otimes(\mathfrak{g}\otimes\mathfrak{h}) \xrightarrow{\rhd} \Omega^{n+1}(X)\otimes\mathfrak{h}.$$ 
{This} allows to define the covariant derivative of $S\in\Omega^n(X)\otimes\mathfrak{h}$, 
$$
d_AS\equiv dS + A \wedge^\rhd S. 
$$
{Putting} together the  differential $d_A\cdot =d\cdot  +A\wedge^\rhd \cdot$ on $\Omega^n(X)\otimes\mathfrak{h}$ with the t-map,  and using the $\g$-equivariance {(We have $t(A \wedge^\rhd S)= [A\wedge t(S)]$.)} implies that the covariant derivative $d_A$ on $\mathfrak{h}$-valued forms is mapped under $t$ to the covariant differential $d_A$ on $\mathfrak{g}$-valued forms. This can be expressed compactly as 
\begin{equation}\label{td=dt}
    td_A = d_At.
\end{equation}

\begin{remark}\label{xmod2alg}
It is well-known that Lie algebra crossed-modules $\mathfrak{G}$ are equivalent to {\it strict} Lie 2-algebras \cite{Wag:2006,Bai_2013}. Following the mathematical literature, we call the strict Lie 2-algebra {\it skeletal} if $t=0$, and {\it trivial} if $t=\operatorname{id}$.
%in which  (\ref{eq:2comp}) is treated as a 2-term $L_\infty$-algebra \cite{Wag:2006}. 
% An important difference, however, is that in general a strict 2-algebra only satisfies \cite{Bai_2013}
% \begin{equation}
%     tY \rhd Y' + tY'\rhd Y = 0,\qquad \forall Y,Y'\in \h;\nonumber
% \end{equation}
% namely the quantity $(t\cdot)\rhd \cdot$ is skew-symmetric on $\h$. This allows us to \textit{define} a Lie bracket $[\cdot,\cdot]$ on $\h$ with it, whose Jacobi identity follows from the 2-Jacobi identities. In other words, the second Peiffer identity is a {\bf definition} of the Lie bracket 
%Here, strict means that the homotopy map/Jacobiator $\mu$, introduced later in  (\ref{eq:1biananom}), is trivial \cite{Bai_2013,Chen:2013,Kim:2019owc}. 

\end{remark}

% (see later), but {\it not} closed; by the Bianchi identities $d_AF=d_{A'}F'=0$, a simple computation shows that 
% \begin{eqnarray}
%     d_AB &=& d_AF' = d_AF'-d_{A'}F' = -a\wedge F',\nonumber \\
%     d_{A'}B &=& -d_{A'}F = d_{A}F - d_{A'}F = -a\wedge F,\nonumber
% \end{eqnarray} 
% hence the 2-curvature $K=d_AB$ in fact satisfies
% \begin{equation}
%     K+a\wedge F' = d_{A'}B + [a\wedge F]=0.\nonumber
% \end{equation}
% Notice that, by setting $a=0$, the 2-flatness condition $K=0$ reduces to simply the Bianchi identity $d_AF=0$. 

% \subsection{Fake-Flatness}
% In the above, we have taken $B$ to be valued in the same Lie algebra $\mathfrak{g}$.

\subsection{Curvatures and Bianchi Identities}\label{2alg}
Given the general 2-Lie algebra framework, we explore the different notions of curvatures that appear. First we have the notion of fake flatness which relates the 2-connection to the 1-curvature up to the t-map. We then express the properties of the 2-curvature and highlight it also satisfies a type of Bianchi identity. Finally, we discuss how the one kind of violation of the 1-Bianchi identity can be recast in terms of a 2-gauge theory based on a \textit{{weak}} 2-Lie algebra.

\subsubsection{Fake-Curvature}\label{sec:fake-flat}
When using the crossed-module formalism, the relation between the 2-connection and the curvature   we introduced in   \eqref{eq:2gau} can be rewritten  as
\begin{equation} 
t(\Sigma)=F'=da+a\wedge a,     \nonumber
\end{equation}
with  $\Sigma = dL + \demi [L\wedge L]$, provided that  $t(L) =a$. In fact  \eqref{eq:2gau} can be readily obtained if $\mathfrak{h}=\mathfrak{g}$ and the t map is the identity. Hence the construction in \eqref{eq:2gau} can be seen as an example of a 2-gauge theory based on the identity crossed-module.

The relation  \eqref{eq:2gau} can also be interpreted as a generalized notion of curvature 
\beq
\mathcal{F}= F'-t(\Sigma),\nonumber
\eeq
which is known as {\it {fake-curvature}}. The condition in which it is constrained to be zero,
\beq\label{fakeflat}
\mathcal{F}= F'-t(\Sigma)=0,
\eeq
is known as the \textit{{fake-flatness condition}}. A na{\" i}ve notion of ``2-parallel transport'' serves as a geometric motivation for imposing  \eqref{fakeflat} \cite{Yekuteli:2015}, but we need not assume it at the infinitesimal level based on a Lie algebra crossed-module/\textit{{strict}} Lie 2-algebra. We will see nevertheless that such condition can also appear when we consider 1- or 2-gauge transformations in Section \ref{1-2gauge}. 

\begin{remark}\label{adjpt}
It is possible to define a notion of higher-parallel transport \textit{without} fake-flatness $\mathcal{F}\neq 0$, which would move us into the realm of \textit{adjusted 2-parallel transport} \cite{Kim:2019owc}. We shall not consider this in detail here.
\end{remark}

As mentioned previously, we note that  \eqref{fakeflat} can be interpreted as sourcing the curvature with $t(\Sigma)$, allowing us to break away from a flat 1-connection. We will come back to this interpretation in Section \ref{sec:2app}.

\subsubsection{2-Curvature and 2-Bianchi Identity}
The 2-curvature is defined as the tensor $K=d_A\Sigma\in \Omega^3(X)\otimes \h$. When the 2-connection is pure 2-gauge $\Sigma=dL+\frac12[L\wedge L]$, we have as expected $K=0$,
\begin{align}
    d_A\Sigma = d^2L + \frac12d[L\wedge L]+ t(L)\rhd (dL+\frac12[L\wedge L])=0
\end{align}
where for simplicity we picked the 1-gauge where  $A=t(L)$ and we used that $d^2=0$, the Peiffer identity and the Jacobi identity for $\mathfrak{h}$.  

One may insert a {\it {2-curvature anomaly}} $\kappa\neq0$,  $K=\kappa$ to go away from the pure 2-gauge case. We will study this in Section \ref{twisting}.  As we are going to show, $K$ is valued in $\operatorname{ker}t$, and so must $\kappa$. Indeed, for any 2-connection, as a consequence of the fake-flatness condition and the 1-Bianchi identity, the 2-curvature must be valued in $\operatorname{ker}t\subset \mathfrak{h}$. 
\begin{equation} \label{Kkert}
t(K)= t(d_{A}\Sigma)= d_{A} t(\Sigma)   =d_{A} F=0.      
\end{equation}
{As} a consequence of the Bianchi identity, we  have that $d_AK\in \ker t$.

On the other hand,  by the graded Leibniz rule, the 2-curvature $K$ satisfies 
\begin{equation}
    d_A K = d_A(d_A\Sigma)=% + A\wedge^\rhd (d_A\Sigma) = dA\wedge^\rhd \Sigma - A\wedge^\rhd d\Sigma + A\wedge^\rhd d_A\Sigma =
    F\wedge^\rhd \Sigma= t(\Sigma)\wedge^\rhd \Sigma = [\Sigma\wedge \Sigma]|_{\operatorname{ker}t},\nonumber
\end{equation}
where we used the Peiffer conditions. Note that since $ d_A K$ is valued in $\ker t$, we should project the commutator $[\Sigma\wedge \Sigma]$ to $\ker t$. However, since $\Sigma$ is a 2-form and $[\cdot,\cdot] = (t\cdot )\rhd \cdot$ is skew-symmetric, %\flo{if the commutator is in center, why it has to be zero again? } 
this term vanishes and hence we achieve the  {\it {2-Bianchi identity}}
\begin{equation}
    d_AK =0.\label{2-bianch}
\end{equation}
{We} shall discuss in Section \ref{sec:gaug2gau} how such identity can be weakened.

% \begin{equation}
%     d_AK = [\Sigma\wedge \Sigma]|_{\operatorname{ker}t},\label{eq:1biananom} 
% \end{equation}
% where $|_{\operatorname{ker}t}$ denotes the projection to $\operatorname{ker}t\subset \mathfrak{h}$.
% and we used the fact that $\ker t$ is in the centralizer of $\mathfrak{h}$. We have obtained the 2-Bianchi identity.  

%%\medskip 

\subsubsection{1-Bianchi Anomaly and Weak 2-Lie Algebras}\label{weak}
Now suppose we forgo the 1-Bianchi identity, then  $K$ needs not be valued in $\operatorname{ker}t$. 
\begin{align}
  tK =  d_AF &= dF + [A,F] = d^2A+\frac{1}{2}d[A\wedge A] + [A\wedge dA] + \frac{1}{2}[A\wedge [A\wedge A]] 
  \nonumber\\ &= d^2A+\frac{1}{2}[A\wedge [A\wedge A]]\neq 0, \nonumber  
\end{align}
where we used that $d[A\wedge A] = [dA\wedge A] - [A\wedge dA] = -2[A\wedge dA]$.  There are two different ways to do this, one is to let $d^2A\neq0$ (globally), in which case we have a monopole. The other way is if the second term is non-vanishing, which occurs when we let go of the Jacobi identity on $\g$. In this case, $\g$ is strictly speaking no longer a Lie algebra; however, we shall see that the following structure we shall derive can also be applied to the case where $\g$ is a Lie algebra, but $t=0$ must be identically zero. 

\begin{remark}
The two ways in which the 1-Bianchi identity is violated are distinct. The violation of the Jacobi identity $[A\wedge [A\wedge A]]$ is of an algebraic nature, and hence introduces non-trivial modifications to our Lie 2-algebra structure; we shall focus on this case in the following. On the other hand, the monopole case $d^2A\neq 0$ is of  differential geometric  nature, which indicates a non-trivial topology of the 1-gauge theory. We shall discuss how this 1-gauge  topological feature can  be treated using a 2-gauge formalism, \textit{without} a violation of the Bianchi identity in Section \ref{QED-anom}. This will be an example of the notion of \textit{anomaly resolution}. 
\end{remark}

%In \cite{Kim:2019owc}, the bracket on the vector space $\g$ is noted $\mu_2$ to allow for the fact that it might not satisfy the Jacobi identity. 

% To examine the extra contributions that may appear, we compute
% \begin{equation}
%     d_AF = dF + [A,F] = \frac{1}{2}d[A\wedge A] + [A\wedge dA] + \frac{1}{2}[A\wedge [A\wedge A]].\nonumber  
% \end{equation}
% However, we have that $d[A\wedge A] = [dA\wedge A] - [A\wedge dA] = -2[A\wedge dA]$, as such we have
% \begin{equation}
%     tK = d_AF = \frac{1}{2}[A\wedge [A\wedge A]],\nonumber
% \end{equation}
% which would normally have vanished due to the Jacobi identity on $\mathfrak{g}$. 

By relinquishing the Jacobi identity, we may write this term as a contribution to $K$ by lifting it along $t$ up to $\mathfrak{h}$. In other words, we introduce a skew-trilinear map---called appropriately the { {Jacobiator}%MDPI: Please confirm if the bold is unnecessary and can be removed. The following highlights are the same. %Hank: confirm
}---satisfying
\begin{equation}
    \mu:\mathfrak{g}^{\wedge 3}\rightarrow\mathfrak{h},\qquad \frac{1}{3!}t\mu(A,A,A)=[A\wedge [A\wedge A]],\label{weakmu}
\end{equation}
such that the { {modified 2-curvature}} reads \cite{Kim:2019owc}
\begin{equation}
    K = d_A\Sigma - \frac{1}{3!}\mu(A,A,A) = 0.\label{eq:1biananom}
\end{equation}
{Since} the term $\mu(A,A,A)$ arises due to the failure of the 1-Bianchi identity, we call it the {\it {1-Bianchi anomaly}}. Note $\mu$ only appears for non-Abelian $\mathfrak{g}$. We note that the 1-gauge transformations need to be carefully analyzed in this case as $\mu(A,A,A)$ will not be a tensor. We discuss this in Section \ref{1-2gauge}.

%\medskip

This map $\mu$ is in fact precisely the {\it {homotopy map}} %Hank: please keep the italics on "homotopy map"
of a { {weak}} \cite{Chen:2012gz,Kim:2019owc} or a { {semistrict}}~\cite{BaezRogers} Lie 2-algebra. More precisely, the homotopy map is a trilinear skew-symmetric map $\mu:\mathfrak{g}^{\wedge 3}\rightarrow\mathfrak{h}$ satisfying
\begin{eqnarray}
[x,[x',x'']]+[x',[x'',x]]+[x'',[x,x']]=t\mu(x,x',x''),\nonumber \\
x\rhd (x' \rhd y) - x' \rhd (x\rhd y) - [x,x']\rhd y = \mu(x,x',t(y))\label{homotopy}
\end{eqnarray}
for each $x,x',x''\in\mathfrak{g}$ and $y\in\mathfrak{h}$. Indeed,  \eqref{weakmu} is equivalent to the first of these conditions.

An important additional property that $\mu$ {must} satisfy is its $\mathfrak{g}$-equivariance:
\begin{equation}
    x\rhd \mu(x_1,x_2,x_3) = \mu([x,x_1],x_2,x_3) + \mu(x_1,[x,x_2],x_3) + \mu(x_1,x_2,[x,x_3]).\label{gequiv}
\end{equation}
{As} such, we can compute
\begin{eqnarray}
    d_A\mu(A,A,A) &=& d(\mu(A,A,A)) + A\wedge^\rhd \mu(A,A,A) \qquad\text{$\mathfrak{g}$-equivariance and Leibniz rule} \nonumber \\
    &=& \big(3\mu(dA,A,A)) + \frac{3}{2} \mu([A,A],A,A)\big) \qquad \text{Trilinearity of $\mu$} \nonumber \\
    &=& 3\mu(F,A,A),\nonumber
\end{eqnarray}
where $\circlearrowright$ denotes a summation over cyclic permutations.  The factor of $\frac{3}{2}$ appears in the {second} line due to the fact that $\mu([A,A],A,A)$ is symmetric under an exchange of the first argument $[A,A]$ and the last two arguments $A,A$.
This gives rise to the { {modified 2-Bianchi identity}}
\begin{equation}
d_AK = F\wedge^\rhd \Sigma - \frac{1}{2}\mu(F,A,A)=0,\nonumber
\end{equation}
which has also appeared in the context of the gauge theory based on a weak Lie 2-algebra~\cite{Kim:2019owc}.

%\medskip

\begin{remark}\label{rem:string}
Notice that if the weak Lie 2-algebra is skeletal, namely $t= 0$, there is no violation to the Jacobi identity in the component $\g$. An example is the skeletal model { {string Lie 2-algebra}} $\mathfrak{string}_k(\mathfrak{g})$ of a simple Lie algebra $\g$ \cite{Kim:2019owc,Baez:2005sn}, where $k\in\bbZ$ is called the {\it level}. The Lie 2-algebra structure is given by $t=0,\rhd=0$, and the Jacobiator is $\mu = k\omega$, where $\omega$ is the fundamental 3-cocycle
\begin{equation}
    \omega = \langle\cdot,[\cdot,\cdot]\rangle \in Z^3(\mathfrak{g},\mathbb{R}).\nonumber
\end{equation}
{This} is one of the most commonly-seen weak Lie 2-algebras in the physics literature. The bundle gerbe associated to the string Lie 2-algebra describes the \textit{string structure} appearing in string theory~\cite{Carey:2005,Sati:2009ic}.
\end{remark}

% %\medskip 

% \begin{table}[t]
%     \centering
%     \begin{tabular}{||r c l||}
%     \hline
%   $F=0=t(\Sigma)$ & $\rightarrow$ & 
%         $
%         \begin{array}{ccccl}\\
%              d_AF=0  &  \rightarrow& K=d_A \Sigma\in \ker t&\rightarrow& d_AK= 0
%              \\
%              \\
%         \end{array}$ \\ 
%     % \hline
%     % \hline
%         $F=t(\Sigma)\neq0$ & $\rightarrow$ & 
%         $\left\{
%         \begin{array}{ccccl}
%              d_AF=0  &  \rightarrow& K=d_A \Sigma\in \ker t&\rightarrow& d_AK= [\Sigma,\Sigma]|_{\ker t}=0\\
%              \\
%              d_AF\neq0  &  \rightarrow& K=d_A \Sigma - \frac{1}{3!}\mu(A,A,A)&\rightarrow& d_AK= F\wedge^\rhd\Sigma -  \frac{1}{3!}\mu(F,A,A)
%         \end{array}\right.$ \\ \\
%         \hline
%     \end{tabular}
%     \caption{We summarize the different cases we consider. Note that, though there is no curvature anomaly, we can still consider a 2-gauge theory. If we have a 1-curvature anomaly, then the Bianchi identity may not hold. In which case we would obtain a non-trivial 2-curvature. }
%     \label{tab:1}% I've fixed it!
% \end{table}

\subsection{Gauge Transformations}\label{1-2gauge}
In this section, we review the different transformations we can perform and the inherited compatibility conditions.

\subsubsection{1-Gauge Transformations}
In order to preserve the fake flatness condition, we derive the transformations of $\Sigma$ and then $K$, from the transformation of the curvature 2-form  \eqref{1-gauge}. 
\begin{eqnarray}
F&\rightarrow& F^\lambda = F + [F,\lambda]\Rightarrow t(\Sigma) \rightarrow t(\Sigma)+ [t(\Sigma),\lambda] = t(\Sigma)- t(\lambda \rhd \Sigma) \nonumber\\
\Sigma &\rightarrow& \Sigma - \lambda \rhd \Sigma \nonumber\\
K=d_A \Sigma &\rightarrow& K- \lambda \rhd K,\label{1gau1}
\end{eqnarray}
where $\lambda\in\Omega^0(X)\otimes\mathfrak{g}$. 
% Note here that, in  \eqref{1gau1}, we have used the graded Leibniz rule and implicitly neglected exact boundary terms, such that
% \begin{equation}
%     d\lambda\wedge^\rhd\Sigma=d(\lambda\wedge^\rhd\Sigma)+\lambda\wedge^\rhd(d\Sigma)\approx \lambda\wedge^\rhd(d\Sigma).\nonumber
% \end{equation}

Now suppose the underlying Lie 2-algebra is weak, with $\mu\neq 0$. We shall see that, provided $\Sigma$ acquires an additional term \cite{Kim:2019owc}
\begin{equation}
    \Sigma\rightarrow \Sigma^\lambda = \Sigma- \lambda\rhd\Sigma - \frac{1}{2}\mu(\lambda,A,A)\label{modified1gau}
\end{equation}
under 1-gauge transformation, then we preserve the covariance of the 2-curvature under the 1-gauge transformations,
\begin{equation}
    K\rightarrow K^\lambda = K - \lambda\rhd K + \mu(\lambda,A,\mathcal{F}).\nonumber
\end{equation}
{Indeed,} working with the modified 2-curvature  \eqref{eq:1biananom}, we have from the definition  \eqref{homotopy},
\begin{small}
\begin{equation}
    -A\wedge^\rhd(\lambda\rhd\Sigma) + [A,\lambda]\wedge^\rhd\Sigma = -\mu(A,\lambda,t\Sigma) - \lambda\rhd(A\wedge^\rhd\Sigma)= \mu(\lambda,A,t\Sigma) - \lambda\rhd(A\wedge^\rhd\Sigma).
\end{equation}\end{small}
{On} the other hand, we have by the $\g$-equivariance of $\mu$, \eqref{gequiv}, that
\begin{eqnarray}
    \mu(d_A\lambda,A,A) &=& \mu(d\lambda,A,A) - \frac{1}{2}(\mu([\lambda,A],A,A) - \mu([A,\lambda],A,A)) \nonumber\\
    &=& d(\mu(\lambda,A,A)) + 2\mu(\lambda,A,dA)\nonumber\\
    &\qquad& + \frac{1}{2}(\frac{2}{3}\lambda\rhd\mu(A,A,A)+2\mu(\lambda,A,[A\wedge A]) \nonumber\\
    &\qquad&-A\wedge^\rhd\mu(\lambda,A,A))\nonumber\\
    &=& 2\mu(\lambda,A,F)+\frac{1}{3}\lambda\rhd\mu(A,A,A)-d_A\mu(\lambda,A,A).\nonumber
\end{eqnarray}
{There} are three such terms, hence we have
\begin{equation}
    \frac{1}{3!}\mu(A,A,A)\rightarrow  \frac{1}{3!}\mu(A,A,A)+\mu(\lambda,A,F) +\frac{1}{3!}\lambda\rhd\mu(A,A,A)- \frac{1}{2}d_A\mu(\lambda,A,A) + o(\lambda^2)\nonumber 
\end{equation}
modulo terms of higher order in $\lambda$. These terms precisely cancel the $d_A\mu(\lambda,A,A)$ term in the 1-gauge transformation of $K$, as desired.

\subsubsection{2-Gauge Transformations}
The shift of the 1-connection parameterized by $L$ such that $a=t(L)$ is interpreted as \textit{{the 2-gauge transformation}}. Indeed, the 2-connection $\Sigma$ was introduced such that the 1-form shift $A\rightarrow A'=A+t(L)$ in the 1-connection was interpreted as a (2-)\textit{{gauge symmetry}}. 

Given the 2-form connection $\Sigma$ undergoes a corresponding 2-gauge transformation, %We work in the gauge where $A=0$ and will discuss the compatibility between the 1- and 2- gauge transformations shortly.
\begin{align}\label{2gau}
%\Sigma\rightarrow \Sigma'=\Sigma + d_{A}L + \frac{1}{2}[L\wedge L],
\Sigma\rightarrow \Sigma'=\Sigma + d_AL + \frac{1}{2}[L\wedge L],
\end{align}
parameterized by a 1-form $L\in \Omega^1(X)\otimes\mathfrak{h}$, we see that the fake-curvature $\mathcal{F}=F-t\Sigma$ is kept invariant, as desired.
% This is clear if we label $a=t(L)$ and leverage the computations above  (\ref{eq:2gau}). 
The 2-curvature is covariant under this 1-form shift transformation since, with $A'=A+t(L)$, 
\begin{eqnarray}
    K \rightarrow K' &=& d_{A'}\Sigma' = d_A\Sigma +t(L)\wedge^\rhd \Sigma + d_{A+t(L)} (d_AL + \frac{1}{2}[L\wedge L])\nonumber\\
    &=& K+ [L\wedge \Sigma] +F\wedge ^\rhd L + \frac{1}{2}d_A [L\wedge L] + t(L) \wedge^\rhd d_AL + \frac{1}{2} t(L) \wedge^\rhd [L\wedge L]\nonumber\\
     &=& K- t\Sigma\wedge^\rhd L +F \wedge^\rhd  L  + \frac{1}{2}d_A [L\wedge L] + [L \wedge  d_AL] + \frac{1}{4} [L \wedge [L\wedge L]]\nonumber\\
     &=& K + \cF \wedge^\rhd L \sim K\label{gtff}
    %d^2 L + dL^2 +   + t(L)\wedge^\rhd dL + t(L) \wedge^\rhd \frac{1}{2}[L\wedge L].\nonumber 
\end{eqnarray}
where we used extensively the Peiffer conditions, and the Jacobi identity for the cubic term in $L$. Note $K$ is invariant {on-shell} of the fake-flatness condition $\mathcal{F}=0$.

%\medskip

Now let us consider how the modified 2-curvature $K$  \eqref{eq:1biananom} transforms in the weak case $\mu\neq 0$. We seek to pick out terms in the computation of  \eqref{gtff} that implicitly use the 2-Jacobi identities. All such terms occur in the quantity
\begin{equation}
    d_{A+t(L)} (d_AL + \frac{1}{2}[L\wedge L]),\nonumber
\end{equation}
which can be organized into three parts:
\begin{equation}
    o(L):~ d_Ad_AL,\qquad o(L^2):~ d_{tL}d_AL + \frac{1}{2}d_A[L\wedge L],\qquad o(L^3):~ \frac{1}{2}tL\wedge^\rhd [L\wedge L].\nonumber
\end{equation}
{Consider} first the term linear in $L$, which gives
\begin{equation}
    d_Ad_AL = (dA)\wedge^\rhd L + A\wedge^\rhd(A\wedge^\rhd L)= F\wedge^\rhd L + \frac{1}{2}\mu(A,A,tL)\nonumber
\end{equation}
by using  \eqref{homotopy}. The additional $\mu$-term here is compensated precisely by the linear $o(L)$-terms in the 2-gauge transformation of $\mu(A,A,A)$:
\begin{equation}
    \frac{1}{3!}\mu(A,A,A)\rightarrow \frac{1}{3!}\mu(A,A,A) + \frac{1}{2}\mu(A,A,tL)+o(L^2).\nonumber
\end{equation}

Next we look at the terms quadratic in $L$. This gives
\begin{equation}
    d_{tL}d_AL+\frac{1}{2}d_A[L\wedge L] = \frac{1}{2}A\wedge^\rhd[L\wedge L] + [L\wedge (A\wedge^\rhd L)] = \mu(A,tL,tL)\nonumber
\end{equation}
via  \eqref{homotopy}, which is compensated precisely by the $o(L^2)$-terms in the transformation
\begin{equation}
    \frac{1}{3!}\mu(A,A,A)\rightarrow \frac{1}{3!}\mu(A,A,A)+\frac{1}{2}\mu(A,A,tL) + \frac{3!}{3!}\mu(A,tL,tL) + o(L^3).\nonumber
\end{equation}
%where we have the factor $3!= 6 = 2\cdot \binom{3}{2}$. 
{Finally,} the cubic term is
\begin{equation}
    tL\wedge^\rhd[L\wedge L] = tL\wedge^\rhd [L\wedge L] = [L\wedge [L\wedge L]] = \frac{1}{3!}\mu(tL,tL,tL),\nonumber
\end{equation}
which is compensated by the $o(L^3)$-term in the transformation
\begin{equation}
     \frac{1}{3!}\mu(A,A,A)\rightarrow \frac{1}{3!}\mu(A,A,A)+\frac{1}{2}\mu(A,A,tL) + \mu(A,tL,tL) + \frac{1}{3!}\mu(tL,tL,tL).\nonumber
\end{equation}
{As} such, we see that the modified 2-curvature  \eqref{eq:1biananom} follows also the 2-gauge transform law~\eqref{gtff}.

\subsubsection{Compatibility Between the 1- and 2-Gauge Transformations}
The shift has to be compatible with the 1-gauge transformation, so that the new curvature transforms covariantly,
\begin{align}
&A \rightarrow A'= A+ a  \rightarrow A' + d_{A'}\lambda \Rightarrow  a = t(L)\rightarrow a + [a,\lambda] = t(L) + [t(L),\lambda] \\
& L \rightarrow L - \lambda \rhd  L
\end{align}
where we used the Peiffer conditions, as always. It is interesting to note that 1-gauge $(\lambda,0)$ and 2-gauge $(0,L)$ transformations do {\it {not}} commute. Through straightforward computations in the strict case $\mu=0$ \cite{Kapustin:2013uxa,Martins:2010ry,Radenkovic:2019qme}, we see that
\begin{equation}
    [(\lambda,0),(0,L)] = (0,\lambda\rhd L),\label{eq:2gaucomm}
\end{equation}
so 2-gauge transformations in general form a semidirect product \cite{Bullivant:2016clk,Martins:2010ry}
\begin{equation}
    \operatorname{Gau}_2 = (\Omega^1(X)\otimes\mathfrak{h})\rtimes(\Omega^0(X)\otimes\mathfrak{g})\nonumber
\end{equation}
defined by  (\ref{eq:2gaucomm}). 

It is possible to perform the same kinematical analysis for the weak case, where $\mu\neq 0$. However, here the commutator between 2-gauge transformations read \cite{Kim:2019owc}
\begin{equation}
    [(\lambda,L),(\lambda',L')] = ([\lambda,\lambda'],\lambda\rhd L' - \lambda'\rhd L) + (0,\mu(A,\lambda,\lambda')) + \mu(\mathcal{F},\lambda,\lambda'). \label{weakgaugeproblems}
\end{equation}
{This} is a major issue, because the additional term $\mu(\mathcal{F},\lambda,\lambda')$ is {\it {not}} a gauge transformation---the 2-gauge algebra $\operatorname{Gau}_2$ fails to close unless the fake curvature condition $\mathcal{F}=0$ is always satisfied! This is one of the motivations for the theory of adjusted parallel transport in \cite{Kim:2019owc}. Of course, when $\mu=0$, we have a set of compatible gauge transformations, even if possibly  $\cF\neq 0$. 

%\medskip

Generally, we also have a ``higher gauge transformation'' on the 2-gauge parameter $L\rightarrow L+d_A\ell$, where $\ell\in \Omega^0(X)\otimes\mathfrak{h}$. If we take the two 2-gauge parameters $L,L'=L+d_A\ell$, and define 
\begin{eqnarray}
    \Sigma' =\Sigma + d_AL + \frac{1}{2}[L\wedge L],&\qquad& \Sigma'' =\Sigma + d_AL' + \frac{1}{2}[L'\wedge L']\nonumber\\
    A' = A + tL,&\qquad& A'' = A + tL' = A+t(L+d_A\ell),\nonumber
\end{eqnarray}
then we have
\begin{eqnarray}
    \Sigma'' - \Sigma' &=& F\wedge^{\rhd}\ell+ [L,d_A\ell]+\frac{1}{2}[d_A\ell,d_A\ell],\nonumber\\
    F''-F' &=& [F,t(\ell)]+[tL,t(d_A\ell)]+ \frac{1}{2}[t(d_A\ell),t(d_A\ell)].\nonumber
\end{eqnarray}
{By} the Peiffer conditions, we see that the two 2-gauge transformations $L,L'=L+d_A\ell$ act identically on the fake-curvature $\mathcal{F} = F - t\Sigma$ \cite{Kapustin:2013uxa,Bochniak:2020vil}. The computation  \eqref{gtff} then implies that the 2-curvature $K$ is invariant on-shell of fake-flatness $\mathcal{F}=0$ under both $L,L'$. Because of this, the study of such higher gauge transformation is not necessary in the context of higher-BF theories \cite{Martins:2010ry}.

\subsection{2-Curvature Anomaly and the First Descendant}\label{twisting} 
Recall from \eqref{gtff} that the 2-curvature $K$ is covariant under a 2-gauge transformation. To introduce a 2-curvature anomaly $\kappa$ into the theory, we require the anomaly equation of motion (EOM) $K=\kappa$ to transform covariantly, identically to how $K$ transforms. On-shell of fake-flatness $\mathcal{F}=0$, then, $\kappa=\kappa(A,\Sigma)$ must be 2-gauge invariant. Now since under a 2-gauge transformation, $\Sigma$ shifts by an arbitrary element in $\h$ and hence $\kappa$ must be a constant as a function of $\h$.

On the other hand, {\it {shift invariance}} $\kappa(A) = \kappa(A+tL)$ implies that it can still have $A$-dependence through $\operatorname{coker}t=\g/\operatorname{im}t$.

% Here we focus on the first option in the following. 
Here, we will study this particularly nice form of the 2-curvature anomaly $\kappa(A)$. We shall see that the covariance of the 2-curvature anomaly EOM $K=\kappa(A)$ will require a \textit{{twist}} in the gauge transformations. 

\subsubsection{Twisting Gauge Transformations} Given the 1-form connection $A$ transforms in the usual manner, we shall demonstrate here that the 1-gauge transformation of the 2-form connection $\Sigma$ must be twisted by an additional term
\begin{equation}
    \Sigma\rightarrow \Sigma^\lambda =\Sigma -\lambda\rhd\Sigma + \zeta_A(\lambda).\label{twisted1gau}
\end{equation}
{This} additional contribution is required such that the 2-curvature anomaly equation $K=\kappa$ transforms appropriately. 

\begin{proposition}~
    \begin{enumerate}
        \item The quantity $\bar K = K-\kappa(A)$ transforms covariantly under 2-gauge transformations 
        \begin{equation*}
            \bar K \rightarrow \bar K^L = \bar K + \mathcal{F}\wedge^\rhd L
        \end{equation*}
        \underline{{iff} %MDPI: Please confirm if the underline is unnecessary and can be removed. The following highlights are the same. %Hank: please keep the underline
} the 2-form $\zeta$ is $\operatorname{ker}t$-valued and only a function of $\operatorname{coker}t$.
        \item the quantity $\bar K$ transforms covariantly under a 1-gauge transformation
        \begin{equation*}
            \bar K \rightarrow \bar K^\lambda = \bar K-\lambda\rhd \bar K
        \end{equation*}
        \underline{{iff}} $\zeta_A$ satisfies the following descent equation
    \begin{equation}
        d_{A^\lambda}\zeta_A(\lambda) = \kappa(A^\lambda) - (\kappa(A) -\lambda\rhd\kappa(A)).\label{eq:desceq}
    \end{equation}
    \end{enumerate}
    {We} call solutions $\zeta_A$ to \eqref{eq:desceq} the {{first descendants}} of the 2-curvature anomaly $\kappa(A)$ (cf. \cite{Kapustin:2013uxa}).
\end{proposition}
\begin{proof}
We prove the second statement first. Indeed, we first have the following computation 
\begin{eqnarray}
  K^\lambda &=& d_{A^\lambda}\Sigma^\lambda = d_{A^\lambda}(\Sigma-\lambda\rhd \Sigma) +d_{A^\lambda}\zeta_A(\lambda) \nonumber\\
  &=& d_A\Sigma - \lambda \rhd (d_A\Sigma) + d_{A^\lambda}\zeta_A(\lambda)\label{eq:2curvtran}
\end{eqnarray}
using \eqref{1gau1}. On the other hand, the 2-curvature anomaly transforms as $\kappa(A)\rightarrow \kappa(A^\lambda)$, hence from \eqref{eq:2curvtran} we have
\begin{eqnarray*}
    K - \kappa(A) &\rightarrow & K^\lambda -\kappa(A^\lambda) \\
    &=& d_A\Sigma - \lambda \rhd (d_A\Sigma) + d_{A^\lambda}\zeta_A(\lambda) - \kappa(A^\lambda) \\
    &=& (K-\kappa(A)) - \lambda\rhd(K-\kappa(A)) \\
    &\qquad& +~ d_{A^\lambda}\zeta_A(\lambda) - \kappa(A^\lambda) + \kappa(A)-\lambda\rhd\kappa(A).
\end{eqnarray*}
{The} last line is precisely the descent Equation \eqref{eq:desceq}. Note moreover that $\zeta_A(\lambda)$ is valued in $\operatorname{ker}t$ iff it does not conflict with the covariance of the fake-curvature,
\begin{equation}
    t(\Sigma) \rightarrow  t(\Sigma^\lambda) = t(\Sigma) - t(\lambda\rhd \Sigma) + \underbrace{t(\zeta_A(\lambda))}_{=0}. \nonumber
\end{equation}

Now we consider a 2-gauge shift symmetry. Note the covariance of the transformation $\bar K\rightarrow \bar K^L$ \eqref{gtff} implies that $\bar K^L -\bar K$ is in fact independent of $\kappa$, and hence both $\kappa$ and $\zeta$ cannot transform under $L$. By hypothesis, $\kappa(A)$ is shift invariant, hence we acquire the following terms from applying a 2-gauge transformation to the descent Equation \eqref{eq:desceq}:
\begin{eqnarray}
    tL\wedge^\rhd \zeta_A(\lambda) &=& -t\zeta_A(\lambda)\wedge^\rhd L=0,\nonumber\\
    {[tL\wedge \lambda]\wedge^\rhd\zeta_A(\lambda)} &=& -t(\lambda\rhd L)\wedge^\rhd\zeta_A(\lambda) = (t\zeta_A(\lambda))\wedge^\rhd(\lambda\rhd L)=0,\nonumber\\
    \zeta_{A^L}(\lambda) &=&\zeta_A(\lambda) + \zeta_{tL}(\lambda).
\end{eqnarray}
where we have used the Peiffer identity and the fact that $\zeta_A(\lambda)$ is $\operatorname{ker}t$-valued. Note the last term remains $\zeta_A$ iff $\zeta$ depends on $A$ only through $\operatorname{coker}t$, which would imply that \eqref{eq:desceq} is invariant under 2-gauge transformations. This ensures that first descendants do not transform under $L$, as desired.
\end{proof}

{Of course, } %MDPI: We removed the \noindent format, please confirm %Hank: confirm
if $\kappa=0$, then the first descendant $\zeta(A,\lambda)$ can be chosen to vanish, in which case we reproduce the covariance of the 2-curvature $K$ \eqref{1gau1}. Conversely, $\zeta_A(\lambda)$ \textit{{necessarily}} occurs in the presence of a non-trivial $\kappa(A)$. 

The {descent Equation}  (\ref{eq:desceq}) guarantees the 1-gauge covariance of the equation of motion $K=\kappa$, and provides a differential equation which allows to express $\kappa$ in terms of $\zeta$. As such, one may conversely view $\zeta$ as a particular twist in the 1-gauge transformation of $\Sigma$, which ``inserts'' the 2-curvature anomaly $\kappa$. 

\subsubsection{The Postnikov Class}
The anomaly $\kappa(A)$ has a cohomological interpretation. Indeed, $\kappa=\kappa(A)\in \Omega^3(X)\otimes\operatorname{ker}t$ can be interpreted as a  {\it {Lie algebra 3-cocycle}} $\kappa\in Z^3(\operatorname{coker}t,\operatorname{ker}t)$.

\begin{definition}
We call the cohomology class $[\kappa]\in H^3(\operatorname{coker}t,\operatorname{ker}t)$ the { {Postnikov class}} of the crossed-module $\mathfrak{G}$. A Lie algebra crossed-module/strict Lie 2-algebra $\mathfrak{G}$ is called { {non-trivial}} if $[\kappa]\neq 0$.
\end{definition}

{$[\kappa]$ classifies} %MDPI: We removed the \noindent format, please confirm %Hank: confirm
 the crossed-module $\mathfrak{G}$ up to elementary equivalence \cite{Wag:2006,Baez:2003fs}. Weak Lie 2-algebras are classified by the same data  \cite{Roytenberg:2007}s. We give further details about the Postnikov class in Appendix \ref{postnikov}.

Notice that the function $\kappa$ is only required to be a Lie algebra 3-cocycle, and hence is not necessarily covariantly closed. This means that, in the presence of $\kappa(A)$, the 2-Bianchi identity \eqref{2-bianch} can in fact be violated, due to the 2-curvature anomaly EOM $K=\kappa(A)$ giving $d_AK = d_A\kappa(A) \neq 0$.

The astute reader may have noticed a close parallel between the Postnikov anomaly $\kappa(A)$ and the Bianchi anomaly $\mu(A,A,A)$. They both define an anomaly of the 2-flatness condition, and the resulting 2-curvature quantity $K$ have identical gauge transformation properties. 

For $t\neq 0$, the two structures are actually different. Indeed,  the 1-Bianchi anomaly $\mu(A,A,A)$ is {\it {not}} invariant under the 1-form shift symmetry $A\rightarrow A+tL$, while $\kappa$ by hypothesis is. This speaks to the fact that, unlike their strict counterparts, weak Lie 2-algebras and non-trivial Lie algebra crossed-modules are \textit{{not}} equivalent when $t\neq 0$. Indeed, the component $\g$ in a weak Lie 2-algebra is not a Lie algebra, as the 2-Jacobi identities \eqref{homotopy} do not hold. The quantity $\frac{1}{2}\mu(\lambda,A,A)$ that appeared in  \eqref{modified1gau}, which seems to serve as the first descendant of $\mu(A,A,A)$, does {\it {not}} satisfy the descent Equation \eqref{eq:desceq}.

On the other hand, it is known that a non-trivial Lie algebra crossed-module $\g=\g_{-1}\xrightarrow{t}\g_0$ is classified, up to elementary equivalence, by precisely the data of a skeletal weak Lie 2-algebra $\s=(\mathfrak{n}\oplus V,\kappa)$ \cite{Baez:2003fs,Baez:2005sn}, where $\mathfrak{n}=\operatorname{coker}t$ and $\operatorname{ker}t=V$. Here, the Postnikov class $\kappa$ plays the role of the homotopy map for the 2-term graded Lie algebra $V\xrightarrow{0}\mathfrak{n}$. Indeed, as $\s$ is skeletal, there is no violation of the 2-Jacobi identities. Therefore, one may see a weak {\it {skeletal}} Lie 2-algebra as a non-trivial Lie algebra crossed-module.

\begin{remark}\label{postweak}
It was proven \cite{Baez:2005sn} that the skeletal, weak, string Lie 2-algebra $\mathfrak{string}_k(\mathfrak{g})$ has an alternative description in terms of a non-trivial crossed-module, called the loop model $\mathfrak{l}_k$. We shall show in Section \ref{loopgauge} that its Postnikov class/infinitesimal Dixmier-Douady class $[\kappa]\in  H^3(\g,\R)$ is represented in the 2-gauge theory by a 2-curvature anomaly.
\end{remark}

The non-trivial crossed-module formulation has the distinct advantage that the 2-gauge theory it defines is free of the problems plaguing that of a weak Lie 2-algebra, such as the lack of closure of gauge transformations in  \eqref{weakgaugeproblems}. This is precisely because of the descent equation satisfied by the first descendant $\zeta(A,\lambda)$ of $\kappa(A)$, which ensures that the 2-gauge structure closes and is consistent \cite{Kim:2019owc}, even in the presence of a non-trivial Postnikov class~\cite{Kapustin:2013uxa}. We shall see this point in action in Section \ref{strpt}.

\section{Applications}\label{sec:2app}
In this section, we discuss concrete examples of 2-gauge structures that arise naturally from physical applications.

\subsection{2-BF Theory}\label{sec:actions}
The simplest action to consider is an action for which  the fake-flatness and the flat 2-curvature  are obtained as equations of motion (EOMs), so such an action is topological. By analogy to the BF case, we would call this action a \textit{{2-BF action}} \cite{Martins:2010ry,Radenkovic:2019qme}. 
%Such constraints define flat 2-connections, and hence the   

We mention briefly here that the partition function of the 2-BF model can be  discretized. This gives rise to the \textit{{Yetter model}} \cite{Yetter:1992rz, Bullivant:2016clk,Bochniak_2021}, in which the 2-connections $(A,\Sigma)$ are assigned to a triangulation  as ``$\mathfrak{G}$-valued colorings'' \cite{Kapustin:2013uxa,Cui_2017}. It has been used to describe topological phases protected by higher-form global symmetries \cite{Kapustin:2013uxa,Zhu:2019,Thorngren2015}.

\subsubsection{Action and EOMs} 
Let $X$ be a manifold of dimension $d$ and let us fix a Lie algebra crossed-module $\mathfrak{G}=\h\xrightarrow{t}\g$. Let $\g^*$ denote the dual space of linear functionals on $\g$, and similarly let $\h^*$ denote the dual space of $\h$. We denote by $\langle\cdot,\cdot\rangle$ the duality pairing for them.

We begin by introducing Lagrange multipliers  $B\in \Omega^{d-2}\otimes\mathfrak{g}^*,\, C\in \Omega^{d-3}\otimes\mathfrak{h}^*$ which implement the aforementioned flatness conditions. The 2-BF action, also called the \textit{BFCG} action \cite{Martins:2010ry,Radenkovic:2019qme}, in the absence of 2-curvature {anomalies is} %MDPI: Please ensure all variables/values in the equation appear in the same format in the text (normal/italic/bold/subscript/superscript).
\begin{equation}
    S_\text{2BF}[A,\Sigma] = \int_X \langle B \wedge\mathcal{F}(A,\Sigma)\rangle + \langle C\wedge \mathcal{G}(A,\Sigma)\rangle,\label{eq:2bf}
    % \\
    % &=& \int_X \langle B \wedge\mathcal{F}(A,\Sigma)\rangle + \langle C\wedge \mathcal{G}(A,\Sigma)\rangle
\end{equation}
where $\mathcal{F}(A,\Sigma) = F - t(\Sigma)$ and $\mathcal{G}(A,\Sigma)=K = d_A\Sigma$. For $d<3$, the 2-$BF$ theory reduces to a $BF$ theory, since the dual field $C$ does not exist. 

The first half of the EOMs are
\begin{equation}
    \delta B \Rightarrow \mathcal{F}=F - t(\Sigma)=0,\qquad \delta C\Rightarrow \mathcal{G}=d_A\Sigma=0,\nonumber
\end{equation}
which implement precisely the fake curvature and 2-flatness conditions, respectively. On the other hand, we also have the option to vary $A$ and $\Sigma$.  These must be done more carefully: we first introduce a map $\Delta: \mathfrak{h}\wedge\h^*\rightarrow \mathfrak{g}^*$ dual to the crossed-module action:
\begin{equation}
    \langle C\wedge (A\wedge^\rhd \Sigma)\rangle = -\langle \Delta(C\wedge \Sigma)\wedge A\rangle.\nonumber
\end{equation}
{Second,} we define the map $t^*:\mathfrak{g}^*\rightarrow\mathfrak{h}^*$ dual (with respect to the pairings $\langle\cdot,\cdot\rangle$) to the crossed-module map $t:\mathfrak{h}\rightarrow\mathfrak{g}$, and write
\begin{equation}
    \langle B\wedge t(\Sigma)\rangle = \langle t^*(B)\wedge \Sigma\rangle.\nonumber
\end{equation}
{We} also introduce the dual of the action and adjoint representation, 
\begin{equation}
    \langle y,x\rhd y'\rangle= -\langle x\rhd^* y,y'\rangle,\qquad\langle x',[x,x'']\rangle = -\langle [x,x']^*,x''\rangle,\nonumber 
\end{equation}
for all $y\in\h,y'\in\mathfrak{h}^*,x\in\mathfrak{g},x'\in \g^*$.

These yield
\begin{equation}
    \delta A \Rightarrow dB+[A\wedge B]^* - \Delta(C\wedge\Sigma)=0,\qquad \delta \Sigma \Rightarrow t^*B + dC + A\wedge^{\rhd^*} C=0.\nonumber 
\end{equation}
{If} we define the quantities
\begin{equation}
    \tilde F \equiv d_AC = dC+A\wedge^{\rhd^*} C,\qquad \tilde K\equiv d_AB = dB+[A\wedge B]^*,\nonumber
\end{equation}
we see that these sets of EOMs read
\begin{equation}
    \tilde F = t^*(B),\qquad \tilde K = \Delta(C\wedge \Sigma), \label{eq:dual}
\end{equation}
the first of which looks like a fake-flatness condition for the dual fields. This suggests that $B,C$ should be treated as a 2-connection as well, valued in a Lie algebra crossed-module $t^*:\mathfrak{g}^*\rightarrow\mathfrak{h}^*$.

Indeed, this is precisely the {\it {dual Lie algebra crossed-module}} $$\mathfrak{G}^*[1]=(t^*:\mathfrak{g}^*\rightarrow\mathfrak{h}^*,\tilde\rhd),$$ whose graded Lie algebra structure is induced by the choice of a Lie algebra 2-cochain {(More concretely, we have
$
    \langle [f,f'],y\rangle = \langle f\wedge f',\delta_{-1}(y)\rangle,\qquad \langle f\tilde \rhd g,x\rangle = \langle f\wedge g,\delta_{-1}(x)\rangle,\nonumber
$
where $f,f'\in\mathfrak{h}^*,g\in\mathfrak{g}^*$.)}%MDPI:  We changed the single equation to inline equation, please confirm %Hank: please do not make this equation inline, as it includes a long space "\qquad".
\begin{equation}
    \delta_{-1}: \mathfrak{h}\rightarrow \mathfrak{h}^{\wedge 2},\qquad \delta_0:\mathfrak{g}\rightarrow \mathfrak{h}\wedge\mathfrak{g}.\nonumber
\end{equation}
{When} dealing with the action   \eqref{eq:2bf}, the dual Lie 2-algebra is Abelian with trivial 2-cochain $(\delta_{-1},\delta_0)=0$, and hence is the same thing as a 2-vector space \cite{Baez:2004in}. For more general details on these algebraic objects, we refer the reader to the mathematical literature \cite{Bai_2013,Chen:2012gz,chen:2022}.

\begin{remark}
Dualizing the crossed-module map $t:\mathfrak{h}\rightarrow\mathfrak{g}$ leads to $t^*:\mathfrak{g}^*\rightarrow\mathfrak{h}^*$, hence the dual Lie 2-algebra $\mathfrak{G}^*[1]$ comes with a shift $[1]$ in the grading of the underlying vector spaces. This is a small subtlety in the mathematical notation that we shall keep in order to be consistent with the literature. 
\end{remark}

\subsubsection{Symmetries of the Action}

It was shown in  \cite{Martins:2010ry} (see also  \cite{Radenkovic:2019qme}) that the 2-$BF$ action  (\ref{eq:2bf}) is preserved under the operations
\begin{eqnarray}
 \lambda : \begin{cases}\cF\rightarrow \cF^\lambda = \cF + [\cF,\lambda] \\ \cG \rightarrow \cG^\lambda = \cG + \lambda\rhd \cG \end{cases}, &\qquad& 
 L: \begin{cases} \cF\rightarrow \cF^L= \cF  \\ 
  \cG\rightarrow \cG^L= \cG + \cF\wedge^\rhd L
 \end{cases}, \label{2BFgaugetrans}\\
    \lambda: \begin{cases}B\rightarrow B^\lambda = B + [\lambda,B]^* \\ C\rightarrow C^\lambda = C + \lambda\rhd^* C \end{cases},&\qquad& L: \begin{cases} B\rightarrow B^L= B + \Delta(C\wedge L) \\ C\rightarrow C^L= C\end{cases}, \label{2BFdualgaugetrans}
\end{eqnarray}
where we recognize the transformations of $\cF$ and $\cG$ we obtained in Section \ref{1-2gauge}. Notice $\cG^L$ is invariant only {\it {on-shell}} of the fake curvature condition $\cF =0$, which we had assumed in~\eqref{gtff}.

Algebraically, this implies that the 2-gauge group $\operatorname{Gau}_2= (\Omega^1(X)\otimes\mathfrak{h})\rtimes (\Omega^0(X)\otimes\mathfrak{g})$ acts naturally on the dual fields $B,C$. In other words, the original Lie 2-algebra $\mathfrak{G}$ has a natural action on the dual Lie 2-algebra $\mathfrak{G}^*[1]$ induced by the data $\rhd^*,\Delta$ emergent from the dual EOMs  \eqref{eq:dual}. These actions define a \textit{{strict coadjoint representation}} \cite{Bai_2013} of the Lie 2-algebra $\mathfrak{G}$ on its dual $\mathfrak{G}^*[1]$.

\begin{remark}\label{2ddsymm}
In general, the dual Lie 2-algebra $\mathfrak{G}^*[1]$ can be non-Abelian and define its own gauge sector. The corresponding gauge parameters $(\tilde\lambda,\tilde L)\in\mathfrak{G}^*[1]$ transforms the dual fields $(C,B)$ as
\begin{equation}
    \tilde\lambda: \begin{cases}C\rightarrow C^{\tilde\lambda} = C + d_C\tilde\lambda \\ B\rightarrow B^{\tilde\lambda} = B + \tilde\lambda \rhd^* B\end{cases},\qquad \tilde L:\begin{cases}C\rightarrow C^{\tilde L} = C + \tilde t \tilde L \\ B\rightarrow B^{\tilde L} = B + d_C\tilde L + \frac{1}{2}[\tilde L\wedge \tilde L]_*\end{cases}.\nonumber
\end{equation}
{If} there is a non-trivial back-action of $\mathfrak{G}^*[1]$ on $\mathfrak{G}$, then $(A,\Sigma)$ would transform under $(\tilde\lambda,\tilde L)$ as well, analogous to how $(C,B)$ transforms under $(\lambda,L)$ in  \eqref{2BFdualgaugetrans}. If certain coherence conditions are satisfied between these actions, then the pair $(\mathfrak{G},\mathfrak{G}^*[1])$  defines a {\it 2-Manin triple}
\begin{equation}
    \mathfrak{D} = \mathfrak{G}~_{\operatorname{ad}^*}\bowtie_{\mathfrak{ad}^*}\mathfrak{G}^*[1],\nonumber
\end{equation}
which serves as a model for a ``2-Drinfel'd double''  \cite{Bai_2013,chen:2022}---a categorified notion of the classical Drinfel'd double $\mathfrak{d} = \mathfrak{g}\bowtie \mathfrak{g}^*$ for a Lie algebra $\mathfrak{g}$ \cite{Semenov1992}. For a more detailed study and analysis, see~\cite{chen:2022}.
\end{remark}

\subsection{3D Gravity}\label{3Dgrav}
One of the most direct applications of 2-BF theory is gravity. In 3-dimensions, gravity is topological, as there are no propagating local degrees of freedom.  In the Einstein--Cartan formalism,  the Einstein equations take the shape \cite{Yepez:2011,Zanelli:2005} ($4\pi G=1$)
\beq
F_{I}\, =  \, T_{I},
\eeq
where $F$ is the curvature of the spin-connection $A$, and $I$ is the internal $\mathfrak{su}(2)$-index (we chose the Euclidian signature). $T_{I}$ is the stress-energy tensor {(Typically in the usual framework, $T_{I}=\frac{\delta {S}_\text{matter}}{\delta e^I}$, where ${S}_\text{matter}$ is the action for the matter degrees of freedom.)}, which could include the cosmological constant contribution. By assumption, $F$ satisfies the Bianchi identity and hence $T$ is conserved. Following our discussion from the previous sections, it seems natural to interpret the stress energy contribution as some form of curvature excitation and as such it would fit within the scheme of 2-gauge theory. In particular we would identify the stress-energy tensor with the 2-connection up to the t-map. 
\beq
F_{I}\, =  \, T_{I} = t_{I}(\Sigma).
\eeq
{Note} however that the stress-energy tensor is actually fixed, there is no shift symmetry. Hence one would expect to recover 3D gravity as some kind of 2-gauge fixed theory. Indeed, starting from the 2-$BF$ action, and fixing the shape of the 2-connection allows to recover 3D gravity coupled with a particle (seen as a topological defect) or with a non-zero cosmological constant. 

%\medskip 

First let us fix the crossed-module to be the (infinitesimal) identity crossed-module 
$$\mathfrak{I}_{\mathfrak{su}(2)}=(t=\operatorname{id}:\mathfrak{su}(2)\rightarrow \mathfrak{su}(2),\rhd = \operatorname{ad}),$$
to discuss 3D gravity in the Euclidean signature. We will note $J_I$ the generators of $\mathfrak{su}(2)$.  We now consider the 2-$BF$ action based on $\mathfrak{I}_{\mathfrak{su}(2)}$ with $X$ a 3D manifold.
\begin{equation}
    S_\text{2BF}[A,\Sigma, B,C] = \int_X \langle B \wedge(F-\Sigma)\rangle + \langle C\wedge d_A \Sigma\rangle,\label{3d2bf-lambda1}
    % \\
    % &=& \int_X \langle B \wedge(F-[e\wedge e])\rangle + \langle C\wedge d_A [e\wedge e]\rangle .\label{eq:2bf-lambda2}
\end{equation}
and discuss the different values $\Sigma$ can take. The fake flatness condition, which is one of the EOMs,  is then 
\begin{equation} \label{fake-eins}
F=\Sigma.
\end{equation}

%\medskip 

The first obvious value is to 2-gauge fix $\Sigma$ to $0$, which would amounts to recovering the pure (Euclidian) gravity 
\begin{equation}
    S_\text{2BF}[A,\Sigma=0, B,C] =S_\text{BF}[B,A]= \int_X \langle B \wedge F \rangle,
\end{equation}
with the 1-form $B$ interpreted as the frame field. Plugging back $\Sigma=0$ in \eqref{fake-eins} allows to recover the 3D vacuum Einstein equation. 

%\medskip 

The next interesting value is to pick a 1-gauge where we  2-gauge fix $\Sigma$ to be $m J_3 \delta_W^{(2)}(x)$, where $\delta^{(2)}(x)$ is the densitized Dirac delta function. This value localizes $\Sigma$ on a worldline $W$, and $m$ is the mass of the defect. In an arbitrary 1-gauge, parameterized by $g$, we have then   $\Sigma^I= p^I \delta_W^{(2)}(x)=\Sigma^I_p$, with $p^I= m g^{-1} J_3 g$ interpreted as the momentum of the defect.  

With such value, the 2-BF action \eqref{3d2bf-lambda1} becomes 
\begin{equation}
    S_\text{2BF}[A,\Sigma^I_p, B,C] = \int_X \langle B \wedge F\rangle - \int_W \langle B, p\rangle  + \int_W \langle C\wedge d_A p\rangle,\label{3d2bf-lambda2}
    % \\
    % &=& \int_X \langle B \wedge(F-[e\wedge e])\rangle + \langle C\wedge d_A [e\wedge e]\rangle .\label{eq:2bf-lambda2}
\end{equation}
where we recognize the standard action of gravity coupled with a particle \cite{deSousa1990}, supplemented by a term encoding the conservation of momentum. Plugging back $\Sigma^I=p^I \delta_W^{(2)}(x)$ in \eqref{fake-eins} allows to recover the 3D Einstein equation in the presence of a point-like particle. 

Such construction could be extended to the discrete case where one could analyze how the Yetter amplitude can provide the Ponzano--Regge amplitude coupled to a particle. This will be explored elsewhere. 

%\medskip 

Finally, in general as a 2-form on a 3D manifold with value in $\mathfrak{su}(2)$, there exists a covector $\tilde e$ and an arbitrary constant rescaling $\lambda$ such that  
\beq \label{cosmo}
\Sigma=\frac\lambda2 [\tilde e\wedge \tilde e]=\Sigma_\lambda.   
\eeq
%in which case the fake flatness becomes  $F=\lambda [\tilde e\wedge \tilde e]$. 
{Plugging} back this value of $\Sigma$ into \eqref{fake-eins} would  resemble Einstein equation in the presence of a cosmological constant $\lambda$, if $\tilde e$ was identified as the frame field $B$. Let us therefore write  without loss of generality $\Sigma=\frac\lambda2 [\tilde e\wedge \tilde e]$ and  impose that $\tilde e =B$. 
At the level of the 2-BF action, we have therefore 
\begin{equation}
    S_\text{2BF}[A,\Sigma_\lambda, B =\tilde e, C] = \int_X \langle B \wedge(F-\frac\lambda2 [\tilde e\wedge \tilde e])\rangle - \lambda\langle [C\wedge \tilde e]\wedge  d_A\tilde e \rangle + \langle \phi, \tilde e - B\rangle ,\label{3d2bf-lambda3}
    % \\
    % &=& \int_X \langle B \wedge(F-[e\wedge e])\rangle + \langle C\wedge d_A [e\wedge e]\rangle .\label{eq:2bf-lambda2}
\end{equation}
where $\phi$ is a Lagrange multiplier. 
%Following the previous remarks, we intend to reocver the 3d gravity action from this 2-BF action. 

The Palatini formulation of 3D gravity involves  $e_\mu^I, \, A_\mu^I$, i.e., 18 variables. In the 2-BF action, we have additional 12 variables from $B_\mu^I$ and $C^I$. To recover the Palatini formulation, we go on-shell of 2-flatness $d_A\Sigma=0$ and impose the constraint
\beq \label{const-lambda}
B=\tilde e,
\eeq
which allows to reduce on-shell to the usual number of variables and 3D gravitational~action
\begin{eqnarray}
    S_\text{2BF}[A,\Sigma_\lambda  , B,C] &\approx& S_\text{gr}[A,\tilde e]= \int_X \langle \tilde e \wedge(F-\frac{\lambda}{3}[\tilde e\wedge \tilde e]\rangle. \label{cosmograv}
\end{eqnarray}
% where $\sim$ means we went on-shell of the constraints  \eqref{const-lambda}. The first constraint $B=e$ was implemented implicitly also in \eqref{eq:2bf-part1}; that was not a problem, as the frame field $e$ was a "spectator" field that did not directly participate in the particle dynamics. 
%Under the constraint  \eqref{const-lambda}, the coframe $e$ becomes a part of the dynamical data in the 2-BF theory. 
{The} dual EOMs  \eqref{eq:dual} state that
\begin{equation}
    d_AC = B,\qquad d_AB=d_A\tilde e = [\Sigma,C],\nonumber
\end{equation} 
namely the coframe $B$ is covariantly exact (which can always be achieved locally \cite{Li_1999}) and the torsion $T = d_A B$ is given by $[C, \Sigma]=\frac\lambda2[C,[\tilde e \wedge \tilde e]]$. Of course, torsion-freeness $T=0$ requires that this quantity must vanish.
% Together with 
% between the 2-form $\Sigma$ and the dual $B$-field.
In order to see this, we  use the Jacobi identity such~that
\begin{equation}
    \frac{1}{2}[C,[\tilde e\wedge \tilde e]] = [\tilde e\wedge [\tilde e,C]].\nonumber
\end{equation}
{Now} as $C$ is a 0-form, we have $[C,C]=0$, hence on-shell of the duel EOM $d_AC=B$ we~have
\begin{equation}
    0=d_A[C,C] = 2[d_AC,C]=2[B,C]=2[\tilde e,C],\nonumber
\end{equation}
and we indeed have torsion-freeness $d_A B= 0$. 
%This demonstrates that  \eqref{const-lambda} is not in contradiction with the 2-gauge theory  \eqref{eq:2bf-lambda1}.

\begin{remark}\label{rmk:simpcon}
In the 2-gauge formalism  \eqref{3d2bf-lambda3}, we can say that we have a pair of (co-)frame fields, $B$ and $e$. The constraint to recover canonical 3D gravity  \eqref{cosmograv} is to identify them $B=e$ via  \eqref{const-lambda}. This is an analogue of the {\it simplicity constraint} in 4D gravity: there, one starts with 4D $BF$ action based essentially on a pair of frame fields, then impose the simplicity constraint that identifies them~\cite{Freidel:2020ayo, Freidel:2020svx, Freidel:2020xyx}.  It would then be interesting to see how this construction at the quantum level can allow us to impose the simplicity constraints more consistently.
\end{remark}

\subsection{Anomaly Resolution: Monopole Electrodynamics} \label{QED-anom} % dont forget references :)
For this example, we combine ideas from \cite{Cordova:2018cvg,Benini_2019,Dubinkin:2020kxo} and introduce the magnetic monopole in the context of 2-gauge theory. Let $X$ denote a closed oriented smooth 4D manifold, which is {\it {not}} spin. 

Consider on $X$ a principal $U(1)$-gauge bundle, whose curvature 2-form $F$ has non-trivial flux across a 2-surface $S\subset X$. This can be interpreted as a violation of the Bianchi identity $dF\neq 0$: on the 3-surface $V$ spanned by $S\subset X$, we have
\begin{equation}
    \int_S F = \int_{V} dF \neq 0 \nonumber
\end{equation}
by Stokes's theorem. Physically speaking, $S$ encloses a \textit{{magnetic monopole}}, whose current $j_m$ is given by the EOM $dF=\ast j_m\neq 0$. In the usual manner, we define
\begin{equation}
    q_m = \frac{1}{2\pi}\int_V \ast j_m = \frac{1}{2\pi}\int_S F\label{monodef}
\end{equation}
as the {\it {magnetic charge}} enclosed by the bounding 3-surface $V\subset X$.

\begin{remark}\label{rmk:qed}
The historical motivation for studying such an ``anomalous'' Maxwell's theory is that $\text{QED}_{3+1}$ suffers from a {\it perturbative chiral anomaly} in the presence of a {\it single} spin-$\frac{1}{2}$ Weyl fermion. The counterterm for this anomaly is the {\it Abelian Chern--Simons action} \cite{adler2004anomalies}
\begin{equation}
    \mathcal{A}[A] = \frac{1}{4\pi}\int_X F\wedge F;\nonumber
\end{equation}
indeed, by an integration by parts, we can introduce this Chern--Simons term $\mathcal{A}[A]=-\frac{1}{4\pi}\int_X \ast j_m \wedge A$ with a monopole current $j_m$, if we go on-shell of the EOM $\ast j_m = dF$. 
\end{remark}

\subsubsection{Monopoles and Large Gauge Transformations} 
Recall that a non-trivial monopole current $j_m\neq 0$ implies that $d^2A=dF =\ast j_m \neq 0$. This type of failure of the Bianchi identity implies that the $U(1)$ connection $A$ acquires a non-trivial holonomy about some closed (timelike) 1-cycle $l\subset X$, called the {\it {monopole worldline}}. 

To treat this problem, we ``fatten'' (i.e., take a small tubular neighborhood around) $l$ and excise it away from $X \rightsquigarrow X'$ \cite{book-wen}. This yields new 4-manifold $X'$, which has a boundary $\partial X' \cong S^2\times l$; see Figure \ref{fig:monopole} {(Left).} %MDPI: We removed the bold. Please confirm this revision.
 The $U(1)$-connection $A$ is now regular on $X'$, but its gauge transformation comes with an ``anomalous'' component \cite{Halperin:1978,Kogan:1995}
\begin{equation}
    A\rightarrow A + d\lambda_0 + d\lambda_1,\qquad d^2\lambda_0=0,\quad d^2\lambda_1\neq 0.\nonumber 
\end{equation}
{This} component is required such that, by integrating over the 2-sphere $S^2\subset S^2\times l$ enclosing the monopole, we achieve
\begin{equation}
    {q}_m = \frac{1}{4\pi}\int_{S^2}F \rightarrow \frac{1}{4\pi}\int_{S^2}F+d^2\lambda_0 + d^2\lambda_1 = \frac{1}{2\pi}\int_{S^1}d\lambda_1 \neq 0 \nonumber
\end{equation}
a non-trivial monopole charge, where $S^1 = H_+\cap H_-\hookrightarrow S^2$ is the equator, and we have used Stokes's theorem on each of the patches {(On a 3D Cauchy slice containing $V$, the equation $d\lambda_1 = \ast dA=\ast F$ is known as the (Abelian)} {\it {monopole equation}} \cite{Dai2005OnTS}.) $H_\pm$ covering $S^2$. Historically, this anomalous $U(1)$ gauge theory also plays an important role in the vortex-driven 2D Kosterlitz--Thouless phase transition \cite{Halperin:1978}.

Conversely, if the $U(1)$ gauge symmetry $A\rightarrow A+ d\lambda_0$ is non-anomalous, then there can be no monopole charge ${q}_m=0$ and the magnetic current must vanish
\begin{equation}
    {q}_m = \frac{1}{4\pi}\int_{S^2\times l}dF = \frac{1}{4\pi}\int_{S^2\times l}\ast j_m =  0.\nonumber
\end{equation}
{The} construction of a non-trivial monopole charge $q_m$ from patching the anomalous gauge transformation $d\lambda_1$ (as a {\v C}ech cocycle) across the equator of a sphere is called { {clutching}}~\cite{book-algtop}. As such, such components $L = d\lambda_1$ are also known as { {large gauge transformations}} in the literature \cite{Kogan:1995,adler2004anomalies}, as it is able to detect ``large'' (i.e., topological) data. More generally, this casts the monopole charge $q_m$ as an element in a {\it {differential cohomology}} of $X$ \cite{FreedMonopole}. 

\begin{figure}[h]
%\centering
\includegraphics[width=0.5\columnwidth]{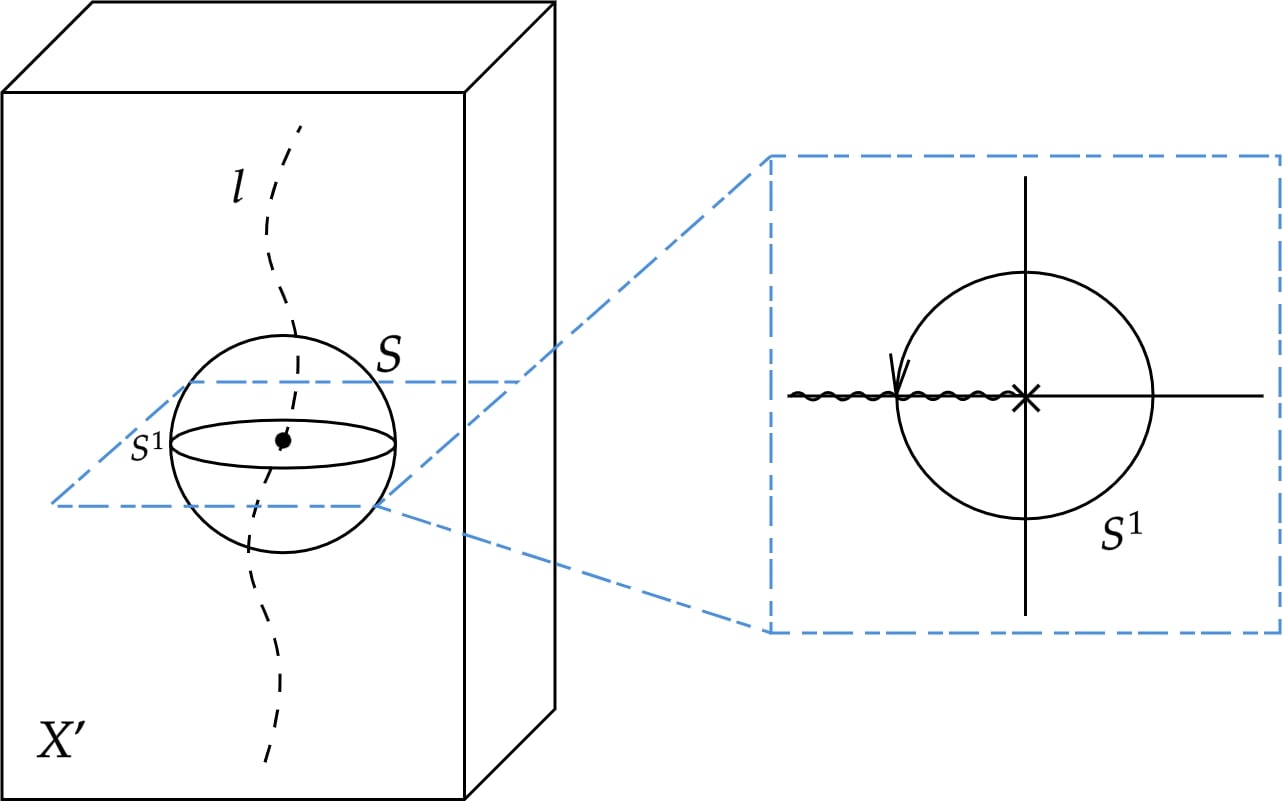}
\caption{The clutching construction for the monopole. {\bf Left}: The excised spacetime 4-manifold $X'$ and its boundary $\partial X' \cong S\times l$. {\bf Right}: The singularity structure of $\lambda_1$ on the equatorial slice $S^1$ across $S\cong S^2$. If the monopole charge $q_m$ is non-trivial, then a branch cut appears.}
\label{fig:monopole}
\end{figure}

Any given monopole current $j_m$ can thus be inserted by designing the singularity of $\lambda_1$; see Figure \ref{fig:monopole} {(Right).} %MDPI: We removed the bold. Please confirm this revision. %Hank: confirm
 It computes the winding number
\begin{equation}
    \int_{S^1}d\lambda_1 = \lambda_1(1) - \lambda_1(0) \in 2\pi\bbZ,\nonumber
\end{equation}
hence the quantized monopole charge $q_m\in\bbZ$ is fixed by topology {(This accounts for the Dirac monopole quantization ${q}_m\in\bbZ$ in 4D, or the quantization of the vortex charge in 2D}~\cite{Halperin:1978}.). In fact, this winding number is precisely the first Chern number $\int_S c_1(F)$, where the {\it {first Chern class}} $c_1(F) = \frac{1}{2\pi}[F]\in H^2(X,\bbZ)$ topologically classifies the $U(1)$-bundle on $X$ up to isomorphism. More details can be found in  \cite{book-wen}.

\begin{remark}\label{tube}
Given two $U(1)$-bundles $P,P'\rightarrow X$ with monopole charges $q_m,q_m'$, the tensor bundle $P\otimes P'$ has the {\it sum} $q_m+q_m'$ as its monopole charge; this is precisely the additive property of characteristic classes \cite{book-algtop}
\begin{equation}
c_1(P\otimes P') = c_1(P)+c_1(P'),\nonumber
\end{equation}
and describes the process of monopole fusion. By the clutching construction, we can identify monopole defects with branch points of $L=d\lambda_1$ along the equator $S^1\subset S^2$. As multiple monopoles time-evolve and fuse, a graph is traced out on the cylinder $S^1 \times l$. The fusion algebra of such graph states is known as the {\it (2+1)D Ocneanu's tube algebra} \cite{ocneanu2001,Delcamp:2016yix,Bullivant:2021}.
\end{remark}

If we write $d\lambda_1 = L$ as a generic 1-form, then we see that the anomalous $U(1)$ gauge theory acquires a {\it {shift}} $A\rightarrow A+L$, generating a ``$U(1)_1$ 1-form symmetry'' \cite{Cordova:2018cvg,Benini_2019}. We can then follow our previous construction and introduce a 2-form $\Sigma$ as in  \eqref{eq:2gau}.

\subsubsection{Resolving the Monopole Anomaly; 2-Gauge Structure} 
One may treat monopole Maxwell's theory as a {\it {compact $U(1)$ gauge theory}} equipped with a quantized Gauss law $$\frac{1}{2\pi}\int_S F = 0 \mod\bbZ,$$ which forces $A \sim A +2\pi$ to be defined only modulo $2\pi\bbZ$ \cite{Kogan:1995}. Alternatively, however, introducing a 2-gauge structure is in fact the most consistent way to treat the monopole~\cite{Cordova:2018cvg,Benini_2019}.

The idea is to insert a {\it {2-form}} gauge field $\Sigma$ that absorbs the flux of the magnetic monopole. This is accomplished by the crucial {\it {monopole property}}
\begin{equation}
    \int_S \Sigma = \int_S F,\qquad \forall\text{ closed 2-surfaces }S\subset X',\label{eq:mono}
\end{equation}
which states that the quantized monopole charge ${q}_m\in\bbZ$ is matched by the 2-form gauge field $\Sigma$. However, the monopole condition  (\ref{eq:mono}) does {\it not} imply the fake-flatness $\Sigma = F$ on the nose, but only up to closed 2-forms. We shall abuse notation and write such closed 2-forms as $dL$, which can in general have a non-trivial integral over the surface $S$. 

The curvature $F$ transforms as $F\rightarrow F + dL$ under this 2-gauge/1-form shift symmetry, which can in fact change the monopole charge if $L$ has non-trivial periods
\begin{equation}
    \frac{1}{2\pi}\int_{S^2} dL\nonumber
\end{equation}
on the boundary 2-sphere $S^2\subset S^2\times l =\partial X'$. In order to absorb this ambiguity, we must force the 2-form $\Sigma$ to transform as
\begin{equation}
    \Sigma \rightarrow \Sigma + dL.\label{eq:ab2gau}
\end{equation}
{This} in essence makes the monopole charge into a gauge redundancy---different ``2-gauge sectors'' labeled by the periods of $L$ have different values of $q_m$.

On the other hand, recall that the $U(1)$-connection $A$, as well as the gauge parameter $\lambda=\lambda_0$, are all regular on the excised 4-manifold $X'$, and hence admit an extension into the whole of $X$. As such $F$---and hence $\Sigma$---is invariant under a 1-gauge transformation.

Thus the anomalous $U(1)$-gauge theory achieves a mixed 0-form/1-form $U(1)_1\times U(1)$ symmetry, governed by a trivial 2-group \begin{equation}
    \mathfrak{I}_{U(1)} = (t=\operatorname{id}:U(1)_1\rightarrow U(1),\rhd = 1).\label{trivu1}
\end{equation} 
{The} action $\rhd = 1$ at the group level implies $\rhd = 0$ at the algebra level, which allows us to write the transformation law  (\ref{eq:ab2gau}) as
\begin{equation}
    \Sigma \rightarrow \Sigma + dL + \lambda\rhd \Sigma = \Sigma + dL;\nonumber
\end{equation}
we have the \textit{{invariant}} 2-curvature and the higher Bianchi identity
\begin{equation}
    K = d\Sigma \rightarrow d\Sigma + d^2L = K,\qquad dK = 0.\nonumber
\end{equation}
{The} fake-flatness $F = t\Sigma = \Sigma$ then encodes the monopole condition  (\ref{eq:mono}), and the 2-curvature $K = \ast j_m\neq 0$ can be used to encode the monopole current $j_m$ without assuming a violation of Bianchi identity $dF\neq 0$.

%\medskip
What we have demonstrated is that, \textit{{to resolve an anomalous 0-form symmetry, we must introduce a 2-group structure with mixed 0- and 1-form symmetry}}. This is precisely the idea leveraged in \cite{Cordova:2018cvg} in order to implement the {\it {Green--Schwarz mechanism}} of anomaly cancellation in QFT (and more generally in \cite{Sati:2008eg,Fiorenza:2020iax}). We describe this procedure briefly in the following.

\subsection{2-Yang--Mills Theory} 
We now utilize the above concept of monopole anomaly resolution in order to study an \textit{{anomaly-free}} version of monopole electrodynamics. By anomaly-free, we mean that the $U(1)$-bundle $P\rightarrow X$ under consideration has trivial first Chern class. 

Such a bundle hosts no monopole anomaly, and the Bianchi identity $dF =0$ is satisfied everywhere. In order to introduce a monopole charge, we intend to design a 2-form $\Sigma$ with a non-trivial quantized period,
\begin{equation}
    0\neq q_m = \frac{1}{2\pi}\int_S \Sigma \in\bbZ,\qquad \Sigma=dL, \label{monoperiod}
\end{equation}
(where the 1-form $L$ has a branch cut as shown in Figure \ref{fig:monopole}) from an action principle. Recall that, by the clutching construction, the value of the monopole charge is fixed by the singularity structure of $L$.

The goal   is therefore  to construct an action of a 2-gauge theory that describes the electrodynamics of {regular} anomaly-free Maxwell's theory, { {as well as}} the monopole configuration of the 2-form connection $\Sigma=dL$. We begin by forming the manifestly invariant quantity $\mathcal{F}=F -\Sigma$ under (regular) 1-form shift symmetry {(The prime is to distinguish $L'$ from the $L$ chosen in  (49). Here we consider regular $L'$, free of branch-cuts.)} $A\rightarrow A+L'$, and write down the Abelian {\it {2-Yang--Mills theory}} \cite{Baez:2002highergauge}
\begin{equation}
    S_\text{2YM}[A,\Sigma] = \int_X\ast\mathcal{F}\wedge\mathcal{F}=\int_X \ast (F-\Sigma)\wedge (F-\Sigma),\label{eq:2ym}
\end{equation}
which has also appeared in the study of topological orders protected by subsystem symmetry \cite{Dubinkin:2020kxo}. By varying the 2-connection $\Sigma$, we obtain the EOM $\ast (F-\Sigma)=0$, which is nothing but fake-flatness $\mathcal{F}=F-\Sigma=0$. However, in order to have non-trivial monopoles charges, we { {must}} consider the off-shell configurations $F\neq \Sigma$. This is because fake-flatness kills the monopole $K=d\Sigma = dF=0$ by the Bianchi identity. Note that we do not expect any issue despite a violation of the fake flatness condition since the theory is Abelian.

\subsubsection{2-Yang--Mills Theory with Sources}
In order to have a non-trivial monopole configuration, we must therefore source the 2-connection $\Sigma$. We do something more general here and source both the 1- and 2-connections $A,\Sigma$ individually, with 1- and 2-form currents $j_e,J$, respectively. This inserts the following~terms
\begin{equation}\label{scur}
    S_\text{2cur}=\int_X \ast j_e\wedge A + \int_X \ast J \wedge\Sigma 
\end{equation}
into $S_\text{2YM}$. Intuitively, the 2-form current $J$ should be related in some way to the monopole current $\ast j_m$; indeed, upon a variation of $\Sigma$, the sourced action  (\ref{eq:2ym}) together with current contribution \eqref{scur} leads to the EOM
\begin{equation}
    \ast (F-\Sigma) = \ast J\implies d\Sigma=-dJ,\nonumber
\end{equation}
where we have used the Bianchi identity $dF=0$. By definition, the pure-gauge 2-connection $\Sigma=dL$ has quantized period given by the monopole charge
\begin{equation}
    q_m =\frac{1}{2\pi} \int_S \Sigma = \frac{1}{2\pi}\int_{S\times l}d\Sigma = -\frac{1}{2\pi}\int_{S\times l} dJ,\nonumber
\end{equation}
and combined with the definition  \eqref{monodef} leads to
\begin{equation}
    dJ = -\ast j_m.\nonumber
\end{equation}
{This} identifies the flux $dJ$ of the 2-form current $J$ with precisely the monopole current. This makes sense, as $J$ sources the 2-connection $\Sigma$ that introduces the monopole.

We emphasize here that the 2-Yang--Mills theory  \eqref{eq:2ym} is anomaly-free, meaning that it does not have any monopole currents $j_m$ as the Bianchi identity $dF=0$ is satisfied. Indeed, the main point of the construction is to source the monopole charge with a 2-form current $J$ { {without}} introducing anomalies.

\subsubsection{Higher Conservation Laws}
Now to derive the conservation laws of the currents $j_e,J$, we make gauge transformations on $S_\text{2YM}+S_\text{2cur}$. A (regular) 0-gauge transformation $A\rightarrow A+d\lambda$ leads to
\begin{equation}
    S_\text{2cur}\rightarrow S_\text{2cur} + \int_X \ast j_e\wedge d\lambda =S_\text{2cur} -\int_X \lambda (d\ast j_e),\nonumber
\end{equation}
which accounts for the conservation $d\ast j_e=0$ of the { {electric}} current. Suppose now we make the 1-form shift transformation $A\rightarrow A+L',\Sigma\rightarrow \Sigma + dL'$. The action $S_\text{2YM}$ remains invariant, but the sourcing terms acquire
\begin{equation}
    S_\text{2cur}\rightarrow S_\text{2cur}+\int_X \ast j_e \wedge L' + \ast J \wedge dL' = S_\text{2cur}+\int_X (\ast j_e - d\ast J) \wedge L',\nonumber
\end{equation}
which implies the {\it {2-conservation law}}
\begin{equation}
    d\ast J = \ast j_e.\label{eq:2noether}
\end{equation}
{This} implies that the conservation $d\ast J=0$ of the 2-form current $J$ occurs only if $j_e=0$---the 2-form current $J$ is conserved only if { {isolated charges are immobile}}, precisely like a dipole. Indeed,  \eqref{eq:2noether} is also known as a {{dipole}} conservation law \cite{Dubinkin:2020kxo}. This mobility restriction is similar to that for fractons \cite{Pretko:2017fbf,Slagle.96.195139}; there had been effort to describe (3+1)D fracton models in the continuum with { {foliated}} field theories \cite{Pretko:2020}, with similar mobility constraints for the charged currents.

\subsubsection{Charged Operators in 2-Yang--Mills Theory}\label{2ymoperators}
    We note here that generalized higher-form global symmetries had been analyzed in~\cite{Gaiotto:2014kfa}. There, a $p$-form global symmetry is described by a set of (unitary) charged operators $U_M$ localized on $(p+1)$-\textit{co}dimensional submanifolds $M\subset X$, together with its fusion rules $U_M$. This view had led to the fervent study of ``global (higher) categorical symmetries'' in recent times (see e.g., \cite{Wen:2019,Kong:2020}), which in fact generalizes our current 2-group case.
    
    To see this, we define for $1,2$-codimensional submanifolds $M,N\subset X$ the following $U(1)$-valued charged operators
    \begin{equation*}
        U_M(\alpha) = e^{i\alpha\int_M\ast j_e}\in U(1)_0,\qquad V_N(\beta) = e^{i\beta\int_N \ast J}\in U(1)_1,
    \end{equation*}
    where $\alpha,\beta\in \R$ are real parameters. We consider $U,V$ as a function on the set of codim-1,-2 submanifolds of $X$, respectively. Now let us define a map $\partial$ by $(\partial V)_M(\alpha) = V_{\partial M}(\alpha)$, which relates charges localized on $M$ (codim-1) to its boundary $\partial M$ (codim-2). We see that the 2-conservation law \eqref{eq:2noether} implies
    \begin{equation*}
        (\partial V)_M(\alpha) = e^{i\alpha\int_{\partial M}\ast J}= e^{i\alpha\int_M d\ast J} = e^{i\alpha\int_M\ast j_e} = U_M(\alpha),
    \end{equation*}
    which is precisely the fake-flatness condition \cite{Bullivant:2016clk} associated to the $t$-map of the 2-group \eqref{trivu1}. In other words, the mobility constraint for the currents translates to the statement that the 1-form operators $V$ lie on the boundary of the 0-form operators $U$. 

    \smallskip
    
    Generally, charged operators of a 2-gauge theory, such as \eqref{eq:2ym}, forms a (strict) 2-group~\cite{Baez:2002highergauge,Wockel2008Principal2A}: the groupoid unit translates to the unit codim-2 {operator} %MDPI: Please ensure all variables/values in the equation appear in the same format in the text (normal/italic/bold/subscript/superscript). Please check if the bold is unnecessary and can be removed. The following bold numbers are the same %Hank: confirm
 $\mathbf{1}_M$ with trivial 1-form charge on the boundary $\partial M$, and the two compositions on the 2-group \cite{baez2004}, denoted by $\circ,\otimes$, respectively, translates to two different ways of fusing these charged operators $(U,V)$. These composition laws are subject to the interchange relation (i.e., the Peiffer identity) \cite{baez2004,Porst2008Strict2A},
    \begin{equation}
        (\mathbf{1}_M\otimes V_{N'})\circ(V_N \otimes \mathbf{1}_{M'}) = (\mathbf{1}_M\circ V_N)\otimes (V_{N'}\circ\mathbf{1}_{M'}),\label{interchange}
    \end{equation}
    where $\partial M = N$ and similarly $\partial M' = N'$. This interchange relation is in fact a generalization of the Eckman-Hilton argument \cite{Gaiotto:2014kfa}: consider a 1-form symmetry $G$ described by the 2-group $G\rightarrow\ast$, where the group at degree-0 $\ast$ is trivial. The $t$-map must then also be trivial, whence the Peiffer identity \eqref{interchange} implies that $G$ must be Abelian.

    %\medskip

    The above construction can be understood as giving a map from the (differential graded) algebra formed by the currents $(j_e,J)$ to a certain algebra of observables---namely the 2-group of charged operators $(U,V)$---of the 2-Yang--Mills theory \eqref{eq:2ym}. This procedure has been studied in the general $L_\infty$-algebra context in \cite{costello_gwilliam_2016}, where the spacetime dependence of these currents/charged operators is described in terms of { {factorization algebras}} on $X$. We refer the reader to the aforementioned reference for details.
    
% \textbf{We want to show how the breaking of the Bianchi identity affect Noether theorem and how adding the extra fields, with an associated symmetry allows to correct it.  }

\subsection{4D Topological Orders: Quasistring Defects and Surface Linking} \label{strpt}
We have seen in the monopole case above that certain curvature anomalies can be used to represent topological invariants of $X$. There, the first Chern class $c_1\in H^2(X,\bbZ)$ classifying complex line bundles on $X$ can be represented by the curvature 2-form $F$ through Chern--Weil theory \cite{book-algtop,book-charclass}. 

In this section, we would like to focus on the {{torsion}} topological invariants, particularly the Stiefel--Whitney classes $w\in H^\bullet(X,\bbZ_2)$ of the tangent bundle $TX \rightarrow X$. They  classify the { {framing}} of $X$ \cite{book-algtop,book-charclass}; see also {Appendix} %MDPI: We changed `Section' to `Appendix'. Please confirm this revision. %Hank: confirm
 \ref{framedquasi}. Our goal in this section is to leverage the structures of a {{finite}} skeletal 2-group $\mathcal{G}=(V\xrightarrow{0} N,\rhd)$ in order to insert torsional topological anomalies, such as these Stiefel--Whitney classes. 
% Since we shall only be interested in the topological defects of the theory, we assume our structure 2-group  is skeletal with $t=0$. %Indeed, any 2-group is elementary equivalent \cite{Wag:2006} to its Hoang data \cite{Ang2018,Brown}, which is a skeletal 2-group \cite{Kapustin:2013uxa}.

\subsubsection{Insertion of Topological Defects}
Let $X$ be a framed 5-manifold with boundary $Y=\partial X$. Given a skeletal 2-group $\mathfrak{G}=(0:V\rightarrow N,\rhd)$, its associated 2-gauge theory encodes the following fake-flatness and 2-flatness conditions,
\begin{equation}
    \mathcal{F} = F = 0,\qquad \mathcal{G} = K -\kappa(A) = 0,\nonumber
\end{equation}
such that the 1- and 2-connections $(A,\Sigma)$ are in a sense ``decoupled''. Excitations of the theory can be inserted separately by modifying these EOMs \cite{Zhu:2019}
\begin{equation}
    \mathcal{F} = f_2 \in C^2(X,\mathbb{R})\otimes\mathfrak{n},\qquad \mathcal{G} = g_3\in C^3(X,\mathbb{R})\otimes V,\nonumber
\end{equation}
where $C^\bullet(X,\mathbb{R})$ denote the complex of $\mathbb{R}$-valued differential cochains on $X$. The worldvolumes  of the excitations in $Y$ are determined by the restrictions of the cochains $f_2,g_3$ to $Y$ via {Poincar{\' e} duality} \cite{book-algtop}
\begin{equation}
    \operatorname{PD}:C^n(Y,\mathbb{R})\rightarrow C_{4-n}(Y,\mathbb{R}).\nonumber
\end{equation}
{More} explicitly, if $\iota:Y\hookrightarrow X$ is the inclusion of the 4D boundary, then $\operatorname{PD}(\iota^*g_3)$ is a 1-cycle (a worldline) on $Y$ \cite{Dubinkin:2020kxo}. Similarly, the 2-cycle $\operatorname{PD}(\iota^* f_2)$ can be interpreted as the worldsheet of a string-like excitation \cite{Zhu:2019,Thorngren2015}.

One of the key points demonstrated in Section \ref{QED-anom} is that the characteristic classes contribute as ``anomalous'' topological excitations, or {defects}, of the theory. As such our goal is to construct a 2-gauge structure hosting the Stiefel--Whitney classes $w_2,w_3$ as defects. Following \cite{Kapustin:2013uxa,Zhu:2019}, we work directly with flat discrete 2-connections that exhibit the Stiefel--Whitney classes as topological defects.

%\medskip

The construction of gauge fields that capture topological defects gives rise to an invertible topological quantum field theory (TQFT), such as Yetter theory \cite{Bullivant:2016clk} or Dijkgraaf--Witten theory \cite{Kapustin:2013uxa,Bullivant:2021}. These objects make an appearance in high-energy physics \cite{Thorngren:2021,Benini_2019} and condensed matter physics \cite{Freed_2021,Kong:2020}, as it is common lore \cite{Freed:2014} that anomalies in QFTs are in a very general sense { {topological}}. 

\begin{remark}
One can also expect that topological features are  relevant in quantum gravity. Indeed, if we accept that quantum gravitational fluctuations allow for topology change, then one way to keep track of them is to use a 2-group structure, through a  2-connections with non-trivial topological configurations, such as those that we shall construct below.  Another interesting direction   would be to formulate the 2-group field theory \cite{Girelli:2022bdf} which generates topological amplitudes to non-trivial (i.e., non-zero Postnikov class) crossed-modules. We shall leave this to a future work.
\end{remark}

\subsubsection{Discrete Flat 2-Connections} 
Recall that flat $G$-connections on $Y$ can be uniquely assigned, up to homotopy, through the choice of a classifying map $f:X\rightarrow BG$, where $BG$ is the classifying space of $G$ \cite{book-algtop}. A similar situation occurs for 2-groups; in the following, we shall focus on the discrete skeletal 2-group $\mathcal{D}(\bbZ_2)\equiv (\bbZ_2\xrightarrow{0}\bbZ_2,\rhd=0)$.

The classifying space $B\mathcal{D}(\bbZ_2)$ can be constructed from a Postnikov tower \cite{Kapustin:2013uxa,Zhu:2019,Brown}, such that a discrete flat $\mathcal{D}(\bbZ_2)$-connection can be assigned onto $X$ through the choice of a classifying map \cite{Kapustin:2013uxa,Zhu:2019,Mackaay:uo}
\begin{equation}
    f: X\rightarrow B\mathcal{D}(\bbZ_2),\nonumber
\end{equation}
up to homotopy. Now a simple computation in group cohomology \cite{book-groupcohomology} yields a non-trivial generator $[\kappa]\in H^3(\bbZ_2,\bbZ_2)\cong\bbZ_2$, which allows us to define
\begin{equation}
    w_3=f^*[\kappa]\label{3sw}
\end{equation}
as the pullback by $f$. This equation uniquely determines $f$ up to homotopy \cite{Kapustin:2013uxa}.

% Equipped with  \eqref{3sw}, t
The discrete 2-gauge structure we aim for should then have as EOMs on the boundary $Y=\partial X$:
%through this construction achieves the EOMs on the boundary $Y=\partial X$:
\begin{equation}
    F =dA=0,\qquad K =f^*\kappa = \kappa(A)= w_3,\nonumber
\end{equation}
where $(A,\Sigma)$ denotes a flat $\mathcal{D}(\bbZ_2)$-connection defined by the classifying map $f$. We have thus defined a discrete 2-connection that realizes the third Stiefel--Whitney class $w_3$ as a 2-curvature anomaly.

These EOMs can be recovered from a 2-BF type action based on the 2-group $\mathcal{D}(\bbZ_2)$. For this, we introduce  the Lagrange multipliers $b\in C^2(Y,\bbZ_2),c\in C^1(Y,\bbZ_2)$ on the boundary $Y=\partial X$, and we recover a 2-BF action similar to \eqref{eq:2bf},
\begin{equation}
    S_{\text{2BF}} = \frac{1}{2}\int_Y b\cup F + c\cup (K-f^*\kappa) = \frac{1}{2}\int_Y b\cup dA + c\cup (d\Sigma - w_3), \nonumber
\end{equation}
where $\cup$ is the cup product on $\bbZ_2$-valued cochains \cite{book-algtop}.

\subsubsection{Twistings of 2-Drinfel'd Double Symmetries}
We now would like to introduce the second Stiefel--Whitney class $w_2$ in a similar way, but this time associated to the 1-curvature $F$. One cannot do that directly, since the 2-group $\mathcal{D}$ is skeletal and so by construction we have $F=0$. In other words, one cannot introduce $w_2$ as a curvature defect without breaking the 2-gauge symmetry.

We can, however, insert it as a { {dual}} 1-curvature anomaly. That is, as a source from the 1-gauge  theory inherited from the dual $\mathcal{D}^*$ gauge theory.  By appending a term
\begin{equation}
    S_{w_2}=\frac{1}{2}\int_Y w_2 \cup \Sigma \nonumber
\end{equation}
to the 2-BF theory $S_\text{2BF}$, we achieve the EOM upon a variation of $\Sigma$:
\begin{equation}
    \tilde F = dc = w_2,\nonumber
\end{equation}
which can be interpreted as an anomaly in the dual curvature  \eqref{eq:dual}. In this way, each of the $\mathcal{D},\mathcal{D}^*$ sectors {{separately}} supports the topological defect $w_3,w_2$, and these cohomology classes contribute to twist the symmetry {(Technically speaking, these twists occur in the \textit{2-groupoid algebra} $\mathbb{C}[\mathcal{D}]$} \cite{Bullivant:2021}; {for instance, we have $\mathbb{C}[\mathcal{D}^{w_3}]=\mathbb{C}^{\kappa}[\mathcal{D}]$ where $w_3=f^*\kappa$ arises from the Postnikov class $\kappa$.})
$$\mathcal{D}\times\mathcal{D}^*\to\mathcal{D}^{w_3}\times(\mathcal{D}^*)^{w_2}\,.$$ 
{We} have effectively replaced the original 2-group symmetry $\mathcal{D}$ by a fibred product $\mathcal{D}^{w_3}\times_{\mathbb{Z}_2}(\mathcal{D}^*)^{w_2}$ of the twisted categorical symmetries. Note we no longer have a 2-Drinfel'd double!

After going on-shell of $F=0$, the boundary action then reads \cite{Thorngren2015}
\begin{equation}
    S_\text{2BF}+S_{w_2} \sim S_\partial=\frac{1}{2}\int_Y c\cup (d\Sigma-w_3) + w_2\cup \Sigma.\label{bdytop}
\end{equation}
{Due} to the closure $dw_n=0$ of the Stiefel--Whitney classes, we have
\begin{equation}
    \frac{1}{2}d\left(c\cup (d\Sigma-w_3) + w_2\cup \Sigma\right)=\frac{1}{2}(dc \cup w_3 + w_2\cup d\Sigma) \sim w_2\cup w_3,\nonumber
\end{equation}
where we have used the EOMs for $c$ and $\Sigma$. This means that $S_\partial$ can be interpreted as the boundary action of a bulk 5D symmetry-protected topological (SPT) phase
\begin{equation}
    \mathcal{C}_\text{5d} = \int_X w_2\cup w_3, \label{5dtopord}
\end{equation}
protected by the global $\mathbb{Z}_2$-symmetry. Conversely, we may begin with the action \eqref{5dtopord}, and interpret the cochains $c,\Sigma$ and their EOMs as trivializations for $w_2,w_3$ on the boundary $Y=\partial X$ \cite{JuvenWang,Thorngren2015}.

\begin{remark}
A topological order is {\it symmetry-protected} if it describes a gapped phase that is topologically trivial when the symmetry is ignored \cite{XieChen:2013}. In general, an invertible topological order can be interpreted as being hosted on the boundary of a bulk symmetry protected topological (SPT) order, through the mechanism of {\it anomaly inflow/resolution} \cite{Witten:2019,JuvenWang}. This is the {\it holographic bulk-boundary correspondence}, which states that a $d$-dimensional (possibly anomalous) topological order $\mathcal{C}$ determines uniquely an {\it anomaly-free} order $\mathcal{D}$ in $(d+1)$-dimensions \cite{Kong:2020}.
\end{remark}

Let us organize the above discussion in the following.
\begin{proposition}
    The 4D boundary theory $S_\partial$ of the 5D SPT $w_2w_3$ does \textit{not} admit a 2-BF theory description. It does, however, have a symmetry given by a quotient of the (non-double) twisted symmetry $\mathcal{D}^{w_3}\times_{\mathbb{Z}_2}(\mathcal{D}^*)^{w_2}$.
\end{proposition}

{Since} %MDPI: We removed the \noindent format, please confirm %Hank: confirm
 a 2-Drinfel'd double symmetry gives rise to a Drinfel'd centre fusion 2-category~\cite{Chen2z:2023}, this statement can be considered as another argument for the fact that this anomalous 4D boundary theory $S_\partial$ is not a Drinfel'd centre \cite{Johnson-Freyd:2020,Johnson_Freyd_2023}. Here, we provide a simple and direct field theoretic argument without deep higher-categorical algebra.

The above general formalism of using a 2-group structure to capture topological defects has appeared in \cite{Kapustin:2013uxa,Zhu:2019}. The action  \eqref{bdytop} and topological order \eqref{5dtopord} specifically has also been studied in \cite{Thorngren2015}, and we shall follow this reference and provide a brief summary of its interesting properties in the remainder of this section. {{The new insight we provided here}}  is the understanding that the EOM $dc=w_2$ appears as a {{dual}} EOM for the associated 2-$BF$ theory, which implies that the boundary action  \eqref{bdytop} is characterized by an underlying 2-Drinfel'd double associated to $\mathcal{D}(\bbZ_2)$; see { {Remark} \ref{2ddsymm}}.

\section{Gauging the 2-Gauge}\label{sec:gaug2gau}
Recall from Section \ref{sec:gau1gau} that we noticed that the curvature $F$ was invariant under a shift $\alpha$, a closed form  with value in the centre of $\g$. Generalizing this to an arbitrary shift led to the  ``gauging the 1-gauge''. This approach extends to the 2-gauge case.   Indeed, given we have a covariantly closed 2-form $\sigma$ such that $d_A\sigma=0$, we see that the 2-curvature $K=d_A\Sigma$ is (strongly) invariant under the shift $\Sigma\rightarrow \Sigma+\sigma$. 

In the following, we shall gauge this global symmetry by taking this shift to be an {{arbitrary}} 2-form $\sigma\in\Omega^2(X)\otimes\h$. In this sense, we will make the 2-curvature {{gauge datum}}. 

\subsection{From Shifting the 2-Connection to a Lie Algebra 2-Crossed-Module}\label{sec:2-con-shift}

\subsubsection{Shifting the 2-Connection}
Consider a 2-connection shift $\Sigma\rightarrow \Sigma+\sigma$ by an arbitrary 2-form $\sigma$. The 2-curvature $K$ then transforms accordingly
\begin{equation}
    K \rightarrow K'=d_A\Sigma + d_A\sigma,\nonumber
\end{equation}
hence the 3-curvature $K\rightarrow K'\neq K$ fails to be invariant, even on-shell of the fake-flatness condition. To remedy this, we introduce a {\it {3-form}} gauge field by
\begin{equation}
    \Gamma\equiv K'-K = d_A\sigma.\label{eq:3gau}
\end{equation}
{Following} the same reasoning as previously, the shift we are performing does not have to come from the same algebra $\Sigma$ is valued in. We consider a 2-form $\mathcal{L}\in\Omega^2(X)\otimes\mathfrak{i}$ valued in another Lie algebra $\mathfrak{i}=\operatorname{Lie}I$, and replace (\ref{eq:3gau}) by a ``pure gauge'' connection  3-form $\Gamma=d_A\mathcal{L}$. 

Repeating the same steps as before, we may expect the existence of another crossed-module $\mathfrak{S}=(t':\mathfrak{i}\rightarrow \mathfrak{h},\bar\rhd)$ such that the shift transformation on the 2-connection
\begin{equation}
    \Sigma\rightarrow \Sigma+ t'\mathcal{L},\qquad t'\mathcal{L} = \sigma\nonumber
\end{equation}
becomes a ``3-gauge transformation'' parameterized by $\mathcal{L}\in\Omega^3(X)\otimes\mathfrak{i}$. This yields yet another invariant fake curvature quantity, called the {2-fake curvature}
\begin{eqnarray}
    \mathcal{G}\equiv K-t'(\Gamma),\nonumber
\end{eqnarray}
and the associated 2-fake-flatness condition $K=t'\Gamma$. Notice once again that the 2-curvature anomaly $K=\kappa$ can now be absorbed as a part of the 3-connection $\Gamma$---we shall make this statement precise in Section \ref{3anom}.

Since $K$ is valued in $\operatorname{ker}t$ on-shell of the fake-flatness condition $\mathcal{F}=0$, enforcing also the 2-fake-flatness condition implies $\operatorname{ker}t = \operatorname{im} t'$, i.e., $t\circ t'=0$. We thus obtain an exact Lie algebra complex
\begin{equation}
    \mathscr{G}: \mathfrak{i}\xrightarrow{t'}\mathfrak{h}\xrightarrow{t}\mathfrak{g},\label{eq:3comp}
\end{equation}
for which $t':\mathfrak{i}\rightarrow \h$ is a crossed-module. Here, $\mathfrak{g}$ acts on both $\mathfrak{h}$ and $\mathfrak{i}$, denoted respectively by $\rhd$ and $\rhd'$, under which the maps $t,t'$ are both $\mathfrak{g}$-equivariant---meaning that the following diagram commutes
\begin{equation}
\begin{tikzcd}
\mathfrak{i} \arrow[r, "t'"] \arrow[d]                                     & \mathfrak{h} \arrow[r, "t"]  \arrow[d]                                   & \mathfrak{g} \arrow[d] \arrow[ld, "\rhd"] \arrow[lld, "\rhd'",description] \\
\operatorname{Inn}\mathfrak{i} \arrow[r, "t'"] & \operatorname{Inn}\mathfrak{h} \arrow[r, "t"]  & \operatorname{Inn}\mathfrak{g}  
\end{tikzcd},\nonumber
\end{equation} 
which is a generalization of \eqref{equivdiag}. Here, $\operatorname{Inn}\g\subset\operatorname{Der}\g$ denotes the Lie algebra of inner derivations. 

\subsubsection{3-Curvature and 3-Bianchi Identity}

Let us define the 3-curvature $T=d_A\Gamma$ in the na{\" i}ve way. With the exactness $\operatorname{im}t'=\operatorname{ket}t$ of the complex  \eqref{eq:3comp} in mind, as well as on-shell of the 2-fake-flatness condition $t'\Gamma=K$, we see that
\begin{equation}
    t'(T) = d_A(t'(\Gamma)) = d_AK = ((t\Sigma)\wedge^\rhd \Sigma)|_{\operatorname{ker}t}.\nonumber
\end{equation}
{If} the Peiffer identity for the part $t:\mathfrak{h}\rightarrow\mathfrak{g}$ of  \eqref{eq:3comp} holds, then this term coincides with $[\Sigma\wedge\Sigma]|_{\operatorname{ker}t}$ as computed in  \eqref{eq:1biananom}, which vanishes. This is the 2-Bianchi identity.

However, we can also consider the case where the Peiffer identity does \textit{{not}} hold, so that  $t:\mathfrak{h}\rightarrow\mathfrak{g}$ is no longer necessarily a crossed-module, but merely a {\it {precrossed-module}}. This means that only the first Peiffer condition (equivariance) is satisfied, and $(ty)\rhd y'$ does not coincide with the bracket $[y,y']$ on $\mathfrak{h}$ as $(t\cdot)\rhd \cdot$ is in general not skew-symmetric {(Indeed, the skew-symmetry of this quantity is an axiom in the Lie 2-algebra formulation that} {\it {defines}} {the Lie bracket on $\mathfrak{h}$} \cite{Bai_2013}.).

To treat the term $(t\Sigma)\wedge^\rhd\Sigma|_{\operatorname{ker}t}$, we invoke the exactness of the complex  \eqref{eq:3comp} to write it as the image of some quantity $\{\Sigma\wedge\Sigma\}_\text{Pf}$ under $t'$. This quantity is defined by the {\it {Peiffer lifting map}} $\{\cdot,\cdot\}_\text{Pf}:\mathfrak{h}^{2\otimes}\rightarrow \mathfrak{i}$, which satisfies
\begin{equation}
    t'\{y,y'\}_\text{Pf} = -((ty')\rhd y)|_{\operatorname{ker}t},\qquad \forall y,y'\in\h.\label{liftcond}
\end{equation}
{In} other words, $\{\cdot,\cdot\}_\text{Pf}$ lifts $((t\cdot)\rhd \cdot)|_{\ker t}$ along $t'$ up to $\mathfrak{i}$. If we wish for the 3-curvature to be valued in $\operatorname{ker}t'\subset\mathfrak{i}$, similar to how the 2-curvature $K$ is valued in $\operatorname{ker}t\subset\mathfrak{h}$, we must replace $T$ by the {\it {modified 3-curvature}} $\mathcal{H}$,
\begin{equation}
    \mathcal{H} = d_A\Gamma + Q_\Sigma,\qquad Q_\Sigma = \{\Sigma\wedge\Sigma\}_\text{Pf}.\label{eq:3curv}
\end{equation}
{This} extra quadratic term $Q_\Sigma$ is an artifact of forgoing the Peiffer identity for $t:\h\rightarrow \g$ but, importantly, {should \textit{{not}} itself be considered an anomaly}. We shall elaborate more on this in Section \ref{3anom}.

\paragraph{{{3-Bianchi identity.} %MDPI: Please confirm if the bold is unnecessary and can be removed. And please check if it should be level 4 heading. %Hank: confirm, please remove the boldface
}}
Given the modified 3-curvature \eqref{eq:3curv}, we can assess what properties we should consider to have if we want to have a 3-Bianchi identity to hold (on shell of the fake-flatness condition $t\Sigma=F$), 
\begin{equation}
    d_A\mathcal{H} = d_A(T+Q_\Sigma) =0. \label{eq:3bian}
\end{equation}
{We}  compute, assuming that there would be % $\Sigma_{\ker} t$, so that there is 
no violation of 1-Bianchi identity,
% \begin{eqnarray}
%     d_A\mathcal{H} &=& F\wedge^{\rhd'}\Gamma + \{d_A\Sigma\wedge \Sigma\}_\text{Pf}+\{\Sigma\wedge d_A\Sigma\}_\text{Pf} \nonumber \\
%     &=& -\{t'\Gamma\wedge \Sigma\}_\text{Pf} - \{\Sigma\wedge t'\Gamma\}_\text{Pf} + \{K\wedge \Sigma\}_\text{Pf} + \{\Sigma\wedge K\}_\text{Pf} \nonumber \\
%     &=& \{(K-t'\Gamma)\wedge \Sigma\}_\text{Pf} + \{\Sigma\wedge (K-t'\Gamma)\}_\text{Pf}, \nonumber  
% \end{eqnarray}
\begin{eqnarray}
    d_A\mathcal{H} &=& F\wedge^{\rhd'}\Gamma + \{K\wedge \Sigma\}_\text{Pf} + \{\Sigma\wedge K\}_\text{Pf} \nonumber \\
    &=& t\Sigma\wedge^{\rhd'}\Gamma + \{K\wedge \Sigma\}_\text{Pf} + \{\Sigma\wedge K\}_\text{Pf} \nonumber
    % \\
    % t'\{\Sigma,\Gamma\}_\text{Pf}++ \{K\wedge \Sigma\}_\text{Pf} + \{\Sigma\wedge K\}_\text{Pf}
\end{eqnarray}
% where we used the lifting condition \eqref{liftcond}
% Having the relation
% \begin{equation}
%   t\Sigma\wedge^{\rhd'}\Gamma   -\{t'\Gamma\wedge \Sigma\}_\text{Pf} - \{\Sigma\wedge t'\Gamma\}_\text{Pf}=0,
% \end{equation}
% which in fact vanishes on-shell of the 2-fake-flatness condition $t'\Gamma=K$. This can be interpreted as a {\bf 3-Bianchi identity}
% in analogy with the 2-Bianchi identity  \eqref{2-bianch}. 
{Demanding} that the last quantity is zero, on-shell of the 2-flatness condition, imposes a relation between $\{\cdot,\cdot\}_\text{Pf}$, $t$, $t'$ and $\rhd'$, which appears when we consider a  Lie algebra 2-crossed-module. Such relation is one of the \textit{{2-Peiffer conditions}}, which we shall see below.

\begin{remark}
It would  be possible to investigate the notion of a ``weak 2-crossed-module'' by violating the 1-Bianchi identity. This can be accomplished by relaxing the Jacobi identity on $\mathfrak{g}$ in the definition of a 2-crossed-module. We shall not delve on this here, however, as this requires much more elaboration.
\end{remark}

\subsubsection{Lie Algebra 2-Crossed-Module}
We are now ready to define what a 2-crossed-module is. Let $\g,\h,\mathfrak{i}$ be Lie algebras.
\begin{definition}[\cite{Radenkovic:2019qme}]\label{def:2crmod} 
The 3-term algebra complex $\mathscr{G}$ in \eqref{eq:3comp}, equipped with the bilinear map $\{\cdot,\cdot\}_\text{Pf}: \h^{\otimes 2}\rightarrow \mathfrak{i}$, is a { {2-crossed-module}} if and only if
\begin{itemize}
    \item the action $\mathfrak{h}\bar\rhd \mathfrak{i}$ defined by $y\bar \rhd z\equiv -\{t'z,y\}_\text{Pf}$ makes $\mathfrak{S}=(t':\mathfrak{i}\rightarrow\mathfrak{h},\bar\rhd)$ into a crossed-module,
    \item the {{2-Peiffer conditions}} are satisfied
    \begin{eqnarray}
        \{t'z,y\}_\text{Pf} + \{y,t'z\}_\text{Pf} &=&-ty \rhd' z,\nonumber \\
        {[z,z']} &=&\{t' z,t' z'\}_\text{Pf},\nonumber
    \end{eqnarray}
    for all $y\in\mathfrak{h}$ and $z,z'\in\mathfrak{i}$,
    \item the {{3-Jacobi identities}} is satisfied
    \begin{eqnarray}
         \{y,[y',y'']\}_\text{Pf}&=&\{( y,y')_\text{Pf},y''\}_\text{Pf} - \{( y,y'')_\text{Pf},y'\}_\text{Pf},\nonumber\\
        \{[y,y'],y''\}_\text{Pf} &=& \{ty\rhd y',y''\}_\text{Pf} - \{ty'\rhd y,y''\}_\text{Pf} \nonumber \\
        &\qquad& + \{y',( y,y'')_\text{Pf}\}_\text{Pf} - \{y,( y',y'')_\text{Pf}\}_\text{Pf}\nonumber
    \end{eqnarray} 
    for all $y,y',y'' \in \mathfrak{h}$, where  $(\cdot,\cdot)_\text{Pf}\equiv t'\{\cdot,\cdot\}_\text{Pf}$ is the $\operatorname{im}t'=\operatorname{ker}t$-valued {Peiffer pairing}, and
    \item the lifting condition \eqref{liftcond}, is satisfied. Both $\{\cdot,\cdot\}_\text{Pf}$ and $(\cdot,\cdot)_\text{Pf}$ are $\g$-equivariant.
\end{itemize}

{Note} %MDPI: We removed the \noindent format, please confirm %Hank: please keep the \noindent here
 $\{\cdot,\cdot\}_\text{Pf}$ is in general \textit{not} skew-symmetric.
\end{definition}

{This} %MDPI: We removed the \noindent format, please confirm %Hank: confirm
 2-crossed-module  serves as the gauge principle for a principal 3-group bundle $\mathscr{P}\rightarrow X$, and it was proposed \cite{Radenkovic:2019qme} that the standard model has such a 3-gauge structure. 

In this case, we have the 1-gauge connection $A$ with value in $\g$, the 2-connection $\Sigma$ with value in $\h$ and the 3-connection $\Gamma$ with value in $\mathfrak{i}$. The modified 3-curvature satisfies the 3-Bianchi identity thanks to the  first of the 2-Peiffer conditions. We have indeed that 
\begin{equation}
    t\Sigma\wedge^{\rhd'}\Gamma = -\{t'\Gamma,\Sigma\}_\text{Pf} - \{\Sigma,t'\Gamma\}_\text{Pf} \nonumber,
\end{equation}
which together with the 2-fake flatness insures that the 3-Bianchi is satisfied. We now study the gauge transformation structure.

\subsection{Gauge Transformations}\label{1-2-3gau}
We now turn to the 3-gauge transformations on the fields $(A,\Sigma,\Gamma)$ in question. Let us fix the notation 
\begin{equation}
    \lambda\in\Omega^0(X)\otimes\g,\qquad L\in\Omega^1(X)\otimes \h,\qquad \mathcal{L}\in\Omega^2(X)\otimes \mathfrak{i}\nonumber
\end{equation}
for the 0-, 1- and 2-form gauge parameters in the theory. We shall derive the 3-gauge transformation rules, under the principle that the curvature quantities
\begin{equation}
    \mathcal{F} = F - t\Sigma,\qquad \mathcal{G} = K - t'\Gamma,\qquad \mathcal{H} = d_A\Gamma+Q_\Sigma\nonumber
\end{equation}
transform covariantly. We shall recover the results given in \cite{Radenkovic:2021fqq}.

We begin by assuming for simplicity that the 2-gauge sector, involving the fields $(A,\Sigma)$, should retain the same transformation laws under the 2-gauge parameters $(\lambda,L)$ as we have derived in Section \ref{1-2gauge}.

\subsubsection{1-Gauge Transformations}
We utilize the action of $\mathfrak{g}$ on the two other Lie algebras $\mathfrak{h},\mathfrak{i}$. Since in the 2-gauge context, $K=d_A\Sigma\rightarrow K - \lambda\rhd K$ is covariant, the 3-form $\Gamma$ must also transform covariantly,
\begin{equation}
    t'\Gamma\rightarrow t'\Gamma - \lambda \rhd t'\Gamma = t'\Gamma - t'(\lambda\rhd' \Gamma)\implies \Gamma\rightarrow \Gamma-\lambda\rhd' \Gamma,\nonumber
\end{equation}
in order to achieve the covariance of the 2-fake-curvature
\begin{equation}
    \mathcal{G}=K - t'\Gamma \rightarrow\mathcal{G}^\lambda = \mathcal{G}-\lambda\rhd \mathcal{G}.\nonumber
\end{equation}
{Moreover}, due to the $\mathfrak{g}$-equivariance of the Peiffer lifting map $\{\cdot,\cdot\}_\text{Pf}$, the modified 3-curvature \eqref{eq:3curv} also achieves the covariant transformation
\begin{equation}
    \mathcal{H} = d_A\Gamma + Q_\Sigma \rightarrow \mathcal{H}^\lambda=\mathcal{H} - \lambda\rhd' \mathcal{H},\label{3curv1gau}  
\end{equation}
as desired.
%These transformations, as well as the form (\ref{eq:3curv}) and covariance of the 3-curvature $T$ are in fact sufficient to determine the rest of the 3-gauge structure via the 2-Peiffer conditions.

%\medskip 

\subsubsection{2-Gauge Transformations} 
We utilize \eqref{gtff} for the 2-gauge transformation law of $K$, with the caveat that the Peiffer identity for $t:\h\rightarrow\g$ is no longer necessarily satisfied. As such, we write the $[L\wedge L]$ term in the 2-gauge transformation $\Sigma^L$ as $tL\wedge L$. This yields
\begin{equation}
    K\rightarrow K^L = K + tL\wedge^\rhd \Sigma^L + F\wedge^\rhd L.\nonumber
\end{equation}
{This} indicates the transformation property
\begin{equation}
    t'\Gamma\rightarrow  t'\Gamma -tL\wedge^\rhd \Sigma^L -t\Sigma\wedge^\rhd L,\nonumber
\end{equation}
such that we achieve the desired covariance
\begin{equation}
    \mathcal{G}=K - t'\Gamma\rightarrow \mathcal{G}^L = \mathcal{G} - \mathcal{F}\wedge^\rhd L.\label{2fakecurv2gau}
\end{equation}
{Using} the lifting condition \eqref{liftcond} for the Peiffer pairing $(\cdot,\cdot)_\text{Pf}=t'\{\cdot,\cdot\}_\text{Pf}$, we have
\begin{equation}
    -tL\wedge^\rhd \Sigma^L  = (\Sigma^L\wedge L)_\text{Pf} = t'\{\Sigma^L\wedge L\}_\text{Pf},\qquad -t\Sigma\wedge^\rhd L = (L\wedge \Sigma)_\text{Pf} = t'\{L\wedge \Sigma\}_\text{Pf},\nonumber
\end{equation}
which allows us to deduce the 2-gauge transformation 
\begin{equation}
    \Gamma\rightarrow\Gamma^L = \Gamma+\{\Sigma^L\wedge L\}_\text{Pf} + \{L\wedge \Sigma\}_\text{Pf}\nonumber
\end{equation}
for $\Gamma$, provided the 1-Bianchi identity is satisfied (i.e., $K$ is valued in $\operatorname{ker}t = \operatorname{im}t'$).

A direct computation with this 2-gauge transformation law for $\Sigma$ and $\Gamma$ then produces~\cite{Radenkovic:2021fqq}
\begin{equation}
    \mathcal{H}\rightarrow \mathcal{H}^L = \mathcal{H} + \{\mathcal{G}^L\wedge L\}_\text{Pf} +\{L\wedge\mathcal{G}\}_\text{Pf}\sim \mathcal{H}\label{3curv2gau},\nonumber
\end{equation}
which is indeed  invariant on-shell of the 1- and 2-fake-flatness conditions $\mathcal{F},\mathcal{G}=0$.

%\medskip 

\subsubsection{3-Gauge Transformations} We expect that a   3-gauge transformation should be parameterized by a 2-form, $\mathcal{L}$, in $\mathfrak{i}$. Hence we naturally posit that $A$ is  invariant, $A\rightarrow A^\mathcal{L} = A$, and hence so is the 1-curvature $F\rightarrow F^\mathcal{L}=F$.

For the remaining fields $\Sigma,\Gamma$, we follow the above gauging the gauge argument and induce a shift transformation
\begin{equation}
    \Sigma\rightarrow\Sigma^\mathcal{L}=\Sigma+ t'\mathcal{L},\qquad \Gamma\rightarrow\Gamma^\mathcal{L}=\Gamma+d_A\mathcal{L}.\nonumber
\end{equation}
{Recall} that the product of a $\mathfrak{g}$-valued form with a $\mathfrak{i}$-valued form is performed through the action $\rhd'$. 
%Indeed, a pure 3-gauge $\Gamma=d_A\mathcal{L}$ takes the form of  (\ref{eq:3gau}).

We now compute that the 2-curvature transforms as (note the $\mathfrak{g}$-equivariance of $t'$)
\begin{equation}
    K \rightarrow K^{\mathcal{L}} = K + d_At'\mathcal{L}= K+t'd_A\mathcal{L},\nonumber
\end{equation}
which gives rise to the { {invariance}} of the 2-fake-curvature
\begin{equation}
    \mathcal{G} \rightarrow \mathcal{G}^\mathcal{L} = \mathcal{G}.\label{2curv3gau}
\end{equation}
{Furthermore}, as $tt'=0$ and $F$ is unchanged, we see that the fake-curvature $\mathcal{F}=F-t\Sigma$ is in fact also invariant under the 3-gauge.

Now performing a 3-gauge transformation on the modified 3-curvature $\mathcal{H}$ gives
\begin{equation}
    \mathcal{H}\rightarrow \mathcal{H}^{\mathcal{L}} = \mathcal{H} + F\wedge^{\rhd'}\mathcal{L} + \{t'\mathcal{L}\wedge \Sigma\}_\text{Pf} + \{\Sigma\wedge t'\mathcal{L}\}_\text{Pf} + \{t'\mathcal{L}\wedge t'\mathcal{L}\}_\text{Pf},\nonumber
\end{equation}
for which we may employ the 2-Peiffer conditions (recall $\mathcal{L}$ is a 2-form)
\begin{equation}
    \{t'\mathcal{L}\wedge \Sigma\}_\text{Pf} + \{\Sigma\wedge t'\mathcal{L}\}_\text{Pf} = -t\Sigma\wedge^{\rhd'}\mathcal{L},\qquad  \{t'\mathcal{L}\wedge t'\mathcal{L}\}_\text{Pf} = [\mathcal{L}\wedge\mathcal{L}] =0\nonumber
\end{equation}
to deduce the covariance
\begin{equation}
    \mathcal{H}\rightarrow \mathcal{H}^{\mathcal{L}} = \mathcal{H}+ \mathcal{F}\wedge^{\rhd'} \mathcal{L} \sim \mathcal{H}.\label{3curv3gau}
\end{equation}
{The} modified 3-curvature $\mathcal{H}$ is thus invariant on-shell of the fake-flatness condition \mbox{$\mathcal{F}=0$}~\cite{Radenkovic:2021fqq} (which we recall is preserved by the exactness $tt'=0$ of the complex  (\ref{eq:3comp})).
 
In summary, we have the following 3-gauge transformations
\begin{eqnarray}
     \lambda&:& \begin{cases}A\rightarrow A^\lambda=A+d_A\lambda \\ \Sigma\rightarrow \Sigma^\lambda=\Sigma - \lambda\rhd \Sigma \\ \Gamma \rightarrow \Gamma^\lambda=\Gamma- \lambda\rhd' \Gamma\end{cases},\nonumber \\
     L&:& \begin{cases}A\rightarrow A^L=A+tL \\ \Sigma\rightarrow \Sigma^L=\Sigma+d_AL + \frac{1}{2}tL\wedge L \\ \Gamma\rightarrow \Gamma^L =\Gamma + \{\Sigma^L\wedge L\}_\text{Pf}+\{L\wedge \Sigma\}_\text{Pf}\end{cases},\nonumber \\
     \mathcal{L}&:&\begin{cases} A\rightarrow A^\mathcal{L}=A \\ \Sigma\rightarrow \Sigma^\mathcal{L} =\Sigma+t'\mathcal{L} \\ \Gamma\rightarrow \Gamma^\mathcal{L} = \Gamma+ d_A\mathcal{L}\end{cases}\label{eq:3gautran}
\end{eqnarray}
that generate the gauge symmetry $\operatorname{Gau}_3$ of our 3-gauge theory. This was also derived in~\cite{Radenkovic:2019qme,Radenkovic:2021fqq}.

\subsubsection{Compatibility Between the 3-Gauge Transformations}
Importantly, it was noted in  \cite{Radenkovic:2019qme} that 2-gauge transformations $L$ do not generate a subalgebra when the Peiffer bracket is not zero. Indeed, they commute only up to a 3-gauge transformation,
\begin{equation}
    [(0,L_1,0), (0,L_2,0)] = (0,0,\mathcal{L}_{12}),\qquad \mathcal{L}_{12}=2(\{L_1\wedge L_2\}_\text{Pf}-\{L_2\wedge L_1\}_\text{Pf}).\label{eq:cenex}
\end{equation}
{On} the other hand, the computations
\begin{equation}
    [(\lambda_1,0,0),(\lambda_2,0,0)] = ([\lambda_1,\lambda_2],0,0),\qquad [(0,0,\mathcal{L}_1),(0,0,\mathcal{L}_2)] =0,\nonumber
\end{equation}
as well as the obvious results
\begin{equation}
    [(\lambda,0,0),(0,L,0)] = (0,\lambda\rhd L,0),\qquad [(\lambda,0,0),(0,0,\mathcal{L})] = (0,0,\lambda\rhd'\mathcal{L}), \nonumber
\end{equation}
allow us to completely characterize the 3-gauge group as
\begin{equation}
    \operatorname{Gau}_3 = \mathscr{E}\rtimes (\Omega^0(X)\otimes\mathfrak{g}),\qquad 0\rightarrow \Omega^2(X)\otimes\mathfrak{i}\rightarrow \mathscr{E}\rightarrow\Omega^1(X)\otimes\mathfrak{h}\rightarrow 0.\nonumber
\end{equation}
{Here}, $\mathscr{E}$ can be seen as a sort of central extension of the 2-gauge by the 3-gauge according to (\ref{eq:cenex}).

\subsection{3-Curvature Anomaly and Its First Descendant}\label{3anom}
The goal in this section is to study the anomaly $\tau$ of the modified 3-curvature $\mathcal{H}$, as well as derive conditions on its descendant. In the absence of the 2-Bianchi anomaly, the modified 3-curvature $\mathcal{H}$ is valued in $\operatorname{ker}t'$, and so must the anomaly $\tau$. We shall see that, analogous to Section \ref{twisting}, given that the 3-curvature anomaly $\tau$ takes a particular form, then it is related to the classifying cohomology class of the underlying 2-crossed-module $\mathscr{G}$ under consideration.

We wish to insert $\tau$ which preserves the covariance of the anomaly EOM $\mathcal{H}=\tau$ under the 3-gauge transformations \eqref{eq:3gautran}. This tells us that $\tau$ should transform convariantly under a 3-gauge transformation, identically to how $\mathcal{H}$ transforms \eqref{3curv3gau}. Therefore, on-shell of the fake-flatness condition $\mathcal{F}=0$, the 3-curvature anomaly $\tau$ should be {{invariant}} under a 3-gauge transformation, meaning that it must be {\it {2-shift-invariant}} $\tau(A,\Sigma) = \tau(A,\Sigma+t'\mathcal{L})$. As such, $\tau$ can only be a function on $\mathfrak{g}$ and $\operatorname{coker}t' = \mathfrak{h}/\operatorname{im}t'$. Notice that the quadratic term depends only on $\operatorname{coker}t'$, as $Q_{t\mathcal{L}} = [\mathcal{L}\wedge \mathcal{L}] =0$ by the 2-Peiffer condition.

We now examine the conditions for which the 3-curvature anomaly EOM $\mathcal{H} = \tau$ is covariant under 2-gauge transformations. Once again, the covariance of $\mathcal{H}$ \eqref{3curv2gau} implies that $\tau=\tau(A)$ cannot depend on the 2-connection $\Sigma$, and must be shift-invariant $\tau(A)=\tau(A+tL)$. This casts $\tau$ as a $\operatorname{ker}t'$-valued, degree-4 function of $\operatorname{coker}t$, which is precisely the data of a 4-cocycle representative of the Lie algebra cohomology class $[\tau]\in H^4(\operatorname{coker}t,\operatorname{ker}t')$ that classifies the 2-crossed-module $\mathscr{G}$ up to equivalence \cite{Wag:2006,Ang2018,Brown}; see also {\it {Remark} \ref{deg4class}} in {Appendix} %MDPI: Remark A.1 was in Appendix A.1, Please check %Hank: confirm
 \ref{liealgcohomo}.

\subsubsection{Twisted 1-Gauge Transformations} Going back to the anomaly EOM $\mathcal{H}=\tau(A)$ for the modified 3-curvature, we have shown that $\tau$ is a function of $\operatorname{coker}t = \g/\operatorname{im}t$ valued in $\operatorname{ker}t$. This puts us in an identical situation as the 2-curvature anomaly $\kappa(A)$. 

With the 1-gauge transformations remaining, we suppose the 1- and 2-form connections $(A,\Sigma)$ transform as usual, and define the first descendant $\xi(A,\lambda)$ of $\tau(A)$ as a twisted 1-gauge transformation in the 3-connection satisfying the descent equation,
\begin{equation}
    \Gamma \rightarrow \Gamma + \lambda\rhd'\Gamma + \xi(A,\lambda),\qquad  d_{A^\lambda}\xi(A,\lambda) = \tau(A^\lambda) - \lambda\rhd'\tau(A),\nonumber
\end{equation}
such that the covariance of $\mathcal{H}$,  \eqref{3curv1gau}, gives
\begin{equation}
    \mathcal{H}^\lambda = \mathcal{H} + \lambda\rhd'\mathcal{H} = \tau(A^\lambda).\nonumber
\end{equation}
{Here}, the first descendant of the 3-gauge $\xi(A,\lambda) \in \Omega^3(X)\otimes\operatorname{ker}t'$ is a 3-form, in contrast to the 2-form $\zeta(A,\lambda)$ encountered in Section \ref{twisting}. 

With $\tau$ and $\xi$ valued in $\operatorname{ker}t'$, this twisted gauge transformation does not conflict with the covariance of the modified 3-curvature $\mathcal{H}$ above. Moreover, due to the exactness $tt'=0$ of the complex  \eqref{eq:3comp}, the 3-gauge descendant $\xi$ is { {independent}} from any 2-curvature anomaly.

%\medskip

\subsubsection{2-Curvature Anomaly and First Descendant as 3-Gauge Data}\label{2curvanomal}

It is more accurate to say that { {there is only the 3-gauge descendant}} here, as the 2-curvature anomaly $K=\kappa(A)$ can be understood as a particular 3-connection via the 2-fake-flatness condition $t'\Gamma=K$. Indeed, as $K=d_A\Sigma$, we can consider $\Gamma = d_A\mathcal{L}(A)$ as a pure 3-gauge whose gauge parameter depends on the 1-connection $A$, with $t'=\operatorname{id}$ the identity. 

Note that $\kappa=\kappa(A)$ can only depend on the 1-connection $A$, which is unaffected by the 3-gauge transformation. Hence locally, we can perform a 3-gauge shift parametrized by $\mathcal{L}$ (which could depend on $A$) in order to remove the 2-curvature anomaly,
\begin{equation}
    K-\kappa(A) \rightarrow K + d_A\mathcal{L}(A) -\kappa(A) = K,\nonumber
\end{equation}
which introduces a gauge-fixing of the 3-connection to a pure gauge $\Gamma=d_A\mathcal{L}(A)$. With this choice understood, we now make a 3-gauge transformation followed by a 1-gauge transformation on the 2-connection
\begin{equation}
    \Sigma\xrightarrow{\mathcal{L}} \Sigma+\mathcal{L}\xrightarrow{\lambda} \Sigma+\lambda\rhd \Sigma+ \mathcal{L}+\mathcal{L}^\lambda,\nonumber
\end{equation}
where we have kept the transformation $\mathcal{L}\xrightarrow{\lambda}\mathcal{L}+\mathcal{L}^\lambda$ implicit, as $\mathcal{L}=\mathcal{L}(A)$ now depends on $A$. 

In the opposite order, we have
\begin{equation}
    \Sigma\xrightarrow{\lambda} \Sigma+\lambda\rhd \Sigma\xrightarrow{\mathcal{L}} \Sigma + \mathcal{L} + \lambda\rhd(\Sigma+\mathcal{L}) = \Sigma + \lambda\rhd \Sigma + \mathcal{L} + \lambda\rhd \mathcal{L}.\nonumber
\end{equation}
{The} difference is the expression
\begin{equation}
    \mathcal{L}^\lambda - \lambda\rhd\mathcal{L},\nonumber
\end{equation}
which upon taking the gauge-transformed covariant derivative $d_{A^\lambda}$ yields
\begin{equation}
    d_{A^\lambda}\mathcal{L}^\lambda - d_{A^\lambda}(\lambda\rhd \mathcal{L}) = (d_A\mathcal{L})^\lambda - \lambda\rhd d_A\mathcal{L} = \kappa(A^\lambda)-\lambda\rhd \kappa(A).\nonumber
\end{equation}

This is nothing but the descent Equation  (\ref{eq:desceq}) satisfied by the first descendant $\zeta(A,\lambda)$ of the 2-curvature anomaly $\kappa(A)$, if we take
\begin{equation}
    \zeta(A,\lambda) = \mathcal{L}^\lambda - \lambda\rhd\mathcal{L}.\nonumber
\end{equation}
{This} leads to the identification with the commutator
\begin{equation}
    [(0,0,\mathcal{L}),(\lambda,0,0)] = (0,0,\zeta(A,\lambda)).\label{eq:descomm}
\end{equation}
{In} other words, the 2-curvature anomaly $\kappa(A)$ arising from a Postnikov class can be absorbed by a pure 3-gauge, while its descendant can be absorbed by the commutator \eqref{eq:descomm}, thereby embedding a non-trivial 2-gauge theory into an { {anomaly-free}} 3-gauge theory. This is the spirit of anomaly resolution.

\section{Applications}\label{sec:3app}
In this section, we discuss concrete examples in which 3-gauge structures naturally arise.

\subsection{3-BF Theory}\label{sec:3actions}
The simplest topological action to consider is once again an action implementing merely the constraints---namely, the 1-, 2-fake-flatness and the flat 3-curvature conditions---as equations of motion (EOM). Such a theory has been studied in detail in  \cite{Radenkovic:2019qme}, and we follow their treatment here as well.

\subsubsection{Action and EOMs} As previously, we fix a 2-crossed-module $\mathscr{G}=\mathfrak{i}\xrightarrow{t'}\h\xrightarrow{t}\g$, and we introduce the dual spaces $\mathfrak{i}^*,\h^*,\g^*$ of linear functionals on the Lie algebras $\mathfrak{i},\h,\g$, respectively. We denote their pairing forms collectively by $\langle\cdot,\cdot\rangle$. We begin by introducing Lagrange multipliers $B\in \Omega^{d-2}\otimes\mathfrak{g}^*,C\in \Omega^{d-3}\otimes\mathfrak{h}^*,D\in\Omega^{d-4}(X)\otimes\mathfrak{i}^*$ implementing the aforementioned conditions.

The 3-$BF$ action (also called the \textit{BFCGDH} action, but we shall not use this name for obvious reasons) without any 3-curvature anomalies is then
\begin{equation}
    S_\text{3BF} = \int_X \langle B\wedge \mathcal{F}(A,\Sigma)\rangle + \langle C\wedge\mathcal{G}(A,\Sigma,\Gamma)\rangle + \langle D\wedge\mathcal{H}(A,\Sigma,\Gamma)\rangle,\label{eq:3bf}
\end{equation}
in which $\mathcal{F} = F-t\Sigma,\mathcal{G} = d_A\Sigma-t'\Gamma$, and $\mathcal{H}=d_A\Gamma+Q_\Sigma$ is the modified 3-curvature. Recall these curvature quantities are covariant, \eqref{gtff}, \eqref{3curv1gau}--\eqref{3curv3gau}. For $d=3$, the 3-$BF$ theory reduces to a 2-$BF$ theory, since the dual field $D$ does not exist. 

The first set of EOMs is
\begin{equation}
    \delta B \Rightarrow \mathcal{F}=0,\qquad \delta C\Rightarrow \mathcal{G}=0,\qquad \delta D\Rightarrow \mathcal{H}=0,\nonumber
\end{equation}
which implement precisely the 1-, 2-fake-flatness and 3-flatness conditions, respectively. Since we also have  to vary $A,\Sigma$ and the 3-connection $\Gamma$, in addition to the maps $\Delta,t^*$ given in Section \ref{sec:actions}, we introduce
\begin{eqnarray}
\Delta':\mathfrak{i}^{\otimes 2}\rightarrow\mathfrak{g} &\,,\,& \langle D\wedge A\wedge^{\rhd'}\Gamma\rangle=-\langle \Delta'(D\wedge \Gamma)\wedge A\rangle,\nonumber \\
    \Omega:\mathfrak{i}\rightarrow\mathfrak{h},&\,,\,& \langle D\wedge Q_\Sigma\rangle =-\langle \Omega(D)\wedge \Sigma\rangle,\nonumber\\
    t'^*:\mathfrak{h}\rightarrow\mathfrak{i}, &\,,\,& \langle C\wedge t'\Gamma\rangle = \langle t'^*C \wedge \Gamma\rangle,\nonumber 
\end{eqnarray}
and also the dual action
\begin{equation}
    \langle z,x\rhd' z'\rangle = \langle x\rhd'^*z,z'\rangle \nonumber 
\end{equation}
for all $z,z'\in\mathfrak{i},X\in\mathfrak{g}$. Notice that $\Omega(D)$ is a ($d-2$)-form.

These yield the dual EOMs
\begin{eqnarray}
    \delta A &\Rightarrow& dB+[A\wedge B]^* - \Delta(C\wedge\Sigma)-\Delta'(D\wedge \Gamma)=0,\nonumber \\
    \delta \Sigma &\Rightarrow& dC + A\wedge^{\rhd^*} C  -t^*B-\Omega(D)=0,\nonumber\\
    \delta\Gamma &\Rightarrow& dD + A\wedge^{\rhd'^*}D-t'^*C=0 .\nonumber
\end{eqnarray}
{If} we define, in addition to $\tilde F=dC+A\wedge^{\rhd^*}C$ and $\tilde K=dB+[A\wedge B]^*$ as in  \eqref{eq:dual}, the quantity
\begin{equation}
    \tilde T = dD + A\wedge^{\rhd'^*}D,\nonumber
\end{equation}
we see that these dual EOMs read
\begin{eqnarray}
    \text{$(d-1)$-form:}&\qquad& \tilde K = \Delta(C\wedge \Sigma) + \Delta'(D\wedge \Gamma),\nonumber \\ 
    \text{$(d-2)$-form:}&\qquad& \tilde F = t^*B+ \Omega(D),\nonumber\\ 
    \text{$(d-3)$-form:}&\qquad& \tilde T = t'^*C.\label{eq:dual1}
\end{eqnarray}

\subsubsection{Symmetries of the Action} 
Similar to the 2-gauge case, we also acquire 3-gauge transformations in the dual fields $B,C,D$. These have been developed in  \cite{Radenkovic:2019qme}, but to write them down, we must introduce yet more structures. 

We define the following maps {(The Peiffer pairing defines two maps $\mathfrak{h}\rightarrow \mathfrak{h}^*\times\mathfrak{i}$ by $\mathpzc{Pf}_y=\{y,\cdot\}_\text{Pf}$ and its conjugate $\overline{\mathpzc{Pf}}_y=\{\cdot,y\}_\text{Pf}$. Then $\omega_1(\cdot,Y) =-\mathpzc{Pf}_y^*$ is the dual and $\omega_2(\cdot,y)=-\overline{\mathpzc{Pf}}_{y}^*$ is the conjugate dual.)}
\begin{eqnarray}
    \omega_{1,2}: \mathfrak{i}\times \mathfrak{h}\rightarrow \mathfrak{h}&\Rightarrow& \langle z,\{y,y'\}_\text{Pf}\rangle = -\langle \omega_1(z,y),y'\rangle = -\langle\omega_2(z,y'),y\rangle,\nonumber\\ 
    \mathcal{Y}: \mathfrak{i}\times\mathfrak{h}^{\otimes 2}\rightarrow \mathfrak{g} &\Rightarrow& \langle z,\{x\rhd y,y'\}_\text{Pf}\rangle = -\langle \mathcal{Y}(z,y,y'),x\rangle,\nonumber
\end{eqnarray}
for each $Z\in\mathfrak{i},Y,Y'\in\mathfrak{h},X\in\mathfrak{g}$. Notice that since $\{y,y\}_\text{Pf} = Q_y$, we have 
\begin{equation}
    \langle \omega_{1,2}(z,y),y\rangle = \langle \Omega(z),y\rangle,\nonumber
\end{equation}
and that $\omega_1\neq \omega_2$ or $\omega_1 \neq -\omega_2$ in general, as $\{\cdot,\cdot\}_\text{Pf}$ is {\it {not}} symmetric or skew-symmetric.

The dual 3-gauge transformations are given by \cite{Radenkovic:2019qme} 
\begin{eqnarray}
    \lambda&:& \begin{cases}B\rightarrow B^\lambda= B+[\lambda,B] \\ C \rightarrow C^\lambda = C + \lambda\rhd C \\ D\rightarrow D^\lambda = D + \lambda\rhd' D\end{cases},\nonumber \\
    L &:& \begin{cases} B\rightarrow B^L =  B + \Delta(C\wedge L) + \mathcal{Y}(D\wedge L\wedge L) \\ C\rightarrow C^L = C + \omega_1(C\wedge L) + \omega_2(C\wedge L) \\ D\rightarrow D^L = D\end{cases},\nonumber \\
    \mathcal{L}&:& \begin{cases} B\rightarrow B^\mathcal{L}= B + \Delta'(D\wedge \mathcal{L}) \\ C\rightarrow C^\mathcal{L} = C \\ D\rightarrow D^\mathcal{L} = D\end{cases}, \label{eq:dual3}
\end{eqnarray}
which preserves the 3-$BF$ action  (\ref{eq:3bf}) when performed alongside the 3-gauge transformations  (\ref{eq:3gautran}). 

Analogous to the 2-gauge case in Section \ref{sec:actions}, we see that the 3-gauge group $\operatorname{Gau}_3=\mathscr{E}\rtimes\Omega^1(X)\otimes\mathfrak{g}$ acts on the dual fields $B,C,D$. As such, one would expect the data $(\Delta,\Delta',\omega^{1,2},\mathcal{Y})$ emergent from the dual EOMs  \eqref{eq:dual1} to define a strict coadjoint representation $\operatorname{ad}^*:\mathscr{G}\rightarrow\operatorname{End}\mathscr{G}^*[2]$ on the dual three-term algebra complex
\begin{equation}
    \mathscr{G}^*[2]: \mathfrak{g}^*\xrightarrow{t^*}\mathfrak{h}^*\xrightarrow{t'^*} \mathfrak{i}^*,\qquad t'^*t^* = (tt')^* = 0.\nonumber
\end{equation}
{Unfortunately}, the duality theory of Lie 3-algebras has not been studied in the literature, and the notion of a ``3-Manin triple'' has yet to be developed. We leave this task to the ambitious reader.

\begin{remark}
Notice that, in order for the dual fields $B,C,D$ to have the right degree-count to serve as a ``dual 3-connection'' $(D,C,B)$, we must have $d=\operatorname{dim}X=5$ in contrast to the case in 2-BF theory (where $d=4$). As such, it seems that the ``3-Manin triple'' most naturally provides the symmetry structure of the 3-BF theory in 5D.
\end{remark}

\subsection{3-Yang--Mills Theory}

Given the 3-gauge structure based on a 2-crossed-module $$ \mathfrak{i}\xrightarrow{t'}\mathfrak{h}\xrightarrow{t}\mathfrak{g}$$ 
in hand, we now construct a {\it {3-Yang--Mills theory}} $S_\text{3YM}$ and its coupling to higher-form currents.

Let $X$ denote a smooth oriented manifold. The {\it {3-Yang--Mills action}} is defined as \cite{Baez:2002highergauge,Song:2021}
\begin{equation}
    S_\text{3YM} = \int_X (\ast\mathcal{F}\wedge\mathcal{F})_\mathfrak{g}+(\ast\mathcal{G}\wedge \mathcal{G})_\h, \quad \textrm{with }  \mathcal{F} = F- t\Sigma,\quad \mathcal{G}=K - t'\Gamma,\label{eq:3ym}
\end{equation}
where we recall the field contents are
\begin{equation*}
    A\in \Omega^1(X,\mathfrak{g}),\qquad \Sigma\in\Omega^2(X,\mathfrak{h}),\qquad \Gamma\in\Omega^3(X,\mathfrak{i}).
\end{equation*}
{Given} the invariance of the pairing forms $(\cdot,\cdot)_{\mathfrak{i},\h,\g}$, \eqref{eq:3ym} is by construction invariant under 3-gauge transformations  (\ref{eq:3gautran}), since each term consist of 3-gauge invariant quantities.

Though, it is important to note that the pairing $(\cdot,\cdot)$ used in \eqref{eq:3ym} is different from the one used in 3-BF theory \eqref{eq:3bf}; they have degree-0 while the one $\langle\cdot,\cdot\rangle$ used for 3-BF theory has degree-2.

\subsubsection{Sourced Equations of Motion in 3-Yang--Mills Theory}
By varying the 3-gauge $\Gamma$, we obtain the 2-fake-flatness EOM $\mathcal{G}=K-\Gamma=0$. Once again, we see that no non-trivial higher anomalous charges can be introduced on-shell of this EOM.

To circumvent this, we source the 1-, 2- and 3-form connections $A,\Sigma,\Gamma$, respectively, with the 1-, 2- and 3-form currents $j,J,\mathcal{J}$. This introduces the following term
\begin{equation}
    S_\text{3cur}=\int_X (A\wedge \ast j)_\g+ (\Sigma\wedge \ast J)_\h+ (\Gamma\wedge \ast\mathcal{J})_\mathfrak{i} \nonumber
\end{equation}
to $S_{\text{3YM}}$  (\ref{eq:3ym}). The properties of the currents are listed below 
\begin{center}
    \begin{tabular}{|c|c|c|}
    \hline
    Currents & Form degree & Valued in \\
    \hline
    $j$ & 1 & $\mathfrak{g}$ \\ 
    $J$ & 2 & $\mathfrak{h}$ \\ 
    $\mathcal{J}$ & 3 & $\mathfrak{i}$\\
    \hline
    \end{tabular}
\end{center}
{Upon} introducing these currents, a variation with respect to the field content $(A,\Sigma,\Gamma)$ yields the modified EOMs. To express these EOMs, we must introduce the adjoints of the linear maps $t,t'$ defined by
\begin{equation*}
    (tY,X)_\mathfrak{g} = ( Y,t^TX)_\mathfrak{h},\qquad ( t'Z,Y)_\mathfrak{h} = (Z,t'^TY)_\mathfrak{i}
\end{equation*}
where $X\in\mathfrak{g},Y\in\mathfrak{h}$ and $Z\in\mathfrak{i}$. 

We thus have (recall $F=dA+\frac{1}{2}[A\wedge A]$ and $K= d\Sigma+A\wedge^\rhd \Sigma$)
\begin{align}
    \delta\Gamma \quad\implies \quad \ast \mathcal{J}&=t'^T K -\Gamma, \nonumber\\
    \delta \Sigma\quad\implies \quad\ast J &= -(t^TF-\Sigma) - d_A\ast(K-t'\Gamma) \nonumber\\
    \delta A \quad\implies\quad \ast j &=  d_A\ast(F-t\Sigma) - \Delta_{\Sigma \wedge} \ast(K-t'\Gamma)\label{3ymeom}\,,
 \end{align}
 where the map $\Delta: \mathfrak{h}\times \mathfrak{h}\to \mathfrak{g}$ is given by
 \begin{equation*}
     (X,\Delta_Y(Y'))_\mathfrak{g} = ( X\rhd Y,Y')_\mathfrak{h}
 \end{equation*}
 for $Y,Y'\in\mathfrak{h}$.

\subsubsection{3-Conservation Laws and Higher Mobility Constraints} 
We now derive  higher-conservation laws of $S_\text{3YM}+S_\text{3cur}$ by making 3-gauge transformations. This can be directly obtained from the EOMs; indeed, it is not difficult to show that, on-shell of these EOMs, we have the following relation
\begin{equation*}
    \ast j = -d_A (tJ) - (-1)^{3(n-3)} \Delta_{\Sigma\wedge}(t'\mathcal{J}),
\end{equation*}
where we have used the fact that $\ast^2 = (-1)^{k(n-k)}$ on $k$-forms on $X$, and $n=\operatorname{dim}X$. 

\begin{proposition}
    We have
    \begin{equation}
        [X,\Delta_Y(Y')] = \Delta_Y(X\rhd Y') - \Delta_{Y'}(X\rhd Y)
    \end{equation}
    where $X\in\mathfrak{g}$ and $Y,Y'\in\mathfrak{h}$.
\end{proposition}
\begin{proof}
    By contracting with $X'\in\mathfrak{g}$ we have
    \begin{align*}
        ( X',[X,\Delta_Y(Y')])_\mathfrak{g} &= -([X,X'],\Delta_Y(Y'))_\mathfrak{g} = -( [X,X']\rhd Y,Y')_\mathfrak{h} \\
        &= -( X\rhd (X\rhd' Y),Y')_\mathfrak{h} +( X'\rhd(X\rhd Y),Y')_\mathfrak{h} \\ 
        &= ( X'\rhd Y,X\rhd Y')_\mathfrak{h} - ( X\rhd Y,X'\rhd Y')_\mathfrak{h} \\
         &= ( X',\Delta_Y(X\rhd Y'))_\mathfrak{g} -( X',\Delta_{Y'}(X\rhd Y))_\mathfrak{g}
    \end{align*}
    as desired.
\end{proof}

By applying the covariant derivative $d_A$ to this relation, we have from the above proposition and directly from the EOMs \eqref{3ymeom} that
\begin{align}
    d_A\ast j &= - F\wedge^\rhd tJ - (-1)^{3(n-3)}\Big( \Delta_{\Sigma\wedge}d_A(t'\mathcal{J}) - \Delta_{t'\mathcal{J}\wedge} K\Big)\,\label{3cons1} \\ 
    d_A\ast J &= -F\wedge^{\rhd'} \ast \mathcal{H} \label{3cons2}\\
    d_A\ast \mathcal{J}&= 0,\label{3cons3}
\end{align}
where we have used the 1- and 2-Bianchi identities
\begin{equation*}
    d_A\mathcal{F} =0 ,\qquad d_A\mathcal{H} = 0.
\end{equation*}
{It} is then clear that the 1-current $j$ is conserved provided the 2-/3-currents $J,\mathcal{J}$ take values in the kernels of the 2-crossed-module maps $t,t'$, respectively. However, even then the 2-current $J$ remains {{non-conserved}} in general, unless we have $\cH=0$

A direct consequence of these 3-conservation law  (\ref{3cons1})-(\ref{3cons3}) is the very interesting mobility constraint. Suppose $J|_{\operatorname{ker}t},\mathcal{J}_{\operatorname{ker}t'}$ denote the restriction of the 2-/3-form currents to the kernels of the maps $t,t'$. If we define the 0-/1-/2-form charges $u,U,\mathcal{U}$ associated to the currents $j,J|_{\operatorname{ker}t},\mathcal{J}_{\operatorname{ker}t'}$, respectively, then:
\begin{enumerate}
    \item the 0-/2-form charge $u,\mathcal{U}$ are mobile (i.e., ``\textit{{liquid}}'') and define a topological operator,~but
    \item the 1-form charge $U$ is {{not}} topological, unless certain flatness or mobility constraints are imposed.
\end{enumerate}

{In} %MDPI: We removed the \noindent format, please confirm %Hank: confirm
 particular, we see that the 1-form charge $U$ is only topological when wrapped around a 3-manifold upon which $\mathcal{H}$ vanishes. As such, the presence of a 3-form 2-curvature anomaly, for which $\cH\neq 0$, can be viewed as an obstruction to the conservation and the homotopy-invariance of the 1-form charge $U$.

% The immobility of the 1-form charges, guarantees that both the 2-form $U(1)_2$ and the $PG$ 0-form symmetries are non-anomalous, in the sense that the corresponding 3-,1-currents $\mathcal{J},j$ are conserved,
% \begin{equation}
%     d\ast \mathcal{J} =0,\qquad d_A\ast j = 0,\qquad d_A\ast J_\ell=t^*(\ast j) = 0.\nonumber
% \end{equation}
% The consistency of the last conservation law requires $j\in\operatorname{ker}t^*\subset P\g$, which by the rank-nullity theorem
% \begin{equation}
%     \operatorname{ker}t^* \cong \operatorname{coker} t = P\mathfrak{g}/\Omega\mathfrak{g}\cong\mathfrak{g}\nonumber
% \end{equation}
% means that $j\in \g$ is valued in the {\it constant} paths in $\g$. Hence, it appears that the $PG$-charges that descend to those of $G$ can still remain mobile. Therefore, the only charges that are allowed to be mobile are the ones labeled by the 0-form $G$ and the 2-form $U(1)_2$ symmetries. 

% Interestingly, despite us starting with the loop model of string Lie 2-algebra, this "mobility data" constitutes precisely the data in the skeletal model \cite{Baez:2005sn,Kim:2019owc} (see {\it Remark \ref{rem:string}}). 

We expect that this novel mobility constraint conditions may find applications in interesting higher-dimensional topological orders and SPTs. We shall investigate more thoroughly  the higher-algebraic structure of the charged operators in 3-Yang--Mills theory \eqref{eq:3ym} (in accordance with \cite{Gaiotto:2014kfa,costello_gwilliam_2016,Costello:2020ndc,Budzik:2022mpd,Zeng:2023qqp}, etc.) in a future work.

\subsection{Loop Model of the String Lie 2-Algebra}\label{loopgauge}
Let us now look at a concrete novel example of higher anomaly resolution using 3-gauge theory. To set this up, we need a \textit{{non-trivial}} Lie 2-algebra (i.e., its Postnikov class is non-zero). A classic example of this is the string Lie 2-algebra $\mathfrak{string}_k(\mathfrak{g})$, which admits a a description in terms of a non-trivial crossed-module $\mathfrak{l}_k$ \cite{Baez:2005sn}.

 \begin{definition}
     The {{loop model}} for the string 2-algebra is the Lie algebra crossed-module $\mathfrak{l}_k = (t:\widehat{\Omega_k\mathfrak{g}}\cong\Omega\mathfrak{g}\oplus\mathbb{R}\rightarrow P\mathfrak{g},\rhd)$ where 
     \begin{enumerate}
    \item the map $t$ is given by the affine projection $\widehat{\Omega_k\mathfrak{g}}\rightarrow\Omega\mathfrak{g}$ composed with the inclusion $\Omega\mathfrak{g}\hookrightarrow P\mathfrak{g}$,
    \item $\operatorname{ker}t\cong\mathbb{R}=\{(0,c)\in \Omega\mathfrak{g}\oplus\mathbb{R}\}$, and $\operatorname{coker}t\cong P\mathfrak{g}/\Omega\mathfrak{g}\cong \mathfrak{g}$ is isomorphic to the constant~paths,
    \item the action is defined as
    \begin{equation}
        p\rhd (\ell,c) = ([p,\ell], 2k\int_{S^1}\langle p,\dot\ell\rangle),\qquad p\in P\mathfrak{g},\quad (\ell,c)\in \widehat{\Omega_k\mathfrak{g}},\label{eq:action}
\end{equation}
whence the induced action $\rhd$ of $\operatorname{coker}t=\mathfrak{g}$ on $\operatorname{ker}t=\mathbb{R}$ is trivial,  
\item the Postnikov class $[\kappa]\in H^3(\mathfrak{g},\mathbb{R})$ of $\mathfrak{l}_k$ is given by the {\it fundamental 3-cocycle} $\langle\cdot,[\cdot,\cdot]\rangle$ on $\mathfrak{g}$~\cite{book-loop,Baez:2005sn}.
\end{enumerate}
 \end{definition}

{It }%MDPI: We removed the \noindent format, please confirm %Hank: confirm
 can be deduced from the Peiffer identity that the Lie bracket on $\widehat{\Omega_k\mathfrak{g}}$ is given by the level-$k$ Kac--Moody extension
\begin{equation}
    [(\ell,c),(\ell',c')] = t(\ell,c)\rhd (\ell',c') = ([\ell,\ell'],2k\int_{S^1}\langle\ell,\dot\ell'\rangle)\,.\nonumber
\end{equation}
% We shall focus on the case of level $k=1$ in the following for brevity, but the case for arbitrary values of $k\in \bbZ$ can be treated identically. 

% For each $f:[0,1]\rightarrow [0,1]$ satisfying $f(0)=0,f(1)=1$, we can lift a path $p\in P\g$ to a loop by taking $s_f:p\mapsto s_f(p) = p-f\cdot p(1)$, where $s_f(p)(\tau) = p(\tau)- f(\tau)\cdot p(1)$ and $\tau\in S^1$ is the parameter on the loop. We shall use the lift $s_f$ defined above to express quantities that involve elements in both $P\g$ and $\Omega\g$. Given $p\in P\g$, we call the element $s_fp$ the {\it looping} of $p(1)\in \g$. It was shown in \cite{Baez:2005sn} that this map $s_f$ induces a homotopy between the identity and the map $p\mapsto f\cdot p(1)$ on $\mathfrak{l}_k$.

% With this lift, we define the following bilinear function
% \begin{equation}
%     \chi(p,p') = (s_f[p,p'],2\int_{[0,1]}\langle p,\dot p'\rangle)=([p,p']-f\cdot [p,p'](1),2\int_{[0,1]}\langle p,\dot p'\rangle),\qquad p,p'\in P\g,\nonumber
% \end{equation}
% which shall become important in the following.

The main result in  \cite{Baez:2005sn} is that $\mathfrak{l}_k$ is homotopy equivalent to the weak skeletal model~\cite{Kim:2019owc}.
\begin{theorem}[Baez--Crans--Stevenson--Schreiber]\label{thm:bcss}
 The maps
\begin{eqnarray}
    \psi:P\mathfrak{g} \rightarrow \mathfrak{g},&\qquad& p \mapsto p(1),\nonumber \\
    \phi: \widehat{\Omega_k\mathfrak{g}} \rightarrow \mathbb{R}, &\qquad& (\ell,c)\mapsto c,\nonumber \\
    \varphi: P\mathfrak{g}^{2\otimes}\rightarrow\mathbb{R}, &\qquad& (p,p')\mapsto k\int_{[0,1]}(\langle p,\dot p'\rangle - \langle p',\dot p\rangle),\nonumber
\end{eqnarray}
define a weak equivalence $\Psi=(\varphi,\phi,\psi):\mathfrak{l}_k\rightarrow\mathfrak{string}_k(\mathfrak{g})$. Moreover, given the Lie 2-group $\mathcal{L}_k$ corresponding to the Lie 2-algebra $\mathfrak{l}_k$, its classifying datum $\tau\in  H^3(G,U(1))$ integrating the Postnikov class $\kappa$ of $\mathfrak{l}_k$ coincides with the Diximier-Douady class $k[\omega]\in H^3(G,\bbZ)$ classifying the string 2-group $\operatorname{String}_k(G)$ as a $U(1)$-bundle gerbe,
\begin{equation*}
    BU(1)\to \operatorname{String}_k(G)\to G\,, \qquad \omega = \exp \langle\cdot,[\cdot,\cdot]\rangle\,,
\end{equation*}
% iven the Lie 2-group $\mathcal{L}_k$ denote the Lie 2-group corresponding to $\mathfrak{l}_k$. The , namely the Postnikov class of $\mathcal{L}_k$, coincides with the level-$k$ Dixmier-Douady class of the $U(1)$-bundle gerbe 
%     \begin{equation*}
%         BU(1) \to \widehat{\mathcal{G}}_k\to G,
%     \end{equation*}
where $BU(1)=U(1)\to \ast$ is the delooping.
\end{theorem}

{More} %MDPI: We removed the \noindent format, please confirm %Hank: confirm
 concretely, it was shown that $\Psi$ is a surjective 2-homomorphism inducing a { {split}} short exact sequence
    \begin{equation}
        0\rightarrow \mathfrak{I}_{\Omega\mathfrak{g}}\rightarrow \mathfrak{l}_k \xrightarrow{\Psi} \mathfrak{string}_k(\mathfrak{g})\rightarrow 0.\label{eq:shortloop}
    \end{equation}
{Combined} with the fact that any 2-algebra $\mathfrak{I}_\mathfrak{g}\simeq 0$ with $t=\operatorname{id}$ are trivial for any $\mathfrak{g}$ up to elementary equivalence, the theorem follows.

\subsubsection{Gauge-Theoretic Interpretation} 
So what is happening in the 2-gauge theory picture? Since $\Psi$ is a (weak) elementary equivalence, it can be considered as part of a 2-bundle map $\mathcal{P}\rightarrow\mathcal{P}'$ (see Appendix \ref{weakpost}). Moreover, the sequence  (\ref{eq:shortloop}) induces also a short exact sequence of bundles
\begin{equation}
    \mathcal{P}_0\rightarrow\mathcal{P}\rightarrow\mathcal{P}',\nonumber
\end{equation}
in which $\mathcal{P}_0$ is a trivial 2-gauge bundle based on $\mathfrak{I}_{\Omega\mathfrak{g}}$. Our goal is to understand the effect of the 2-homomorphism $\Psi$ on the 2-gauge kinematical data. We shall focus on the level $k=1$ case.

We write down the gauge-covariant curvature quantities,
\begin{eqnarray}
    \mathcal{F} &=& F - t\Sigma = F - \sigma,\nonumber\\ 
    \mathcal{G} &=& K - \kappa(A) = d_A\Sigma - \kappa(A),\nonumber
\end{eqnarray}
where the 2-curvature anomaly $\kappa(A)$ is given by the Postnikov class mentioned in \linebreak Theorem~\ref{thm:bcss}. We now describe the 2-gauge data on $\mathcal{P}'$. The 1-connection $A'$ is valued in $\mathfrak{g}$, while the 2-connection $\Sigma'$ is valued in $\mathbb{R}$. The action is trivial
\begin{equation}
    A'\wedge^\rhd \Sigma' = 0,\nonumber
\end{equation}
whence the 1- and 2-curvatures read
\begin{equation}
    F' = d_{A'}A' = dA' + \frac{1}{2}[A'\wedge A'],\qquad K' = d_{A'}\Sigma' = d\Sigma'.\nonumber
\end{equation}
{The} Jacobiator is given by the fundamental 3-cocycle $\mu(A,A,A) = \omega(A,A,A)= \frac{1}{3!}\langle A\wedge[A\wedge A]\rangle$, whence we have the gauge-invariant data
\begin{equation}
    \mathcal{F}' = F',\qquad \mathcal{G}' = K' + \frac{1}{3!}\langle A\wedge [A\wedge A]\rangle.\nonumber
\end{equation}

%\medskip

The 2-homomorphism $\Psi=(\varphi,\phi,\psi)$ given in {Theorem \ref{thm:bcss}} should then induce a 2-bundle homomorphim that sends the 2-gauge data $(\mathcal{F},\mathcal{G})$ to $(\mathcal{F}',\mathcal{G}')$. To see this, we first note that $2\int_{S^1}\langle A\wedge\dot\sigma\rangle = \phi(A\wedge^\rhd \Sigma) = \varphi(A,t\Sigma)$ by  \eqref{weak2hom}. Then, noticing that $dc'$ can be gauged away, this yields
\begin{equation}
    F' = \psi F = F(1),\qquad  K' = \phi K = -2\int_{S^1}\langle A\wedge\dot\sigma\rangle,\label{eq:curvdata}
\end{equation}
where we have made the pullback map $f^*:\Omega^\bullet(X)\rightarrow\Omega^\bullet(X)$ on forms that comes with a 2-bundle homomorphism implicit. With $A' = \psi A = A(1)$ understood, the Jacobiator $\mu$ can be reconstructed from  \eqref{weak2hom} as
\begin{equation}
    \mu(A',A',A') =~ \circlearrowright\varphi(A,[A\wedge A])= 3!\varphi(A,[A\wedge A]),\nonumber
\end{equation}
where we have used the total skew-symmetry of $\varphi$.

This implies that the right-hand side should factor through the $\kappa$. More precisely, we should have
\begin{equation}
    \phi(\kappa(A)) = \varphi(A,[A\wedge A]),\nonumber
\end{equation}
such that
\begin{equation}
    \frac{1}{3!}\langle A',[A'\wedge A']\rangle = \frac{1}{3!}\mu(A',A',A') = \varphi(A,[A\wedge A]).\nonumber
\end{equation}
{This} is nothing but the statement that the 2-curvature anomaly $\kappa(A)$ coincides with $\frac{1}{3!}\omega(A',A',A')$. Indeed, one can show
\begin{equation}
    \int_0^1\langle p_1,\frac{d}{d\tau}[p_2,p_3]\rangle =\frac{1}{3!}\langle p_1(1),[p_2(1),p_3(1)]\rangle \nonumber
\end{equation}
for any $p_1,p_2,p_3\in P\g$ by performing an integration by parts then using the total skew-symmetry of the form $\langle\cdot,[\cdot,\cdot]\rangle$ \cite{Baez:2005sn}. The 2-homormophism $\Psi$ then implements
\begin{equation}
    K' - \frac{1}{3!}\mu(A',A',A') = \mathcal{G}' = \phi\mathcal{G} = \phi(K-k\omega(A,A,A)),\nonumber
\end{equation}
as desired.

\subsubsection{Anomaly Resolution for the String 2-Group}\label{string2resolve}
The formulation of the 2-gauge theory associated to the loop model $\mathfrak{l}_k$ presents an interesting possibility: which is that of absorbing the string structure anomaly by introducing a 3-gauge structure based on the differential 2-crossed-module
\begin{equation}
\mathbb{R}\cong\mathfrak{i}\xrightarrow{t'=\operatorname{id}}\widehat{\Omega_{k=1}\mathfrak{g}}\xrightarrow{t}P\mathfrak{g}\,.\label{string3alg}
\end{equation}
{Such} a 3-gauge theory (as described in Sections \ref{sec:gaug2gau} and \ref{2curvanomal}) would then host a real-valued 3-form connection $\Gamma$ which satisfies the so-called ``2-monopole condition''
\begin{equation}
    \int_V \frac{\mathsf{k}}{4\pi^2}\Gamma + \frac{1}{3!}\omega(A,A,A)=0,\qquad \forall ~\text{closed 3-surfaces }V\subset X\, \label{eq:2monocond}
\end{equation}
canceling out the Postnikov class anomaly arising from the crossed-module $\mathfrak{l}_k$, where $\mathsf{k}$ is to be understood as the quantized 2-monopole charge of the 3-gauge theory. 

If we view the Dixmier-Douady class $\omega$ as a mixed anomaly between the gauged 0-form $G$-symmetry and the 1-form $U(1)$-symmetry, then the 3-gauge structure induced by \eqref{string3alg} would allow us to gauge the anomalous categorical symmetry described by $\operatorname{string}_k$. This is an example of the {{symmetry enrichment}} procedure in the study of topological phases, and is also a higher analogue of the Green--Schwarz mechanism introduced in \cite{Benini_2019,Cordova:2018cvg}.

%\medskip

Importantly, the  2-monopole condition \eqref{eq:2monocond} can be interpreted as a certain \textit{{gauge-ability condition}} for the string structure defect. This condition enforces the consistent coupling between the anomalous higher-gauge symmetry $\mathfrak{l}_k$ and the 2-form $U(1)$-symmetry at degree-3 of \eqref{string3alg}. In analogy with the fact that a fermionic TQFT can be bosonized into a spin-TQFT \cite{Gaiotto:2015zta} by gauging a 1-form $\mathbb{Z}_2$-coupling, the 3-gauge structure \eqref{string3alg} suggests that a ``stringionic'' TQFT can also be made into a string-TQFT \cite{Debray:2023rlx} {(An invertible string-TQFT can be understood as a linear functional on the string-bordism group $\Omega^\text{String}_\ast$.)}, by gauging a 2-form $U(1)$-coupling. The consistency condition \eqref{eq:2monocond} is crucial in the construction of this~coupling.

% %\medskip

% We end this Appendix by re-emphasizing the final paragraph of Section \ref{conclusion}. 

\section{Conclusions}\label{conclusion}
% In this paper, we have described explicitly a perspective which 

% By casting higher-structures in a field theoretic language, we are then able to unite 

As a concluding comment, we emphasize that the idea that higher-gauge principles can be applied to anomaly resolution and the Green--Schwarz mechanism has already appeared previously in the mathematical literature \cite{Schreiber:2013pra,Sati:2009ic,Fiorenza:2020iax}. However, this paper provides an explicit {ab ovo} %MDPI: We removed the italics. Please confirm this revision. %Hank: confirm
 treatment on this topic, in the context of field theory. We showed that this allowed for concrete constructions in various applications in theoretical physics, and perfectly echoes the recently-popular ``generalized Landau paradigm''. 

For gauge anomalies in particular, our work emphasizes the central role played by the higher homotopy $L_\infty$-algebras in realizing (analogues of) the results in \cite{Cordova:2018cvg,Benini_2019}. We have also pinned down the dual symmetries/higher-Drinfel'd doubles present in the topological higher-BF theories that we constructed. 

A particularly interesting application of our methods was described in Section \ref{string2resolve}, where an anomalous 2-group symmetry governed by the \textit{{weak string 2-algebra}} can be gauged consistently into a non-anomalous mixed 0-, 1- and 2-form symmetry governed by a 3-group. All of the 3-gauge symmetries in this case, as well as the descent EOMs for the underlying background gauge fields, can be read directly from the explicit construction in Section \ref{sec:gaug2gau}.

Lastly, we make a brief technical observation regarding our findings in Section \ref{string2resolve}. The ``gauge-ability'' condition \eqref{eq:2monocond} we found cancels the 3-form obstruction (related to $\pi_3\operatorname{Spin}_n\cong\mathbb{Z}$) directly, instead of canceling the corresponding \textit{{4-form}} anomaly proportional to the instanton number {(The obstruction to the extension of a spin structure to a string structure is well-known to be equivalent to the fractional Pontrjagyn class $\frac{1}{2}p_1$} \cite{Sati:2008eg}. {By introducing an appropriate non-Abelian gauge coupling, one can express the Pontrjajyn class $p_1$  in terms of the second Chern class $c_2 = -\frac{1}{8\pi^2} \operatorname{tr}F\wedge F$.)} These approaches may be related through a higher analogue of the {{transgression}} $H^4(BG,U(1))\to H^3(G,U(1))$, but this is beyond the scope of this work.

% In other words, our technique allows for an explicit construction of a higher-gauge field theory which directly co-kills $\pi_3$, and differs from the usual approach. We expect this to lead to more interesting applications in the future.

\appendix

\section{Classification of Lie Algebra Crossed-Modules}\label{algxmod}
In this section we examine the classification of Lie algebra crossed-modules by Lie algebra cohomology, following  \cite{Wag:2006}. Recall that a given two Lie algebras $\mathfrak{h},\mathfrak{g}$ over a { {fixed}} field $k$ of characteristic zero, a Lie algebra crossed-module is a map  $t:\mathfrak{h}\rightarrow\mathfrak{g}$ and an action $\rhd$ of $\mathfrak{g}$ on $\mathfrak{h}$ such that the following { {Peiffer conditions}}
\begin{equation}
    t(X\rhd Y) = [X,tY]_\mathfrak{g},\qquad  tY\rhd Y' = [Y,Y']_\mathfrak{h}\label{eq:pfeif}
\end{equation}
are satisfied for each $Y,Y'\in\mathfrak{h},X\in\mathfrak{g}$. Mathematically, it is equivalent to a { {strict}} Lie 2-algebra {(Namely a two-term differential graded $L_\infty$-algebra.)}, where the homotopy map $\mu=0$ introduced in the main text vanishes.

Consider the following four-term algebra complex built from the Lie algebra crossed-module, 
\begin{equation}
    0 \rightarrow V \hookrightarrow \mathfrak{h}\xrightarrow{t}\mathfrak{g}\rightarrow \mathfrak{n}\rightarrow 0,\label{eq:ext}
\end{equation}
where $V = \operatorname{ker}t$ and $\mathfrak{n}=\operatorname{coker}t$. Due to the Peiffer identity in  (\ref{eq:pfeif}), the Lie algebra $V\subset Z(\mathfrak{h})$ must lie in the centre of $\mathfrak{h}$, and hence is Abelian. It admits an action by $\mathfrak{n}$ induced by the crossed-module action $\rhd$.

\begin{definition}\label{def:elemequiv}
We say that two crossed-modules $t:\mathfrak{h}\rightarrow\mathfrak{g},t':\mathfrak{h}'\rightarrow\mathfrak{g}'$ with the respective actions $\rhd,\rhd'$ are {\it elementary equivalent} if
\begin{enumerate}
    \item  $\operatorname{ker}t=\operatorname{ker}t'=V$ and $\operatorname{coker}t=\operatorname{coker}t' = \mathfrak{n}$,
    \item there exists Lie algebra homomorphisms $\phi:\mathfrak{h}\rightarrow\mathfrak{h}',\psi:\mathfrak{g}\rightarrow\mathfrak{g}'$ compatible with the actions $\rhd,\rhd'$ such that
    \begin{equation}
        \phi(X\rhd Y) = \psi(X)\rhd' \phi(Y)\nonumber
    \end{equation}
    for all $X\in\mathfrak{g}$ and $Y\in\mathfrak{h}$. Moreover, the diagram \begin{equation}
        \begin{tikzcd}
            &                         & \mathfrak{h} \arrow[r, "t"] \arrow[dd, "\phi"] & \mathfrak{g} \arrow[rd] \arrow[dd, "\psi"] &                        &   \\
0 \arrow[r] & V \arrow[ru] \arrow[rd] &                                                &                                            & \mathfrak{n} \arrow[r] & 0 \\
            &                         & \mathfrak{h}' \arrow[r, "t'"]                  & \mathfrak{g}' \arrow[ru]                   &                        &  
\end{tikzcd}\nonumber
    \end{equation}
    commutes.
\end{enumerate}
\end{definition}
Let us denote the set of elementary equivalence classes of Lie algebra crossed-modules {by} %MDPI: Please confirm if the bold should be retained. The following {\bf XMod} are the same %Hank: please keep the bold for "XMod"
 ${\bf XMod}(\mathfrak{n},V)$.

\subsection{Lie Algebra Cohomology}\label{liealgcohomo}
We first review some basic facts about Lie algebra cohomology, which is  a very powerful and important tool for classification of $L_\infty$-algebras. We once again follow the treatment of \cite{Wag:2006}.

Let $\mathfrak{n}$ be a Lie algebra over the field $k$ and let $V$ be an Abelian $\mathfrak{n}$-module. Define its differential graded { {Chevalley--Eilenberg complex}} 
\begin{equation}
    (C^\bullet(\mathfrak{n},V),d),\qquad C^p(\mathfrak{n},V) = \begin{cases} \Lambda(\mathfrak{n}^p,V) &; p>0 \\ V &; p=0\end{cases},\nonumber
\end{equation}
where $\Lambda(\mathfrak{n}^p,V)$ denotes the exterior algebra of alternating forms on $p$-copies of $\mathfrak{n}$ over $V$. The differential $d:C^p(\mathfrak{n},V)\rightarrow C^{p+1}(\mathfrak{n},V)$ is given explicitly by
\begin{eqnarray}
    dc(x_0,\dots,x_p) &=& \sum_{i<j}(-1)^{i+j}c([x_i,x_j],x_0,\dots,\hat{x}_i,\dots,\hat{x}_j,\dots,x_p)\nonumber \\
    &\quad& - \sum_{i=1}^p(-1)^ix_i\rhd c(x_0,\dots,\hat{x}_i,\dots,x_p)\nonumber
\end{eqnarray}
for each cochain $c\in C^p(\mathfrak{n},V)$, where $\hat{\cdot}$ denotes an omitted element.

\begin{lemma}
$d^2=0$.
\end{lemma}
\begin{proof}
Recall the { {Cartan formula}}
\begin{equation}
    L_x = d\iota_x + \iota_x d,\qquad x\in\mathfrak{n}\nonumber
\end{equation}
where $\iota_x:C^{p+1}(\mathfrak{n},V)\rightarrow C^p(\mathfrak{n},V)$ is the interior evaluation
\begin{equation}
    \iota_x:c \mapsto ((x_1,\dots,x_p)\mapsto c(x,x_1,\dots,x_p)) \nonumber
\end{equation}
and $L_x: C^p(\mathfrak{n},V)\rightarrow C^p(\mathfrak{n},V)$ is the Lie evaluation
\begin{equation}
    L_x:c\mapsto ((x_1,\dots,x_p)\mapsto x\rhd c(x_1,\dots,x_p) - \sum_i c(x_1,\dots,[x,x_i],\dots,x_p)),\nonumber
\end{equation}
which by construction commutes with $d$. Now let $v\in V = C^0(\mathfrak{n},V)$ be a 0-form, then
\begin{eqnarray}
    d^2v(x_1,x_2) &=& -dv([x_1,x_2]) + x_1\rhd dv(x_2) - x_2\rhd dv(x_1) \nonumber \\
    &=& [x_2,x_1]\rhd v + x_1\rhd (x_2 \rhd v) - x_2\rhd (x_1\rhd v) = 0,\nonumber
\end{eqnarray}
which vanishes by the $\mathfrak{n}$-module structure on $V$. 

Now let $p>0$ and assume the induction hypothesis: $d^2=0$ on $C^{p-1}(\mathfrak{n},V)$. Consider $c\in C^p(\mathfrak{n},V)$, then by the Cartan formula
\begin{eqnarray}
    d^2c(x_{-1},x_0,x_1,\dots,x_p) &=& \iota_{x_{-1}}(d^2c)(x_0,x_1,\dots,x_p) \nonumber\\
    &=& (L_{x_{-1}} - d\iota_{x_{-1}})dc(x_0,x_1,\dots,x_p) \nonumber \\
    &=& (L_{x_{-1}}d - d(L_{x_{-1}} - d\iota_{x_{-1}})c(x_0,x_1,\dots,x_p) \nonumber \\
    &=& (L_{x_{-1}}d - dL_{x_{-1}} + d^2\iota_{x_{-1}})c(x_0,x_1,\dots,x_p) = 0,\nonumber
\end{eqnarray}
where the first two terms cancel by the property $L_xd=dL_x$, and the last term vanishes due to the induction hypothesis (recall $\iota_{x_{-1}}c\in C^{p-1}(\mathfrak{n},V)$).
\end{proof}

This nilpotency allows us to define the { {Lie algebra cohomology}}
\begin{equation}
    H^\bullet(\mathfrak{n},V) = \operatorname{ker}d/\operatorname{im}d.\nonumber
\end{equation}
{These} groups are extremely useful, as they are isomorphic to the de Rham cohomology of the topological group $G$ \cite{book-loop}. Moreover, they classify various algebraic structures; for~instance,
\begin{enumerate}
    \item { {Degree $p=0$}}: the group $H^0(\mathfrak{n},V) = V^{\mathfrak{n}}\subset V$ classifies the $\mathfrak{n}$-invariants: namely elements $v\in V$ annihilated by $\mathfrak{n}$ via the action $\rhd$. Indeed, the 0-cocycle condition merely states
    \begin{equation}
        dv(x) = x\rhd v = 0,\qquad v\in V=C^0(\mathfrak{n},V),\nonumber
    \end{equation}
    which means that $v\in Z^0(\mathfrak{n},V)$ is $\mathfrak{n}$-invariant.
    \item { {Degree $p=1$}}: the group $H^1(\mathfrak{n},V)$ classifies algebra representations of $\mathfrak{n}$ on $V$ (i.e., derivations $\operatorname{Der}_\mathfrak{n}(V)$) modulo inner representations. Indeed, the 1-cocycle condition~reads
    \begin{equation}
        dc(x_1,x_2) = c([x_1,x_2]) - x_1\rhd c(x_2) + x_2\rhd c(x_1) = 0,\nonumber
    \end{equation}
    which implies that $c\in Z^1(\mathfrak{n},V)$ is a linear representation of $\mathfrak{n}$ on $V$. The 1-coboundaries are inner derivations $c(x) = dv(x) = x\rhd v$ for some $v\in V=C^0(\mathfrak{n},V)$. If $\mathfrak{n}$ acts trivially on $V$, then $H^1(\mathfrak{n},V)$ is in fact isomorphic to the (dual of the) Abelianization $\mathfrak{n}/[\mathfrak{n},\mathfrak{n}]$.
    \item { {Degree $p=2$}}: the group $H^2(\mathfrak{n},V)$ classifies { {central extensions}} $\widehat{\mathfrak{n}}$ of $\mathfrak{n}$ by $V$, which fits in the three-term exact sequence
    \begin{equation}
        0\rightarrow V\rightarrow \widehat{\mathfrak{n}}\rightarrow\mathfrak{n}\rightarrow 0.\nonumber
    \end{equation}
    To see this at a glance, a { {set-theoretic}} section $s:\mathfrak{n}\rightarrow\widehat{\mathfrak{n}}$ sees an obstruction to being a { {Lie algebra-theoretic}} section given by
    \begin{equation}
        c(x_1,x_2) = s([x_1,x_2]) - [s(x_1),s(x_2)].\nonumber
    \end{equation}
    It can be shown, with the $\mathfrak{n}$-module structure of $V$ and the Jacobi identity, that \linebreak$c\in Z^2(\mathfrak{n},V)$ is a 2-cocycle, and any two choices of such sections $s$ yields 2-cocycles $c,c'$ that differ by a 2-coboundary $c-c'=da$.
\end{enumerate}

{In} %MDPI: We removed the \noindent format, please confirm %Hank: confirm
 general, the set $H^p(\mathfrak{n},V)$ classifies ($p+1$)-term extensions of $\mathfrak{n}$ by $V$. Moreover, equivalence classes of such extensions can be equipped with an Abelian group structure such that $H^p(\mathfrak{n},V)$ coincides with it not just as a set, but also as a group.

\begin{remark}\label{deg4class}
Recall a 2-crossed-module $\mathscr{G}$ as defined in Section \ref{sec:gaug2gau}. The exactness $t't=0$ of the complex  \eqref{eq:3comp} states that $\mathscr{G}$ gives rise to a 5-term exact sequence
\begin{equation}
    0\rightarrow V=\operatorname{ker}t'\rightarrow\underbrace{ \mathfrak{i}\xrightarrow{t'}\h\xrightarrow{t}\g}_{=\mathscr{G}}\rightarrow\mathfrak{n}=\operatorname{coker}t\rightarrow 0\nonumber
\end{equation}
of Lie algebras, which by the above statement is classified by a degree-4 Lie algebra cohomology class $H^4(\mathfrak{n},V)$. This class has also appeared as part of the data that classifies {\it crossed-squares} \cite{Brown}; indeed, each crossed-square has an associated 2-crossed-module \cite{Ang2018}.
\end{remark}

We shall show in detail next that, at degree 3, $H^3(\mathfrak{n},V)$ classifies precisely the four-term complex  (\ref{eq:ext}) of a Lie algebra crossed-module.

\subsection{Theorem of Gerstenhaber}
Before constructing the 3-cocycle $c\in Z^3(\mathfrak{n},V)$, we introduce the notion of addition in the set of crossed-modules. Given two crossed-modules $t:\mathfrak{h}\rightarrow\mathfrak{g},t':\mathfrak{h}'\rightarrow\mathfrak{g}'$ with the same kernel $V$ and cokernel $\mathfrak{n}$, it can be shown that
\begin{equation}
    (t\oplus t'): \mathfrak{h}\oplus\mathfrak{h}'/\, \overline\Delta\rightarrow \mathfrak{g}\oplus_\mathfrak{n}\mathfrak{g}'\nonumber
\end{equation}
is another crossed-module, called the { {crossed-module sum}} of $t$ and $t'$. Here, $\overline\Delta$ is the kernel of the addition map $+:V\oplus V\rightarrow V$, while $\mathfrak{g}\oplus_\mathfrak{n}\mathfrak{g}'$ is the fibre pullback; explicitly, $$\overline\Delta = \{(v,-v)\mid v\in V\},\qquad \mathfrak{g}\oplus_\mathfrak{n}\mathfrak{g}' = \{(X,X')\in\mathfrak{g}\oplus\mathfrak{g}'\mid pX=p'X'\}.$$ 
{Note} that as direct sums are commutative, we have $(t\oplus t') \cong (t'\oplus t)$.

This notion descends to elementary equivalence classes of crossed-modules, and endows the set ${\bf XMod}(\mathfrak{n},V)$ with the structure of an Abelian group. We shall show that this Abelian group is isomorphic precisely to $H^3(\mathfrak{n},V)$. To begin, we construct a bilinear skew-symmetric map
\begin{equation}
    f(x_1,x_2) = s_1([x_1,x_2]) - [s_1(x_1),s_1(x_2)],\qquad x_1,x_2 \in \mathfrak{n} \nonumber
\end{equation}
from a section $s_1:\mathfrak{n}\rightarrow\mathfrak{g}$ of the map $p:\mathfrak{g}\rightarrow\operatorname{coker}t=\mathfrak{n}$ in  (\ref{eq:ext}). Though $s_1$ may not be a Lie algebra map, the projection $p$ is, so $pf = 0$ and $f$ is valued in $\operatorname{ker}p$. By the exactness $\operatorname{ker}p = \operatorname{im}t$ of  (\ref{eq:ext}), there exists a bilinear skew-symmetric map $e:\mathfrak{n}^{\wedge 2}\rightarrow \mathfrak{h}$ such that $f=te$.

We now pick another section $s_2:\operatorname{im}t\subset \mathfrak{g}\rightarrow \mathfrak{h}$ of the crossed-module map $t:\mathfrak{h}\rightarrow\mathfrak{g}$, whence $e=s_2f$. Let $\circlearrowright$ denote a summation over cyclic permutations of $x_1,x_2,x_3$, then by construction,
\begin{eqnarray}
    tde(x_1,x_2,x_3) &=& t\left[\circlearrowright e([x_1,x_2],x_3) - \circlearrowright s_1(x_1)\rhd e(x_2,x_3)\right]\nonumber \\
    &=& \circlearrowright f([x_1,x_2],x_3) - \circlearrowright t(s(x_1)\rhd e(x_2,x_3))\qquad \text{Peiffer conditions  (\ref{eq:pfeif})} \nonumber \\
    &=& \circlearrowright f([x_1,x_2],x_3) - \circlearrowright [s_1(x_1),\underbrace{te(x_2,x_3)}_{=f(x_2,x_3)}]\qquad \text{Definition of $f$} \nonumber \\
    &=& \circlearrowright \left([s_1([x_1,x_2]),s_1(x_3)] - s_1([[x_1,x_2],x_3])\right) \nonumber \\
    &\quad& - \circlearrowright\left([s_1(x_1),[s_1(x_2),s_1(x_3)]] - [s_1(x_1),s_1([x_2,x_3])]\right) \qquad \text{Jacobi identity}\nonumber \\
    &=& \circlearrowright\left( [s_1([x_1,x_2]),s_1(x_3)] - [s_1([x_2,x_3]),s(x_1)]\right)\qquad \text{Cyclicity}\nonumber  \\
    &=&0,\nonumber
\end{eqnarray}
as such $de$ is in fact valued in $\operatorname{ker}t$. Again by the exactness of the sequence  (\ref{eq:ext}) we may find a skewsymmetric trilinear map $c:\mathfrak{n}^{\wedge 3}\rightarrow V$ such that $ic = de$, where $i:V\hookrightarrow \mathfrak{h}$ is the inclusion. Picking yet another section $s_3:\mathfrak{h}\rightarrow V$ yields $c=s_3De$.

Now we must show that $dc=0$. It may be tempting to say that, since $ic=de$, we have $idc = dic = d^2e=0$ by the nilpotency $d^2=0$. However, this does not immediately follow, as $s_1$ is not necessarily a section and hence $s_1(\cdot)\rhd$ is not necessarily a well-defined action. By explicit computation, terms involving the problematic operation $s_1(\cdot) \rhd$ in $idc$ read
\begin{eqnarray}
    &\quad& \sum_{i<j}(-1)^{i+j}s_1([x_i,x_j])\rhd e(x_1,\dots,\hat{x}_i,\dots,\hat{x}_j,x_4) \nonumber \\
    &\qquad\quad & - \sum_{i=1}^4 (-1)^is_1(x_i)\rhd \left[\sum_{j\neq i} (-1)^j s_1(x_j)\rhd e(x_1,\dots,\hat{x}_j,\dots,x_3)\right] \qquad \text{Rearrange}\nonumber \\
     &=&\sum_{i<j}(-1)^{i+j}\left(s_1([x_i,x_j]) - [s_1(x_i),s_1(x_j)]\right)\rhd e(x_1,\dots,\hat{x}_i,\dots,\hat{x}_j,x_4)\nonumber\\
     &\qquad &\qquad\qquad\qquad\qquad\qquad\qquad\qquad\qquad\qquad\qquad\qquad\qquad\text{Definition of $f$}\nonumber \\
     &=&\sum_{i<j}(-1)^{i+j}\underbrace{f(x_i,x_j)}_{=te(x_i,x_j)}\rhd e(x_1,\dots,\hat{x}_i,\dots,\hat{x}_j,x_4) \qquad \text{Peiffer condition}\nonumber \\
     &=& \sum_{i<j}(-1)^{i+j}[e(x_i,x_j),e(x_1,\dots,\hat{x}_i,\dots,\hat{x}_j,x_4)] \qquad \text{Cyclicity of summation}\nonumber \\
     &=&0, \nonumber
\end{eqnarray}
hence we nevertheless have $dc=0$. This allows us to conclude that $c\in Z^3(\mathfrak{n},V)$.

We now wish to show that changing the choices of the sections $s_{1,2,3}$ adds to $c$ a 3-coboundary. By linearity, we can write $s_1' = s_1 + \delta$ for some map $\delta:\mathfrak{n}\rightarrow \mathfrak{g}$. Defining a bilinear skew-symmetric map $f'$ analogously, we see that
\begin{equation}
    f'(x_1,x_2) = f(x_1,x_2) + [s_1(x_1),\delta(x_2)] + [\delta(x_1),s_1(x_2)] + [\delta(x_1),\delta(x_2)] - \delta([x_1,x_2]).\nonumber
\end{equation}
{Notice} the terms $[s_1(x_1),\delta(x_2)] + [\delta(x_1),s_1(x_2)] - \delta([x_1,x_2])$ constitute precisely the coboundary $d\delta(x_1,x_2)$ of a cochain $\delta:\mathfrak{n}\rightarrow\mathfrak{g}$, with $x_1,x_2\in\mathfrak{n}$ lifted up to $\mathfrak{g}$ by the map $s_1$.

Now as $f',f$ are valued in $\operatorname{ker}p=\operatorname{im}t$, we can find $\mathfrak{h}$-valued bilinear maps $\epsilon,\varepsilon$ such that $t\epsilon(x_1,x_2) = d\delta(x_1,x_2)$ and $t\varepsilon(x_1,x_2) = [\delta(x_1),\delta(x_2)]$. Further, we can also find a $\operatorname{ker}t=\operatorname{im}i$-valued bilinear map $\varphi$ such that 
\begin{equation}
    e'(x_1,x_2) = e(x_1,x_2) + \epsilon(x_1,x_2) + \varepsilon(x_1,x_2) + i\varphi(x_1,x_2)\nonumber
\end{equation}
when lifted by $s_2$. Our goal now is to apply the differential $d$; however, the trouble here is that $d$ and $s_2$ need not commute, as $s_2$ is not in general a section. Now by computation
\begin{eqnarray}
    tds_2\delta(x_1,x_2) &=& t(s_1(x_1)\rhd s_2\delta(x_2) + s_1(x_2)\rhd s_2\delta(x_1) - s_2\delta([x_1,x_2])),\nonumber\\
    &\qquad& \qquad\qquad\qquad\qquad\qquad\qquad\qquad\qquad\qquad \text{Peiffer condition}\nonumber \\
    &=& ts_2([s_1(x_1),\delta(x_2)]-[s_1(x_2),\delta(x_1)] - \delta([x_1,x_2]))\nonumber\\
    &=& ts_2d\delta(x_1,x_2),\nonumber
\end{eqnarray}
so $\Delta_1=ds_2\epsilon - s_2d\epsilon$ is valued in $\operatorname{ker}t$. Similarly, the difference $\Delta_2=ds_2\varepsilon - s_2d\varepsilon$ also lies in $\operatorname{ker}t$, which allows us to finally write 
\begin{equation}
    c'(x_1,x_2,x_3) = c(x_1,x_2,x_3) + d\epsilon(x_1,x_2,x_3) + d\varepsilon(x_1,x_2,x_3) + i\left(\Delta_1 + \Delta_2\right)(x_1,x_2,x_3) + di\varphi(x_1,x_2,x_3).\nonumber
\end{equation}
{Using} the injectivity of $i$, we have $di\varphi = i(d|_V \varphi)$, hence defining $\sigma=\epsilon+\varepsilon$ and $\gamma=\Delta_1+\Delta_2 +d|_V\varphi$ yields
\begin{equation}
    c' = c + d\sigma + i\gamma = c+d\sigma \mod\operatorname{ker}t,\nonumber
\end{equation}
whence lifting by $s_3$ up to $V$ yields $c'=c+d\sigma$. This shows that the cohomology class of $c$ does not depend on the choice of the section $s_1$.

Now suppose we have distinct sections $s_2,s_2'$, defining $e=s_2f$ and $e'=s_2'f$. It is clear that $t(e-e') = ts_2f-ts_2'f=f-f=0$, hence $e-e'$ is valued in $\operatorname{ker}t=\operatorname{im}i$. This means that $s_3$ lifts $d(e-e')$ to a coboundary $d\omega$ such that $c'=c+d\omega$, demonstrating that the cohomology class of $c$ does not depend on the choice of the section $s_2$ as well. Lastly, any two sections $s_3,s_3'$ must coincide, at least on the image $\operatorname{im}i=\operatorname{ker}t$, hence the { {cocycle itself}} $c$ does not depend on the choice of $s_3$.

\begin{lemma}
Let $t,t'$ denote two elementary equivalent crossed-modules, then the 3-cocycles $c,c'$ they define coincide $[c]=[c']\in H^3(\mathfrak{n},V)$ in cohomology.
\end{lemma}
\begin{proof}
First, pick sections $s_{1,2,3},s_{1,2,3}'$ in the respective crossed-modules  $t,t'$ and construct the 3-cocycles $c,c'\in C^3(\mathfrak{n},V)$. Suppose an elementary equivalence $(\phi,\psi)$ between the two crossed-modules exists, then $\psi s_1$ is a section of $p'$. The above shows that the 3-cocycle $\tilde c'$ constructed from the sections $(\psi s_1,s_2',s_3')$ differ from that $c'$ constructed from $(s_1',s_2',s_3')$ only by a coboundary. Our task is thus to show that $\tilde c'$ also coincides with $c$ up to coboundary.

Toward this, we define $s_2'\psi f\equiv \tilde e'$ and compare this to $\phi e = \phi s_2 f$. First, we know that $t's_2'=1$, hence $\tilde e' - \phi e$ is valued in $\operatorname{ker}t'=\operatorname{im}i'$, so we can find a map $v:\mathfrak{n}^{\wedge 2}\rightarrow V$ such that $\tilde e'-\phi e = i'v$. 

We now take the differential $d$ of this equation. By definition of the elementary equivalence, we can rewrite contributions $\psi(x_i)\rhd \phi(e) = \phi(x_i\rhd e)$ in the differential, as such $d(\phi e)= \phi de$. Now $s_3\phi$ is a section of $i'$, hence 
\begin{equation}
    \tilde c' - c = s_3D \tilde e' - (s_3 \phi)de = dv\nonumber
\end{equation}
is a coboundary. This proves the lemma.
\end{proof}

The lemma allows us to put a well-defined map $b:{\bf XMod}(\mathfrak{n},V)\rightarrow H^3(\mathfrak{n},V)$. 
\begin{theorem}[Gerstenhaber, attr. by MacLane]\label{thm:gers}
 $b$ is an isomorphism of Abelian groups.
\end{theorem}

{The} %MDPI: We removed the \noindent format, please confirm %Hank: confirm
 classifying data of a Lie algebra crossed-module $t:\mathfrak{h}\rightarrow\mathfrak{g}$ is exactly $(\mathfrak{n},V,c)$ with $c\in H^3(\mathfrak{n},V)$.

\subsection{The Postnikov Class}\label{postnikov}
Let us now turn to the reason why we called an element in $H^3(\mathfrak{n},V)$ a ``Postnikov class'' in the main text. Formally, a Lie 2-algebra integrates to a Lie 2-group $t:H\rightarrow G$~\cite{Bai_2013,Baez:2005sn}, for which a ``Gerstenhaber theorem'' also holds: the crossed-module $t: H\rightarrow G$ is classified by its {\it {Ho{\'a}ng data}} $(N,V,\kappa)$ \cite{Ang2018,Brown,baez2004}, named after the Vietnamese mathematician Ho{\'a}ng Xu{\^ a}n S{\' i}nh \cite{Baez2023HoangXS}, where $N=\operatorname{coker}t$, $V=\operatorname{ker}t$ $\kappa\in H^3(N,V)$ is a group cohomology class (as opposed to a Lie algebra cohomology class).

The name ``Postnikov class'' comes from topology. Given any ``nice'' space $X$ (a finite CW complex), its fundamental group $\pi_1(X)$ in general acts on higher homotopy groups $\pi_{\geq 2}(X)$ via {\it {monodromy}}. The {\it {homotopy 2-type}} $\Pi_2(X)=(\pi_1(X),\pi_2(X),\operatorname{Ptn}(X))$ is modeled by the group crossed-module \cite{Brown}
\begin{equation}
    1\rightarrow \operatorname{ker}\partial = \pi_2(X) \rightarrow \pi_2(X,Y) \xrightarrow{\partial} \pi_1(Y)\rightarrow \pi_1(X) =\operatorname{coker}\partial\rightarrow 1,\nonumber 
\end{equation}
where $Y\subset X$ is a closed subspace and $\partial$ is the natural boundary map. Up to homotopy, it is classified by the {\it {Postnikov class}} $\operatorname{Ptn}(X)\in H^3(\pi_1(X),\pi_2(X))$, which determines how 2-cells are glued upon the 1-cells.

It is possible to construct the classifying space $B(N,V)$ satisfying the condition $\Pi_2 B(N,V) = (N,V,\kappa)$ \cite{Kapustin:2013uxa,BaezRogers}. Such a space sits in the {\it {Postnikov tower}} fibration sequence
\begin{equation}
    B^2V \rightarrow B(N,V) \rightarrow BN,\nonumber
\end{equation}
where $BN=K(N,1)$ is the classifying Eilenberg--MacLane space of $N$ and $B^2V=K(V,2)$ is the second delooping of $V$, satisfying $\pi_2(B^2V)=V$ with other homotopy groups vanishing. 

In other words, the Postnikov class determines how $B(N,V)$ is constructed from the base $BN$ by gluing the second delooping space $B^2V$. The homotopy classification theorem states that gauge-equivalent discrete flat 2-connections $H^1(X,(N,V))$ are isomorphic to homotopy classes of classifying maps $X\rightarrow B(N,V)$ \cite{Ang2018,BaezRogers}; this is how 2-gauge topological field theories are constructed \cite{Kapustin:2013uxa,Zhu:2019}.

\section{2-Bundle Homomorphisms}\label{weakpost}
In this appendix, we show that an elementary equivalence gives rise to a homomorphism between 2-gauge bundles. We also generalize this perspective to the weak case. 

Let $\mathcal{P},\mathcal{P}'\rightarrow X$ denote two 2-gauge bundles on $X$, equipped with connections $(A,\Sigma)$ and $(A',\Sigma')$, respectively. Intuitively, from the gauge theory perspective, a {\it {2-bundle homomorphism}} $g:\mathcal{P}\rightarrow\mathcal{P}'$ should satisfy two properties: (1) it is a bundle map over $X$; namely the triangle 
\begin{equation}
    \begin{tikzcd}
\mathcal{P} \arrow[rr, "g"] \arrow[rd] &   & \mathcal{P}' \arrow[ld] \\
                                          & X &                        
\end{tikzcd}\nonumber
\end{equation}
commutes, and (2) preserves all gauge-invariant data.

From our computations in the main text, the gauge-invariant data consist precisely of the fake-flatness $\mathcal{F}$  (\ref{fakeflat}) and the 2-curvature $\mathcal{G}=K$. As such homomorphisms $\psi$ must~satisfy
\begin{equation}
    \mathcal{F} = g^*\mathcal{F}',\qquad \mathcal{G} = g^*\mathcal{G}'.\nonumber
\end{equation}
{Let} us write, locally, $g^* = f^*\otimes \Psi$ in terms of components, where $f^*$ is the pullback of $f:X\rightarrow X$ on forms and $\Psi=(\phi,\psi)$ is a map on the Lie algebras
\begin{equation}
    \phi:\mathfrak{h}'\rightarrow \mathfrak{h},\qquad \psi:\mathfrak{g}'\rightarrow\mathfrak{g}.\nonumber
\end{equation}
{The} fake-flatness condition $\mathcal{F}=\psi^*\mathcal{F}'$ implies
\begin{equation}
    F = (f^* \otimes \psi)F',\qquad t\Sigma= (f^*\otimes \psi) t'\Sigma'= t(f^*\otimes \phi)\Sigma';\label{eq:gaugemap}
\end{equation}
by linearity and $F=d_AA,F'=d_{A'}A'$, the first condition in  (\ref{eq:gaugemap}) means that $f^*$ commutes with the de Rham differential $d$, and that $\psi$ is a Lie algebra homomorphism {(This means that $A = \psi A'$ and $[A\wedge A] = \psi[A'\wedge A'] = [\psi A'\wedge \psi A']$.)}. The second condition means $t\phi = \psi t'$ commutes with the crossed-module maps $t,t'$.

\subsection{Equivalences of 2-Gauge Bundles}
The 2-curvature condition reads
\begin{equation}
    \mathcal{G}=d_A\Sigma = (f^*\otimes \phi)d_{A'}\Sigma' = (f^*\otimes \phi)(d\Sigma' + A\wedge^{\rhd'} \Sigma'),\nonumber
\end{equation}
where $\rhd'$ is the crossed-module action in $\mathcal{P}'$. Using the second condition from  (\ref{eq:gaugemap}), the first term reads
\begin{equation}
    (f^*\otimes \phi)d\Sigma' = d\Sigma = d(f^*\otimes \phi)\Sigma',\nonumber
\end{equation}
while the second term reads
\begin{equation}
    A\wedge^\rhd \Sigma = (f^*\otimes \phi)A'\wedge^{\rhd'}\Sigma'.\nonumber  
\end{equation}
{However}, the condition $A = (f^*\otimes \psi)A'$ means that we must have
\begin{equation}
    (f^*\otimes \phi)A'\wedge^{\rhd'}\Sigma' = ((f^*\otimes \psi)A')\wedge^{\rhd} (f^*\otimes \phi)\Sigma'.\nonumber
\end{equation}
{This} tells us that, not only does $g_{-1}$ also has to be a Lie algebra homomorphism, but also the condition
\begin{equation}
    \phi(X\rhd' Y) = (\psi X)\rhd (\phi Y),\qquad \forall~ X\in\mathfrak{g}',Y\in\mathfrak{h}'.\label{eq:equiv}
\end{equation}
{This} is precisely the definition of an { {elementary equivalence}} of Lie algebra crossed-modules~\cite{Wag:2006,Baez:2005sn}.

As such, we may interpret elementary equivalence as an equivalence of the gauge-invariant data on the 2-gauge bundles $\mathcal{P},\mathcal{P}'$. The { {Gerstenhaber Theorem} \ref{thm:gers}} then implies
\begin{corollary}
If the 2-gauge bundles $\mathcal{P},\mathcal{P}'$ exhibit distinct Postnikov classes $\kappa\neq \kappa' \in H^3(\mathfrak{n},V)$ as 2-curvature anomalies, then there does {\it not} exist a 2-bundle homomorphism between them.
\end{corollary}

\subsection{Weak Lie 2-Algebra Homomorphisms}
Let $\mathfrak{g}$ be a simple Lie algebra. We wish to describe the loop model $\mathfrak{l}_k$ for the string Lie 2-algebra $\mathfrak{string}_k(\mathfrak{g})$ at level $k\in\bbZ$ described in the main text. It is a Lie algebra crossed-module, but with a non-trivial Postnikov class $[\kappa]\in H^3(\mathfrak{g},\mathbb{R})$ given by the generating fundamental 3-cocycle $\omega=\langle \cdot,[\cdot,\cdot]\rangle$ on $\mathfrak{g}$ \cite{book-loop}.

To describe this Lie 2-algebra $\mathfrak{l}_k$, we first need to extend our notion of Lie 2-algebra homomorphisms to the weak case. This is accomplished by appending an additional component $\varphi$ to the notion of an elementary equivalence between Lie algebra crossed-modules. This component $\varphi$ is used to control the failure of $(\phi,\psi)$ from being a strict elementary equivalence, as well as the Jacobiators $\mu,\mu'$. 

The full definition is \cite{Baez:2005sn} the following.
\begin{definition}\label{def:weakequiv}
A {\it 2-homomorphism}  $\Psi=(\varphi,\phi,\psi)$ between weak Lie 2-algebras $t:\mathfrak{h}\rightarrow\mathfrak{g},t':\mathfrak{h}'\rightarrow\mathfrak{g}'$---with respective Jacobiators $\mu,\mu'$---are given by the set of chain maps
\begin{equation}
    \varphi: \mathfrak{g}\times\mathfrak{g}\rightarrow\mathfrak{h}',\qquad \phi:\mathfrak{h} \rightarrow\mathfrak{h}',\qquad \psi:\mathfrak{g}\rightarrow\mathfrak{g}',\nonumber
\end{equation}
such that $t'\phi = \psi t$ and the following conditions are satisfied:\vspace{-12pt}
\begin{eqnarray}
    t'\varphi(X,X') &=& \psi([X,X']) - [\psi X,\psi X']',\nonumber \\
    \varphi(X,tY) &=& \phi(X\rhd Y) - (\psi X)\rhd' (\phi Y),\nonumber \\
    \mu'(\psi(X),\psi(X'),\psi(X'')) - \phi(\mu(X,X',X'')) &=& \circlearrowright\varphi(X,[X',X'']) + \circlearrowright [(\psi X)\rhd'\varphi(X',X'')]\label{weak2hom}
\end{eqnarray}
for each $X,X',X''\in\mathfrak{g},Y\in\mathfrak{h}$ and $\circlearrowright$ denotes a summation over cyclic permutations of the arguments.
\end{definition}

{In more} mathematically sophisticated terms, $\Psi=(\varphi,\phi,\psi)$ is an invertible chain homotopy between two-term $L_\infty$-algebras \cite{Baez:2005sn,Roytenberg:2007}. In this way, it can be understood that $\varphi$ can only appear between {{weak}} Lie 2-algebras, and not strict Lie 2-algebras.

{Definition \ref{def:weakequiv}} gives us a weaker notion of elementary equivalence: that two weak Lie 2-algebras are { {weakly equivalent}} if there exist an invertible 2-homomorphism $\Psi=(\varphi,\phi,\psi)$ between them whose kernel and cokernel are (strictly) elementary equivalent to the trivial Lie 2-algebra. This was the notion of equivalence that was used in \cite{Baez:2005sn}.

\section{Framed Submanifolds; the Fermionic Quasistring Order}\label{framedquasi}
The order  \eqref{5dtopord}, $\mathcal{C}_\text{5d}$, hosts on $Y$ a (closed) ``magnetic'' quasistring described by $\Sigma$, and an ``electric'' dual quasiparticle %(the electron)
described by $c$; we denote by $l^2,l$ their worldvolumes, respectively. To understand what this means geometrically, we recall that a {{framing}} is equivalent to a trivialization of the normal bundle \cite{book-algtop}, and that the Stiefel--Whitney classes $w_n$ keep track of the twists in the framing of $(n-1)$-dimensional embedded submanifolds~\cite{Thorngren2015}; see Figure \ref{fig:twist}.

% A {\bf $w_n$-structure} assigns $\pm 1$ to a $(n-1)$-manifold. The sign is negative when the framing is twisted. A {\bf twist} is a full self-braiding induced by a $2\pi$-rotation in the framing; .

\begin{figure}[h]
%\centering
\includegraphics[width=0.9\columnwidth]{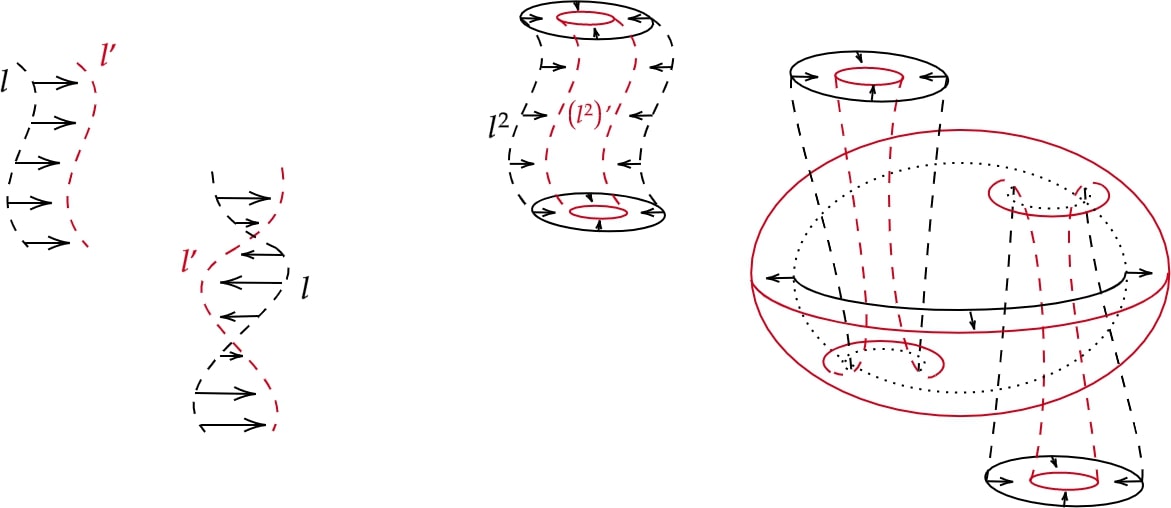}
\caption{Twists in the framing structure: the ``push-offs'' $l',(l^2)'$ of the worldvolumes $l,l^2$ along their framings are shown in red. {\bf Left}: The quasiparticle described by $c$, and a twist of its 1-dimensional worldline $l$. {\bf Right}: The quasistring described by $\Sigma$, and a twist in its 2-dimensional worldsheet $l^2$.}
\label{fig:twist}
\end{figure}

{$\Sigma,c$ are} %MDPI: We removed the \noindent format, please confirm %Hank: confirm
 fields associated to $w_{3},w_2$, hence they ``detect'', respectively, via their values $\pm1$ on $1,2$-dimensional submanifolds $l,l^2\subset X$, twists in the framings of $l,l^2$. Conversely, those submanifolds that exhibit these twistings are interpreted as the worldvolumes of the quasi-particle/string.

\begin{remark}
Many topological orders, such as the 5D quasistring order $\mathcal{C}_\text{5d}$ \eqref{5dtopord} here, are \textit{bordism invariants} \cite{JuvenWang,Witten:2019,Thorngren2015,Freed:2014}. This means, in particular, that they all vanish $\mathcal{C}_\text{5d}=0$ on bounding 5-manifolds. Bordism invariants are elements of the {\it (framed) bordism group} $\Omega^O_\ast$ \cite{Guo_2020,Freed_2021}. They constitute the {\it non-perturbative} part of the anomalies that appear in QFTs \cite{Freed:2014,Witten:2019}.
\end{remark}

% \paragraph{\textbf{The $w_2w_3$ gravitational anomaly.}} 
One of the most interesting properties of the order $\mathcal{C}_\text{5d}$ given in \eqref{5dtopord} is that it detects a {gravitational anomaly} {(Namely, an anomaly under diffeomorphism,)} { {aside}} {from the chiral anomaly due to time-reversal symmetry $\bbZ_2^T$.)} \cite{JuvenWang,Thorngren2015}: taking $Y=\mathbb{C}P^2$ and $X_\varphi = Y\times_\varphi S^1$ 
%as the mapping torus of complex conjugation 
 and the diffeomorphism $\varphi:(z_1,z_2)\mapsto (\bar{z}_1,\bar z_2)$, where the $z_i$ are coordinates on $\mathbb{C}P^2$, one has $\mathcal{C}_\text{5d}=1 \mod 2$, which evaluates to a non-trivial anomaly
\begin{equation}
    \mathcal{Z}(Y)=(-1)^{\mathcal{C}_\text{5d}} = \exp \left(i\pi \int_{X_\varphi} w_2\cup w_3\right) = -1\label{gravanom}
\end{equation}
associated to the diffeomorphism $\varphi$. 

In fact, this mapping torus $X_\varphi=\mathbb{C}P^2\times_\varphi S^1$ generates the framed bordism group $\Omega^O_5$; in other words, any other 5-dimensional cobordism $X$ that evaluates to $-1$ in \eqref{gravanom} is cobordant to the mapping torus $X_\varphi$ of $Y=\mathbb{C}P^2$.

The partition function $\mathcal{Z}$ in  \eqref{gravanom} defines an invertible { {fermionic}} topological quantum field theory (TQFT) $\mathcal{Z}$ \cite{Freed:2014,Witten:2019,Freed_2021} given by the order $\mathcal{C}_\text{5d}$. In general, whether a TQFT is bosonic or fermionic is determined by the {self-braiding statistics} of its defects \cite{XieChen:2013,Zhu:2019}.

% \begin{remark}
% The above means that a half self-braiding (i.e., a {\it reversal}/$\pi$-rotation of the framing) is sent by the invertible TQFT $\mathcal{Z}$ to the generator $(-1)^F$ of fermion parity $\bbZ_2^F$. This is {\bf topological spin-statistics}, which is a property shared by all invertible reflection-positive fermionic TQFTs \cite{Freed_2021}.
% \end{remark}

In dimensions $\geq 3$, point-like defects can be braided such that their worldlines $l,l'$ are { {linked}}, as shown in Figure \ref{fig:twist} {(Left).} %MDPI: We removed the bold. Please confirm this revision. The `(Right)' was the same %Hank: confirm
 This procedure is encoded by a {linking number} $\operatorname{lk}(l,l')$, which changes by 1 upon a twist \cite{Guo_2020}. In dimensions $\geq 4$, one can braid { {worldsheets}} $l^2,(l^2)'$ with each other, as shown in Figure \ref{fig:twist} ({Right}). We also have a corresponding { {surface}}-linking number $\operatorname{lk}(l^2,(l^2)')$, which also changes by 1 upon a twist \cite{Thorngren2015}. 

The spin-TQFT $\mathcal{Z}$ exhibiting the gravitational anomaly given in  \eqref{gravanom} defines a topological order with {{fermionic}} quasiparticle and quasistring excitations \cite{JuvenWang,Thorngren2015}. What this means is that the Wilson loop and surface operators corresponding to these quasiparticles and quasistrings are accompanied by the following phases,
\begin{equation}
    (-1)^{\operatorname{lk}(l,l')} \qquad  (-1)^{\operatorname{lk}(l^2,(l^2)')}, \label{selfstat}
\end{equation}
in the quantum theory; if these phases are present, then the exictations are {\it bosonic}.

\newpage

\printbibliography

\end{document}